\documentclass[12pt]{article}

\usepackage[margin=1in]{geometry}
\usepackage{amsmath,amssymb}
\usepackage{booktabs}
\usepackage{multirow}
\usepackage{graphicx}
\usepackage{float}
\usepackage{placeins}
\usepackage{enumerate}
\usepackage{setspace}
\usepackage{natbib}
\usepackage{url}
\usepackage{xcolor}
\usepackage[hidelinks]{hyperref}

\newtheorem{theorem}{Theorem}
\newtheorem{prop}{Proposition}

\newtheorem{assumption}{Assumption}

\newtheorem{remark}{Remark}
\newenvironment{proof}{\medskip\noindent\textbf{Proof: }}{\hfill$\square$}
\newcommand{\E}{\mathbb{E}}

\renewcommand{\arraystretch}{1.0}

\title{Pseudo-Panel Event Studies:\\ Econometric Foundations for Credible Applications}
\author{Dongwoo Kim\thanks{Department of Economics, Simon Fraser University and Korea University. I thank Michele Battisti for introducing me to this topic. I also thank Louise Laage for helpful comments and discussions that helped shape the ideas in this paper. I gratefully acknowledge support from the Social Sciences and Humanities Research Council of Canada under the Insight Grant (435-2024-0322). All errors are my own.}\\ Simon Fraser University \& Korea University}
\date{October 12, 2026}

\begin{document}
\maketitle

\begin{abstract}
Event studies with panel data are widely used to estimate treatment effects. With repeated cross-sections, future treatment dates are unobserved for untreated individuals. Pseudo-panel event studies impute pre-treatment outcomes of the treated using matched untreated individuals. This paper establishes identification in such designs and characterizes bias when identification conditions fail. I show that the pseudo-panel regression is equivalent to weighted least squares with group means. This representation simplifies computation and yields analytical inference. In simulations, analytical confidence intervals attain near-nominal coverage, whereas conventional intervals substantially undercover. I also develop a doubly robust estimator using richer covariates, a symmetric matching design with synthetic type fixed effects, and a fully nonparametric estimator that accommodates heterogeneous treatment effects. In a Korean panel survey, the pseudo-panel regression yields an employment child penalty of 42.1\%, close to 46.9\% in the actual panel. Applications to repeated cross-sections yield employment penalties of about one-third.
\end{abstract}


\newpage

\section{Introduction}

The labor-market impact of having a child is typically estimated by comparing outcomes before and after childbirth. For instance, using Danish administrative data and an event study, \citet*{kleven2019children} estimate sharp and persistent ``child penalties'' in earnings, employment, and hours worked. \citet{kleven2019cross} apply the same approach and find similar patterns across six countries. However, panel data that follow individuals long enough to estimate child penalties are rare. In many settings, researchers have access only to repeated cross-sections, where different individuals are sampled in each wave. In such data, a childless respondent's future birth date (and hence their negative event time) is never observed. To overcome this limitation, \cite{kleven2025child} and \citet*{kleven2025atlas} develop a pseudo-panel event study design that imputes pre-birth outcomes of parents from observed outcomes of childless individuals matched on age, survey year, and discrete characteristics, under the assumption that the childless provide a valid comparison for parents. The design is attractive because it makes event studies possible with repeated cross-sections, which are widely available and typically have much larger samples than panels.\footnote{The approach builds on the classic pseudo-panel method of \cite{deaton1985panel}, who proposed tracking cohort means through repeated cross-sectional surveys. \cite{verbeek1992can}, \cite{moffitt1993identification}, and \cite{mckenzie2004asymptotic} develop identification and asymptotic theory for cohort-based pseudo-panels.}

This paper develops an econometric framework for pseudo-panel event studies. I establish identification in such designs under conditions that ensure matched childless individuals serve as a valid control group for parents, and characterize the bias when these conditions fail. I then show that the pseudo-panel regression is equivalent to a weighted least squares (WLS) regression with group means, and use this representation to derive its analytical variance. I also extend the design in three directions: (i) a doubly robust estimator that admits richer covariates in pre-birth outcome imputation; (ii) symmetric matching with synthetic type fixed effects (FE); and (iii) a fully nonparametric estimator that accommodates heterogeneous treatment effects. 

The design has become popular in applied research. \cite{zhou2022rising} estimate the rising child penalty in China and \cite{huang2025china} investigate the association between the penalty and declining fertility. \cite{marchionni2025last} find persistent earnings losses after childbirth in Latin America and \cite{nguyen2026vietnam} in Vietnam. \cite{jessen2025child} apply the approach to a single cross-section and document larger cognitive skill losses for women after childbirth. Beyond childbirth, \cite{wu2025marriage} finds that both men and women in China become more supportive of traditional gender roles after marriage. In policy research, \cite{bussolo2024marriage} estimate marriage and child penalties in South Asia, and \cite{baez2026motherhood} document a lasting decline in mothers' labor force participation after childbirth in Turkey. More generally, the framework applies to cross-sectional data in which treatment dates are observed only for treated individuals.

A pseudo-panel is constructed by matching parents observed at birth to childless individuals (donors). The donor pool is a mixture of future parents and people who will never have a child. If these groups differ in unobserved determinants of the outcome, the imputed pre-birth outcome can be biased. I establish identification of the average treatment effect on the treated (ATT) in a potential-outcomes framework under no anticipation, parallel trends, and comparability between parents and childless donors conditional on the matching characteristics, and derive the bias when that comparability condition fails. This bias cancels in the gender difference only if it is the same for women and men. I also clarify the additional restrictions on untreated outcomes and treatment-effect heterogeneity needed to interpret the resulting regression coefficients as causal effects.

I then show that the ordinary least squares (OLS) regression using the pseudo-panel yields the same coefficients as WLS using mean outcomes within cells defined by discrete matching variables. I use this equivalence to derive the estimator's probability limit and asymptotic variance. A donor can supply several imputed outcomes, and an individual can appear in different years as both donor and parent. Conventional inference using heteroskedasticity-robust standard errors ignores this dependence because it treats the imputed pre-birth outcomes as independent observations. The analytical variance I derive accounts for this dependence, along with sampling variation in cell weights. Under regularity conditions, the approach applies to the baseline estimator and all three extensions. Using cell means also reduces regression time in large samples.

The original design matches on discrete variables, so continuous covariates must be discretized. Matching on many variables or variables with many categories can produce sparse or empty cells. I use a doubly robust (DR) estimator to impute the pre-birth cell means that enter the pseudo-panel regression. It combines an outcome model estimated on the donor sample with inverse-probability weights from a propensity score model for group membership. When parents and donors are comparable conditional on the richer covariates and overlap holds, the estimator consistently estimates these cell means if either model is correctly specified. With both models correctly specified, the estimator is semiparametrically efficient for the imputed mean when the parent and donor samples are independent.

I also consider an alternative symmetric pseudo-panel construction that matches parents observed at birth both backward to donors and forward to parents of the same observed type, defined by gender, first-birth year, age at first birth and matching characteristics. In the baseline design, post-birth outcomes come from whichever parents appear at each event time, so the composition of parents can change over event time. This symmetric matching extension, developed from the idea in \cite{kleven2025child}, weights each type at each event time by the number of parents within the type observed at birth. I include synthetic type fixed effects to absorb persistent differences in untreated outcome levels and state the additional identifying restriction they require.

For the last extension, I consider a fully nonparametric estimator in the spirit of \citet{callaway2021difference}. Under conditional parallel trends and the same comparability condition, I estimate the ATT by comparing parents and childless individuals observed in the same year, at the same age and with the same matching characteristics. The estimator needs no pre-birth donor observations and aggregates type-specific effects with positive weights, thereby accommodating heterogeneous treatment effects.

I conduct Monte Carlo simulations in a child-penalty setting to examine the performance of the proposed estimators. When selection into parenthood differs between women and men, the DR estimator reduces root mean squared error using additional covariates in the pre-birth imputation. Synthetic type fixed effects remove bias caused by persistent differences in untreated outcome levels across types. Restricting the sample to a balanced panel can increase selection bias. When treatment effects differ across cohorts, the pseudo-panel regressions remain biased even with type fixed effects, whereas the nonparametric estimator recovers the average effect. Conventional 95\% confidence intervals for the pseudo-panel regression cover the true penalty in only 47--51\% of replications. Analytical inference and the cluster bootstrap show coverage close to the nominal level. With the matching cells held fixed, the individual-level pseudo-panel regression takes 1.8 times as long as the cell-mean regression with 10,000 people, and 25.5 times as long with 200,000, demonstrating computational improvements delivered by the cell-mean approach.

The proposed estimators are then applied to three South Korean datasets. Korea has the largest gender wage gap among OECD countries \citep{oecd2026korea}, and offers both a long-running household panel and several repeated cross-sectional surveys. The Korean Labor and Income Panel Study (KLIPS) provides the panel validation. Over years 0--10 after birth, the standard pseudo-panel regression, matched on age, year, marital status and education, yields an employment penalty of 42.1\%, close to the 46.9\% from the actual-panel regression. For annual earnings, the pseudo-panel penalty (64.7\%) is close to the panel penalty (65.6\%). The DR estimator and symmetric matching with synthetic type fixed effects give employment penalties of 40\% and 41\%, respectively, and earnings penalties of 64\% and 65\%. Across matching specifications, the nonparametric estimator gives employment penalties of 35--38\% and substantially lower earnings penalties of 42--47\%. The method is then applied to two repeated cross-sections, the Household Income and Expenditure Survey (HIES) and the Household Finance and Living Conditions Survey (HFLC). The pseudo-panel regression gives employment penalties of 33.5\% in HIES and 32.0\% in HFLC. 

Recent work clarifies several methodological issues of this design. \cite{worldbank2024nonparametric} and \cite{balcazar2025nonparametric} match mothers to fathers based on their observed characteristics at childbirth and reweight observations at each event time to match the covariate distribution of the matched mothers. \cite{balcazar2025nonparametric} use this approach to address negative regression weights. \cite{bussolo2024marriage} use matching in both directions, as in the symmetric matching design considered here, but without synthetic type fixed effects.\footnote{Their post-marriage matches come from earlier marriage cohorts observed in the same survey year. My post-birth matches retain the same first-birth cohort across survey years.} They also use cohort means weighted by cohort size and note that this gives the same regression coefficients as using individual observations, without a formal proof or asymptotic analysis. \cite{thakral2026child} derive identifying restrictions for child-penalty regressions with age-at-first-birth or individual fixed effects. None of these papers provides a unified framework for identification, estimation, and inference in pseudo-panel event studies and their extensions.

The remainder of the paper is organized as follows. Section~\ref{sec:framework} sets up the framework, establishes identification and characterizes selection bias. Section~\ref{sec:estimation} develops the cell-level WLS equivalence and analytical inference. Section~\ref{sec:extensions} develops extensions. Section~\ref{sec:simulation} reports the Monte Carlo evidence. Section~\ref{sec:empirical} applies the proposed estimators to Korean panel and repeated cross-sectional data. Section~\ref{sec:conclusion} concludes. Technical proofs, additional results for symmetric matching, additional simulations and further empirical results are in the Supplementary Appendix.

\section{Framework and Identification}\label{sec:framework}

I adopt the potential outcomes framework of the staggered treatment literature, closely following the notation in \citet{kleven2025child}. Let $F_i$ denote the calendar year of individual $i$'s first childbirth, with the convention that $F_i = \infty$ if $i$ never has children. Outcomes are observed in calendar years $t = 1, \ldots, T$. The potential outcome of $i$ at time $t$ when first birth occurs in year $f$ is $Y_{it}(f)$ and the untreated potential outcome is $Y_{it}(\infty)$. The realized outcome is $Y_{it} = Y_{it}(F_i)$. Individual $i$ has gender $g \in \{w, m\}$ (women, men) and a vector of discrete pre-treatment characteristics $X_i$, including year of birth.\footnote{The matching variables $X_i$ are considered fixed: their values must remain the same in the birth year and in the earlier years from which pre-birth outcomes are imputed. Education or marital status can enter this vector only if their measurements satisfy that requirement.} The age $a_{it}$ at time $t$ is determined by $(X_i, t)$. For parents, the age at first birth is $\tilde{a}_i = a_{i,F_i}$, and the event time is $\tau_{it} = t - F_i$.

The parameter of interest is the average treatment effect on the treated:
\begin{equation}\label{eq:ATT}
\text{ATT}(f,t) = \E[Y_t(f) - Y_t(\infty) \mid F = f], \quad \text{for } 2 \le f \le t \le T.
\end{equation}
I suppress the gender superscript and the individual subscript when the context is clear. In practice, the estimator is applied separately by gender, and the child penalty is defined as $\text{ATT}^m(f,t) - \text{ATT}^w(f,t)$.

For parents observed at childbirth, the pseudo-panel imputes pre-birth outcomes using the mean outcomes of childless donors observed in the corresponding earlier year with the same characteristics. I establish identification of the ATT under three assumptions governing untreated outcome trends, anticipation and donor comparability. I assume that the observed cross-sections identify the relevant conditional means and covariate distributions of a common target population. All conditional expectations used below are well defined on the relevant support.

\begin{assumption}[Parallel Trends, PT]\label{assu:PT}
For all dates $t \neq t'$ and all cohorts $f \neq f'$, including $f=\infty$:
\begin{equation}\label{eq:PT}
\E[Y_t(\infty) - Y_{t'}(\infty) \mid F = f] = \E[Y_t(\infty) - Y_{t'}(\infty) \mid F = f'].
\end{equation}
\end{assumption}

\begin{samepage}
\begin{assumption}[Conditional No Anticipation, NA]\label{assu:NA}
For all $t < f$ and all $x$ in the support of $X \mid F = f$:
\begin{equation}\label{eq:NA}
\E[Y_t(f) \mid F = f, \, X = x] = \E[Y_t(\infty) \mid F = f, \, X = x].
\end{equation}
\end{assumption}
\end{samepage}

\begin{samepage}
\begin{assumption}[Conditional Independence, CI]\label{assu:CI}
For every date $s$, cohorts $f,f'>s$, including $\infty$, and $x$ in their common support:
\begin{equation}\label{eq:CI-alt}
\E[Y_s(\infty) \mid F = f, \, X = x] = \E[Y_s(\infty) \mid F = f', \, X = x].
\end{equation}
\end{assumption}
\end{samepage}

Assumption~\ref{assu:PT} requires the same mean untreated outcome trends across first-birth cohorts. Assumption~\ref{assu:NA} requires that anticipation of parenthood have no effect on mean pre-birth outcomes within each first-birth cohort and covariate cell.\footnote{This is weaker than the pointwise no-anticipation condition $Y_t(f) = Y_t(\infty)$ for all $t < f$ used in \cite{roth2023s}.} Assumption~\ref{assu:CI} requires that, conditional on $X$, people who are childless at a given date have the same mean untreated outcome regardless of when or whether they later have a child. Since $X$ includes year of birth, this comparison also holds age fixed. Assumptions~\ref{assu:NA} and~\ref{assu:CI} state two economic restrictions separately. CI equates untreated means across future first-birth cohorts within $X$, while conditional NA sets each cohort's mean anticipation effect to zero within $X$.\footnote{Assumption~3 in \citet{kleven2025child}, Equation~(5), directly equates parents' pre-birth mean outcomes with those of matched donors. In my notation, $\E[Y_s(f) \mid F=f, X=x] = \E[Y_s(F) \mid F>s, X=x]$ for $s<f$. This equality can hold when differences in untreated means are offset by differences in anticipation effects. Assumptions~\ref{assu:NA} and~\ref{assu:CI} exclude such offsetting and imply Kleven's condition on the relevant support. Kleven's condition also suffices for Theorem~\ref{thm:identification} when combined with Assumption~\ref{assu:PT} and no anticipation in cohort means, $\E[Y_s(f)-Y_s(\infty) \mid F=f]=0$ for $s<f$.} The following theorem establishes identification of the ATT under these assumptions.

\begin{samepage}
\begin{theorem}[Identification of ATT]\label{thm:identification}
Let Assumptions~\ref{assu:PT}--\ref{assu:CI} hold. Then for any $(f,t)$ with $2 \le f \le t \le T$,
\begin{multline}\label{eq:identification}
\text{ATT}(f,t) = \E[Y_t \mid F = f] - \E[Y_t \mid F > t] \\
- \E_x\!\big[\E[Y_{f-1} \mid F > f-1, X = x]\big| F = f\big] + \E_x\!\big[\E[Y_{f-1} \mid F > f-1, X = x]\big| F > t\big].
\end{multline}
\end{theorem}
\end{samepage}

All expectations on the right-hand side of \eqref{eq:identification} are identified from observable data. The first two terms are mean observed outcomes of parents and childless individuals at time $t$, and the last two terms are outcomes of childless individuals at time $f-1$ integrated over the covariate distributions of parents ($X \mid F = f$) and those still childless at $t$ ($X \mid F > t$), respectively. The last two terms subtract the pre-birth outcome difference attributable to the groups' different covariate distributions. If parents and childless individuals have the same covariate distribution, this adjustment is zero.

\begin{remark}[Comparison with panel data]\label{rem:panel-did}
With panel data, each parent's own pre-birth outcome is observed, and applying Assumptions~\ref{assu:PT} and~\ref{assu:NA} within gender $g$ yields
\begin{equation}\label{eq:panel-att-g}
\text{ATT}^g(f,t) = \E[Y_t - Y_{f-1} \mid F = f, \, g] - \E[Y_t - Y_{f-1} \mid F > t, \, g],
\end{equation}
so cohort-level parallel trends and no anticipation suffice and Assumption~\ref{assu:CI} is not needed. In repeated cross-sections, these pre-birth means must instead be recovered from matched childless donors. Assumption~\ref{assu:CI} justifies this substitution.
\end{remark}

I now derive the bias that arises when Assumption~\ref{assu:CI} fails. This assumption is restrictive because career and fertility decisions are jointly determined. For instance, women in high-paying professions may delay or forgo childbearing. If permanently childless individuals in the donor group differ in unobserved labor market ability from future parents, the imputed pre-birth outcomes are biased.

To illustrate the bias, consider the following additive model for observed outcomes:
\begin{equation}\label{eq:outcome-model}
Y_{it}^g = \pi_g(\tau_{it}) + \underbrace{\phi_{a_{it}}^g}_{\text{age FE}} + \underbrace{\psi_t^g}_{\text{time FE}} + \underbrace{\alpha_i^g}_{\text{ind.\ FE}} + X_i'\beta^g + \underbrace{\varepsilon_{it}^g}_{\text{shock}}.
\end{equation}
Here $\pi_g(\tau)$ is the common effect of parenthood at event time $\tau$. It is zero for never-parents and before childbirth. The terms $\phi_a^g$, $\psi_t^g$ and $\alpha_i^g$ are age, calendar-year and individual fixed effects.\footnote{The coefficient of year of birth in $\beta^g$ is normalized to zero because it is collinear with the age and calendar-year fixed effects.} I assume $\E[\varepsilon_{it}^g \mid F_i, X_i] = 0$, so selection on unobservables operates through $\alpha_i^g$.

With panel data, comparing each person's outcome after childbirth with their own outcome in the year before childbirth removes both $\alpha_i^g$ and $X_i'\beta^g$. For $\tau\geq0$ and $s=F_i-1=t-\tau-1$,
\[
\E[Y_{it}^g-Y_{is}^g\mid\tau_{it}=\tau,a_{it}=a]
=\pi_g(\tau)+(\phi_a^g-\phi_{a-\tau-1}^g)+(\psi_t^g-\psi_s^g).
\]
After accounting for age and calendar-year effects, this mean change equals $\pi_g(\tau)$ even if the timing of childbirth is related to $\alpha_i^g$.

In repeated cross-sections, however, parents observed at childbirth in year $t$ at age $a$ are compared with donors observed at $s=t-1$. This donor pool includes both future parents and never-parents. Conditional on $(a-1,x,g)$, let $\omega_k^g=\Pr(F=s+k\mid F>s,a-1,x,g)$ for $k\ge1$, and let $\omega_\infty^g=\Pr(F=\infty\mid F>s,a-1,x,g)$.\footnote{These shares sum to one. Their dependence on $s$, $a-1$ and $x$ is suppressed in the notation.} Define $d_g(t,a)=(\phi_a^g-\phi_{a-1}^g)+(\psi_t^g-\psi_{t-1}^g)$ and
\[
\Delta_g(t,a,x)=\E[\alpha_i^g\,|\,F=t,a,x]-\sum_{k\ge1}\omega_k^g\,\E[\alpha_i^g\,|\,F=t-1+k,a-1,x]-\omega_\infty^g\,\E[\alpha_i^g\,|\,F=\infty,a-1,x],
\]
which is the difference in means of individual fixed effects between parents observed at birth and their donor pool. $\Delta_g(t,a,x)$ also equals their gap in mean untreated outcomes at $t-1$. The pseudo-panel comparison of parents observed at birth with the donor pool one year earlier is
\begin{multline}\label{eq:pseudo-bias}
\E[Y_{it}^g\mid F=t,a,x]-\E[Y_{i,t-1}^g\mid F>t-1,a-1,x]
=\pi_g(0)+d_g(t,a)+\Delta_g(t,a,x).
\end{multline}
Within this model, Assumption~\ref{assu:CI} implies $\Delta_g(t,a,x)=0$.

As in the panel case, the pseudo-panel regression absorbs $d_g(t,a)$ through age and year fixed effects. However, $\Delta_g(t,a,x)$ may remain when Assumption~\ref{assu:CI} is violated, and this bias cancels in the gender difference only if it is the same for women and men. For example, if childless women have higher untreated outcomes than mothers and the corresponding male gap is zero, the child penalty from this comparison is overestimated. Richer matching variables can reduce this bias.\footnote{\citet{kleven2025child} finds positive jumps in men's outcomes at $\tau=0$ when matching only on age, year and gender, which he attributes to selection rather than childbirth. These jumps disappear when education, marital status, race and state are added.}

\section{Estimation}\label{sec:estimation}

I first describe the pseudo-panel regression of \citet{kleven2025child}. All regressions are estimated separately by gender, and $g$ is suppressed when no confusion arises. The first-birth year $F_i$ (and hence $\tau_{it}$) for parents is typically inferred from the oldest child's age. Define a \emph{cell} as $c = (X, t)$: all individuals with demographics $X$ observed in year $t$. Within each cell, let $N^C_c$ denote the number of childless individuals and $\bar{Y}^C_c$ their mean outcome. Estimation is restricted to an event-time window $\tau \in \mathcal{T} := \{-\underline{T}, \dots, \overline{T}\}$.

Consider a parent $i$ observed at $\tau = 0$ in cell $(X_i, F_i)$. For each pre-birth event time $\tau < 0$, the \emph{donor cell} is $(X_i, F_i + \tau)$: childless individuals with the same demographics observed in year $F_i + \tau$. Since $X_i$ includes year of birth, they have the same age as the parent had $|\tau|$ years earlier. The parent's pre-birth outcome is imputed as the donor-cell mean, $\tilde{Y}_{i\tau}=\bar{Y}^C_{(X_i,\, F_i + \tau)}$. This backward matching restricts which parents can be used. The donor cell at event time $\tau < 0$ requires childless observations in year $F_i + \tau$, before the parent's first child is born, so when the data begin in calendar year $\underline{t}$, parents with $F_i + \tau < \underline{t}$ have no donors at that event time.\footnote{With $\underline{T} = 5$ pre-birth years and data starting in 2006, for instance, only first-birth cohorts with $F_i \ge 2011$ have complete pre-birth donor coverage.} For post-birth event times $\tau \ge 0$, the observed outcomes of all parents are used. Any parent $j$ observed in year $t$ with $\tau_{jt} = \tau$ contributes $\tilde{Y}_{j\tau} = Y_{jt}$.

Using the constructed pseudo-panel, the following regression is estimated by OLS, and the estimator is referred to as pseudo-OLS:
\begin{equation}\label{eq:event-study}
	\tilde{Y}^{g}_{i\tau}
	=
	\sum_{\tau' \in \mathcal{T}} \beta^{g}_{\tau'}
	\, \mathbf{1}\{\tau = \tau'\}
	+ \phi^g_{a_{i\tau}} + \psi^g_{t_{i\tau}} + \nu^{g}_{i\tau}.
\end{equation}
Here $t_{i\tau}=F_i+\tau$ is the year of observation $(i,\tau)$, and $a_{i\tau}$ is the parent's age in that year. The terms $\phi_a^g$ and $\psi_t^g$ are age and calendar-year fixed effects. The coefficient $\beta^g_{\tau^\star}$ is normalized to zero for a base period $\tau^\star < 0$. The child penalty at event time $\tau$ is $\beta^m_\tau - \beta^w_\tau$. Penalties are also reported in percentage terms. For each gender, the estimated counterfactual mean $\hat c^g_\tau$ is obtained by averaging the predicted outcome net of the estimated event-time effect, $\widehat{\tilde Y}^g_{i\tau}-\hat\beta^g_\tau$, across pseudo-panel observations at event time $\tau$. The scaled penalty is $\hat\beta^m_\tau/\hat c^m_\tau-\hat\beta^w_\tau/\hat c^w_\tau$.\footnote{For event studies with panel data, \cite{leventer2026identification} shows that the gender gap in normalized event-study effects is not identified under his normalized triple-difference framework when parallel trends fails within gender, and proposes the effect of parenthood on the gender earnings ratio as an alternative estimand.} 

Assumptions~\ref{assu:PT}--\ref{assu:CI} identify the cohort-time ATT nonparametrically. The event-study regression in \eqref{eq:event-study} approximates untreated outcomes by additive age and year effects and aggregates treatment effects across cohorts into a common event-time profile. If this specification is correct, the population event-time coefficients equal the ATT. Otherwise, the coefficients need not recover the average causal effects.

Although the pseudo-panel is constructed at the individual level, its regressors and imputed pre-birth outcomes are constant within each \emph{regression cell} $r=(X,t,\tau)$. The pseudo-OLS regression can therefore be rewritten as an equivalent WLS regression using cell means. Let $\mathcal{R}^g$ denote the set of regression cells for gender $g$. For $r=(X,t,\tau)$ with $\tau\geq0$, let $N^P_r$ denote the number of parents observed in that cell. Define
  \[
  \bar{Y}_r =
  \begin{cases}
  \displaystyle
  \frac{1}{N^P_r}\sum_{j:(j,t)\in r} Y_{jt}, & \tau \ge 0,\\[2ex]
  \displaystyle
  \bar{Y}^{C}_{(X,\, t)}, & \tau < 0,
  \end{cases}
  \qquad
  w_r =
 \begin{cases}
   N^{P}_{r}, & \tau \ge 0,\\
   N^{P}_{(X,\, t-\tau,\, 0)}, & \tau < 0,
 \end{cases}
 \]
where $N^P_{(X,\, t-\tau,\, 0)}$ is the number of parents with characteristics $X$ observed at childbirth in year $F=t-\tau$. Each of these parents receives the same donor-cell mean as the imputed outcome at $\tau<0$, so the mean appears once per parent in the individual-level regression. Let $\mathbf d_r$ be the vector of event-time, age and year dummies for cell $r$, omitting the dummy for $\tau^\star$ and one year dummy, and let $\mathbf b^g$ collect their coefficients. The regression assigns the same value, $\mathbf d_r'\mathbf b^g$, to all observations in a cell. Since their deviations from the cell mean sum to zero,
\[
\begin{aligned}
\sum_{r\in\mathcal R^g}\sum_{(i,\tau)\in r}
(\tilde Y_{i\tau}-\mathbf d_r'\mathbf b^g)^2
&=\sum_{r\in\mathcal R^g}\sum_{(i,\tau)\in r}
(\tilde Y_{i\tau}-\bar Y_r)^2\\
&\quad+\sum_{r\in\mathcal R^g}w_r
(\bar Y_r-\mathbf d_r'\mathbf b^g)^2.
\end{aligned}
\]
The first term on the right does not depend on $\mathbf b^g$, while the second is the WLS objective using cell means and weights $w_r$. Minimizing either regression objective therefore gives the same coefficients under full column rank: $\widehat{\mathbf b}^{g}_{\mathrm{WLS}}=\widehat{\mathbf b}^{g}_{\mathrm{OLS}}$. 

This cell representation makes the regression weights explicit and expresses the estimator as a function of cell means and counts, which I use to derive its analytical variance. It also reduces the size of the regression. With the covariates, calendar years and event window held fixed, the number of regression cells is bounded, while the number of individual-level observations grows with the number of parents. The cell regression can therefore be much faster in large samples.

\subsection{Consistency of the Cell-Mean Estimator}\label{sec:consistency}

I now show that, as cell sizes grow, the pseudo-panel OLS estimator converges to the probability limit of the infeasible ``oracle'' OLS estimator that uses the true (never observed) panel data. The convergence relies on two population objects. For post-birth cells ($\tau \ge 0$), the cell mean converges to the population conditional mean among parents at that event time:
\[
\bar{Y}_r \to_p \E[Y_{it} \mid X_i = x, \, \tau_{it} = \tau].
\]
For pre-birth cells ($\tau < 0$), the donor-cell mean converges to the childless population mean:
\[
\bar{Y}_r \to_p \E[Y_{it} \mid X_i = x, \, F_i > t].
\]
Under Assumptions~\ref{assu:NA} and~\ref{assu:CI}, the childless conditional mean equals the untreated potential outcome mean for parents:
\begin{equation}\label{eq:CI-cell}
\begin{aligned}
\E[Y_{it} \mid X_i = x, \, F_i > t]
&= \E[Y_{it}(\infty) \mid X_i = x, \, F_i > t] \\
&= \E[Y_{it}(\infty) \mid X_i = x, \, F_i = f] \quad \text{for any } f > t,
\end{aligned}
\end{equation}
so that each pre-birth cell mean consistently estimates the mean untreated outcome of parents with the same characteristics.

Define the oracle outcome for each observation in the pseudo-panel as
\[
Y^*_{i\tau} =
\begin{cases}
Y_{i,F_i+\tau}(\infty), & \tau < 0,\\
Y_{i,F_i+\tau}(F_i), & \tau \ge 0.
\end{cases}
\]
Let $\mathbf d_{i\tau}$ denote the corresponding event-time, age, and year dummies. Expectations for the oracle regression are taken over the limiting distribution of its observations, with cell probabilities $p_r=\operatorname{plim}\{w_r/\sum_{r'}w_{r'}\}$.

\begin{assumption}[Sampling and Design]\label{assu:lln}
Let $n_s$ denote the sample size of the year-$s$ cross-section. (i) As $\min_s n_s \to \infty$, the parent and donor counts and outcome sums of every included cell, normalized by the cross-section sizes, converge jointly in probability to their population values, and the relative sizes $n_s/\sum_{s'}n_{s'}$ converge to positive limits. For each included pre-birth cell with first-birth cohort $f$, covariates $x$ and date $s<f$, the mean of $Y_{is}(\infty)$ among the parents counted by $N^P_{(x,f,0)}$ also converges in probability to $\E[Y_s(\infty)\mid F=f,X=x]$. (ii) The set of included cells is fixed in the limit, and every cell has positive population mass. (iii) The limiting design matrix $Q^* := \E[\mathbf d_{i\tau} \mathbf d_{i\tau}']$ is nonsingular.
\end{assumption}

\noindent Part (i) of Assumption~\ref{assu:lln} is a joint law of large numbers for the observed cell sums and the unobserved pre-birth means used in the oracle regression. The latter condition requires parents observed at birth to represent their cohort and covariate group in mean untreated pre-birth outcomes. Stable relative sizes of the cross-sections give limits for the regression weights. Part (ii) keeps every cell mean well defined in the limit, and part (iii) makes the probability limit unique. The following proposition establishes consistency.

\begin{prop}[Convergence to Oracle OLS]\label{prop:oracle}
Suppose Assumptions~\ref{assu:NA}--\ref{assu:lln} hold. Then, as $\min_s n_s \to \infty$,
\[
\widehat{\mathbf b}^{\textup{OLS}} \to_p \mathbf b^*_0 := (Q^*)^{-1} \E[\mathbf d_{i\tau} Y^*_{i\tau}],
\]
which is the probability limit of the infeasible oracle OLS regression of $Y^*_{i\tau}$ on $\mathbf d_{i\tau}$.
\end{prop}

\subsection{Asymptotic Theory for the Cell-Mean Estimator}\label{sec:cell-inference}

The sampling variation of the pseudo-panel comes from both outcomes and the counts that determine its weights. The same donor-cell mean can be used to impute pre-birth outcomes for parents from different first-birth cohorts. In surveys that reinterview respondents, an individual can contribute as a donor in one year and as a parent in another. Spouses can also contribute to opposite sides of the gender difference. Conventional standard errors treating the constructed cells as independent observations do not account for this variation.

Let $N$ denote the number of independent clusters, each of which contains all observed years and all parent and donor roles of its members.\footnote{In a repeated cross-section with independent individual sampling, each observation is its own cluster.} For cluster $\ell$, let $n^P_{r,\ell}$ be the number of its parent observations in regression cell $r$ with $\tau\ge0$ and $S^P_{r,\ell}$ the sum of their outcomes, and let $n^C_{c,\ell}$ and $S^C_{c,\ell}$ be the same quantities for its childless records in donor cell $c$. Collect these over all included cells of both genders in $H_\ell=\big((S^P_{r,\ell},n^P_{r,\ell})_{r},\,(S^C_{c,\ell},n^C_{c,\ell})_{c}\big)$ and write $\hat\eta=N^{-1}\sum_{\ell=1}^N H_\ell$, $\eta=\E[H_\ell]$. Cell means are ratios of coordinates of $\hat\eta$, $\bar Y_r=\sum_\ell S^P_{r,\ell}/\sum_\ell n^P_{r,\ell}$ for post-birth cells and $\bar Y^C_c=\sum_\ell S^C_{c,\ell}/\sum_\ell n^C_{c,\ell}$ for donor cells, and the regression weights are obtained from these parent counts.

\begin{assumption}[Regularity]\label{assu:reg}
(i) Clusters are i.i.d., and $H_\ell$ has finite second moments.
(ii) The event window and the set of included cells are fixed. Every denominator used for a cell mean or normalized weight has a strictly positive population limit.
(iii) The limiting weighted design matrix $Q^*=\sum_r p_r\mathbf d_r\mathbf d_r'$, where $p_r$ is the probability limit of $w_r/\sum_{r'}w_{r'}$, is nonsingular.
(iv) The sample rule that determines which cells are included selects the set of (ii) with probability tending to one.
\end{assumption}

Assumption~\ref{assu:reg} supplies the law of large numbers for the observed cell moments and the finite second moments and independence across clusters needed for a central limit theorem. The additional law of large numbers for unobserved pre-birth means in Assumption~\ref{assu:lln}(i) is used only for the oracle comparison. The sets of regression cells $r=(X,t,\tau)$ and donor cells $c=(X,t)$ retained in the limit remain fixed as $N$ grows.\footnote{Allowing the number of cells to grow with $N$ would require additional rate and bias conditions \citep{mckenzie2004asymptotic}.}

\begin{prop}[Asymptotic Normality]\label{prop:cell-clt}
Under Assumption~\ref{assu:reg}, write the child penalty at an event time, or a scaled penalty averaged over post-birth event times, as $\hat P=G(\hat\eta)$ and its probability limit as $P^*=G(\eta)$, for a $G$ that is continuously differentiable in a neighborhood of $\eta$. Then
\[
\sqrt N(\hat P-P^*)\ \xrightarrow{d}\ N(0,V),\qquad
V=J\Sigma J',\quad J=DG(\eta),\quad \Sigma=\operatorname{Var}(H_\ell).
\]
Partition the coordinates of $H_\ell$ into the parent coordinates $(S^P,n^P)$, labeled $P$, and the donor coordinates $(S^C,n^C)$, labeled $D$, and partition $J$ and $\Sigma$ conformably. Then
\begin{equation}\label{eq:two-sample-var}
V=J_P\Sigma_{PP}J_P'+J_D\Sigma_{DD}J_D'
       +J_P\Sigma_{PD}J_D'+J_D\Sigma_{PD}'J_P'.
\end{equation}
Moreover, let $\hat J=DG(\hat\eta)$, $\hat\psi_\ell=\hat J(H_\ell-\hat\eta)$ and $\bar\psi=N^{-1}\sum_\ell\hat\psi_\ell$. The covariance estimator
\begin{equation}\label{eq:analytical-block-cov}
\widehat{\operatorname{Var}}(\hat P)
=\frac{1}{N(N-1)}\sum_{\ell=1}^N
 (\hat\psi_\ell-\bar\psi)(\hat\psi_\ell-\bar\psi)'
\end{equation}
satisfies $N\widehat{\operatorname{Var}}(\hat P)\to_p V$.\footnote{For independent annual samples with fixed numbers of clusters, replace $\Sigma$ by the weighted average of the within-year covariances. Remark~\ref{rem:fixed-year-sampling} gives the variance estimator and the corresponding bootstrap.}
\end{prop}

The first two terms in \eqref{eq:two-sample-var} account for sampling variation in the parent and donor moments, and the cross terms arise when the same cluster enters both, as when a person is observed as a donor and later as a parent. $\Sigma_{PP}$ also includes covariances across genders and event times. The child penalty is the difference between the male and female estimates. When spouses belong to the same cluster, a household shock moves both, so the variance of the difference is not the sum of the two variances, and positive covariance lowers it while negative covariance raises it. Covariances across event times also enter the variance of an average penalty. Conventional regression standard errors that ignore donor reuse and dependence within sampling clusters can therefore be too small or too large.

The Jacobian $J$ in Proposition~\ref{prop:cell-clt} makes explicit that the regression coefficients vary with the sample through two channels.\footnote{Appendix~\ref{app:analytical-implementation} gives the complete derivatives, including those of the counterfactual denominator.} For the unscaled regression, let $\mathbf D$ be the matrix with rows $\mathbf d_r'$, $w$ the limiting row weights, $\mu$ the vector of distinct population cell means, $L$ the matrix that assigns each row its cell mean, $Q=\mathbf D'\operatorname{diag}(w)\mathbf D$ and $\mathbf b^*=Q^{-1}\mathbf D'\operatorname{diag}(w)L\mu$. Then
\[
d\mathbf b=Q^{-1}\!\left\{\mathbf D'\operatorname{diag}(w)L\,d\mu
       +\mathbf D'\operatorname{diag}(dw)(L\mu-\mathbf D\mathbf b^*)\right\}.
\]
The first term is the effect of sampling variation in the cell means. The second is the effect of the estimated weights, which vanishes when the additive regression fits the cell means exactly. Scaled penalties add a third channel, the sampling variation of the counterfactual denominator, which \citet{kleven2025child} treats as nonstochastic. $J$ includes the derivative with respect to the denominator, so $V$ accounts for this variation too.

A cluster bootstrap that resamples clusters and repeats the full estimation has the same first-order limiting distribution under the conditions of Proposition~\ref{prop:cell-clt}, provided that bootstrap samples also select the limiting set of cells with conditional probability tending to one.\footnote{Consistency of bootstrap SEs additionally requires the uniform-integrability condition in Remark~\ref{rem:bootstrap} in Appendix~\ref{app:proof-cell-clt}. This condition can fail when the counterfactual denominator of a scaled penalty is close to zero in some bootstrap samples.} The analytical standard error (SE) I derive is practically useful since it avoids the repeated estimation required by a bootstrap.

\section{Extensions}\label{sec:extensions}

The baseline design matches on discrete cells, constructs pre-birth outcomes from childless donors while post-birth outcomes come from whichever parents appear at each event time, and extrapolates age/year effects through an additive regression specification. This section develops three extensions, each addressing one of these limitations. The first is a doubly robust imputation that admits richer covariates. The second matches parents observed at birth both backward and forward and adds synthetic type fixed effects. The third is a fully nonparametric estimator that accommodates heterogeneous treatment effects.

\subsection{Doubly Robust Estimation}\label{sec:dr}

The pseudo-panel estimator can be biased by the parent--donor gap $\Delta_g$ because of unobserved heterogeneity. Here I develop a doubly robust (DR) imputation that accommodates richer, possibly continuous, covariates. It does not weaken conditional independence. Rather, it makes the assumption more credible by conditioning on a richer covariate vector. I view the imputation as a two-sample problem because pre-birth outcomes of parents are imputed using a sample of childless individuals observed at an earlier date, whose covariate distribution may differ from that of the parents.\footnote{\cite{santanna2026composition} develop doubly robust DiD estimators for repeated cross-sections whose covariate composition changes over time. Treatment-group membership is observed before treatment in their setting.} Fix a gender $g$, a first-birth cohort $f$ and a pre-birth event time $\tau < 0$, so the donor year is $t=f+\tau$. Let $S_i = 1$ indicate membership in the target sample of cohort-$f$ parents observed at birth, and $S_j = 0$ membership in the donor sample of childless individuals observed in year $t$. The fixed cohort and donor year are suppressed in the notation.

Let $Z_i=(X_i, V_i)$ denote a \emph{rich} covariate vector that includes both the discrete variables $X_i$ and additional covariates $V_i$, which may include finely discrete variables and continuous variables. Conditioning on $V_i$ is useful when it is correlated with the unobserved heterogeneity. The DR framework adjusts for $V_i$ through models. The goal is to produce an improved pre-birth imputation $\hat{Y}^{DR}_{i,\tau}$ that replaces the donor-cell mean $\bar{Y}^C_{(X_i, F_i+\tau)}$. The DR-imputed pseudo-panel is then fed into the same event-study regression \eqref{eq:event-study}.

For DR estimation, I maintain Assumptions~\ref{assu:NA} and~\ref{assu:CI} with the richer pre-treatment vector $Z = (X, V)$ in place of $X$. Neither restriction at this finer conditioning level follows automatically from its original version conditional on $X$.\footnote{Conditional mean independence is not monotone in the conditioning set, and mean anticipation effects that vanish within an $X$ cell need not vanish within its $Z$ subcells.} An overlap condition is also required.
\begin{assumption}[Overlap]\label{assu:overlap}
For all $z$ in the support of $Z_i \mid S_i = 1$:
\begin{equation}\label{eq:overlap}
0 < \Pr(S_i = 1 \mid Z_i = z, \, g) < 1.
\end{equation}
\end{assumption}
Assumption~\ref{assu:overlap} requires that for every covariate profile $z$ observed among parents, there exist childless donors with similar profiles. Weak overlap is commonly addressed by trimming.\footnote{Trimming can exclude parents observed at birth, whose counts determine the pre-birth cell weights. It can therefore change the population represented by the estimate.}

The imputation combines two nuisance models. The first models the conditional mean outcome among childless donors as a function of the covariates:
\begin{equation}\label{eq:outcome-model-dr}
m_g(z) = \E[Y_i \mid S_i = 0, \, Z_i = z, \, g] \ \textit{(Outcome model, OM)}. 
\end{equation}
The second models the probability that an individual belongs to the parent sample, conditional on the covariates:
\begin{equation}\label{eq:propensity}
e_g(z) = \Pr(S_i = 1 \mid Z_i = z, \, g)  \ \textit{(Propensity score, PS)}.
\end{equation}
Both models are directly estimable. The outcome model is estimated by regressing outcomes on $Z$ within the donor sample. The propensity score is estimated by a binary choice model (e.g., probit or logit) for $S_i$ on $Z$ in the pooled sample of parents and donors. Intuitively, the propensity score measures how ``parent-like'' a covariate profile is. Donors resembling parents (high $e_g(z)$) receive larger weights.

For each parent $i$ at pre-birth event time $\tau < 0$, the DR imputation is:
\begin{equation}\label{eq:dr-estimator}
\hat{Y}^{DR}_{i,\tau}
=
\hat{m}_g(Z_i)
+
\frac{\sum_{j \in \mathcal{D}} \hat{w}_g(Z_j) \big(Y_{jt} - \hat{m}_g(Z_j)\big)}{\sum_{j \in \mathcal{D}} \hat{w}_g(Z_j)}, \quad \hat{w}_g(Z_j) = \frac{\hat{e}_g(Z_j)}{1 - \hat{e}_g(Z_j)},
\end{equation}
where $\mathcal{D}$ denotes the donor pool in the donor year, and $\hat{w}_g(Z_j)$ is the inverse probability weight, an estimate of the odds $w_g(z) = e_g(z)/(1-e_g(z))$. The weights are normalized to sum to one. The first term predicts the untreated outcome at parent $i$'s covariates using the outcome model. The second term corrects for misspecification of $m_g$ by adding a weighted average of donor residuals.\footnote{The weighted residual correction follows the augmented inverse-probability-weighting approach of \citet{robins1994estimation}. \citet{santanna2020doubly} use this principle to construct DR estimators of the ATT in panel and repeated-cross-section DiD settings.} When $V_i = \emptyset$ and $m_g$ is saturated in covariates, the expression reduces to the donor-cell mean $\bar{Y}^C_{(X_i, F_i+\tau)}$, recovering the baseline pseudo-OLS. The average imputation over parents is doubly robust as it remains consistent if either the outcome model or the propensity score is correctly specified.

The DR estimator of the mean pre-birth outcome also has an efficiency property when both nuisance models are correct. I state it for an independent two-sample setting, separating the precision of one imputed mean from inference for the subsequent event-study regression.

\begin{prop}[Double Robustness and Efficiency]\label{prop:dr}
Fix $g$, $f$ and $\tau$ as above, and let $\mu^0_{g,\tau}(z) := \E[Y_{i,f+\tau}(\infty) \mid Z_i = z, S_i = 1, g]$ and $\mu^0_{g,\tau} := \E[\mu^0_{g,\tau}(Z_i) \mid S_i = 1, g]$ be the parents' conditional and mean untreated pre-birth outcomes. Write $E_s$ and $\operatorname{Var}_s$ for expectation and variance conditional on $S=s$, and $P^Z_s$ for the distribution of $Z$ given $S=s$. Suppose Assumptions~\ref{assu:NA} and~\ref{assu:CI} hold with $Z$ in place of $X$, together with Assumption~\ref{assu:overlap}. Let the estimated nuisance functions have probability limits $m^\dagger_g$ and $w^\dagger_g$ with $E_0[w^\dagger_g]>0$ and finite second moments, and suppose that sample averages of the estimated functions and of the weighted residuals converge to their population limits. Write $\hat\mu^{DR}_{g,\tau}=n_1^{-1}\sum_{i:S_i=1}\hat Y^{DR}_{i,\tau}$ for the average imputation over the $n_1$ parents. Then:
\begin{enumerate}
\item[\textup{(a)}] If the outcome model is correctly specified, so that $m^\dagger_g=m_g$ and $\hat{m}_g(Z_i)-m_g(Z_i) \to_p 0$ for each parent $i$, and the weights are uniformly bounded, then $\hat{Y}^{DR}_{i,\tau} \to_p \mu^0_{g,\tau}(Z_i)$ for each parent $i$, and consequently $\hat\mu^{DR}_{g,\tau} \to_p \mu^0_{g,\tau}$.
\item[\textup{(b)}] If $w^\dagger_g \propto w_g$, as when the propensity score is correctly specified, then $\hat\mu^{DR}_{g,\tau} \to_p \mu^0_{g,\tau}$.
\item[\textup{(c)}] Suppose the parents' covariates $\{Z_i: S_i=1\}$ and the donor observations $\{(Y_j,Z_j): S_j=0\}$ are two mutually independent i.i.d.\ samples with sizes $n_1,n_0$, where $n_s/N\to\lambda_s\in(0,1)$ and $N=n_1+n_0$. Under correctly specified, smooth finite-dimensional nuisance models and the regularity conditions stated in Appendix~\ref{app:proof-dr}, the estimator $\hat\mu^{DR}_{g,\tau}$ is semiparametrically efficient for $\mu^0_{g,\tau}$ in the unrestricted two-sample observed-data model. Writing $\rho_g=dP_1^Z/dP_0^Z$, its limiting variance is
\begin{equation}\label{eq:dr-eff-bound}
\begin{aligned}
\sqrt N(\hat\mu^{DR}_{g,\tau}-\mu^0_{g,\tau})
&\xrightarrow{d}\mathcal N(0,V^{\mathrm{eff}}_{g,\tau}),\\
V^{\mathrm{eff}}_{g,\tau}
&=\frac{\operatorname{Var}_1[m_g(Z)]}{\lambda_1}
+\frac{E_0[\rho_g(Z)^2\operatorname{Var}_0(Y\mid Z)]}{\lambda_0}.
\end{aligned}
\end{equation}
\end{enumerate}
\end{prop}

\noindent In part~(c), the two terms of \eqref{eq:dr-eff-bound} reflect uncertainty about the parents' covariate distribution and about donor outcomes given those covariates. Under the correctly specified parametric models and regularity conditions assumed there, nuisance estimation has no first-order effect on the estimated pre-birth mean. The estimator therefore attains the efficiency bound.

To implement the DR estimator while controlling for the cell variables $X$ nonparametrically, I apply the correction within donor cells. For a donor cell $c=(X,t)$, let $\mathcal D_c$ contain its childless donors. The additional covariates $V$ vary within $c$. Both nuisance models are indexed by $c$, so a parent's covariates are compared with donors observed in the donor year. Write $m_{g,c}(v)=\E[Y\mid S=0,c,V=v,g]$ and let $w_{g,c}(v)$ be proportional to the within-cell density ratio $dP(V\mid S=1,c,g)/dP(V\mid S=0,c,g)$. The within-cell imputation is
\begin{equation}\label{eq:cell-dr}
\hat Y^{DR}_{i,\tau}=\hat m_{g,c}(V_i)+
\frac{\sum_{j\in\mathcal D_c}\hat w_{g,c}(V_j)\{Y_j-\hat m_{g,c}(V_j)\}}
     {\sum_{j\in\mathcal D_c}\hat w_{g,c}(V_j)},
\end{equation}
which gives larger weights to donors whose covariate profiles are more typical of the parents in the cell. Proposition~\ref{prop:dr} applies cell by cell, with $Z=(X,V)$ and overlap within each cell.

DR replaces the pre-birth donor-cell means with imputed outcomes but leaves the event-study regression unchanged. The cell-level regression still applies, because regressors and weights remain constant within regression cells even though DR imputations differ across parents. If every imputed cell mean is consistent, the coefficients converge to the oracle projection of Section~\ref{sec:consistency}. DR imputations depend on the estimated outcome and propensity models. Valid inference must account jointly for their estimation errors and the sampling variation in cell means, counts and the counterfactual means used to scale the penalty. 

For finite-dimensional nuisance models, let $\zeta$ collect the cell moments, OM and PS parameters, and nuisance-dependent sums entering the child penalty or scaled penalty. Now write the estimator as $G_{DR}(\hat\zeta)$ in place of $G(\hat\eta)$ in Proposition~\ref{prop:cell-clt}. Let $\Psi_\ell(\zeta)$ collect cluster $\ell$'s contributions to the estimating equations for the outcome and propensity models, cell moments, and sums. The following proposition derives the asymptotic distribution of the DR penalty estimator and a consistent estimator of its variance.

\begin{prop}[Parametric DR Inference]\label{prop:dr-inference}
Under Assumption~\ref{assu:reg} and the regularity conditions in Appendix~\ref{app:proof-dr-inference}, let $\hat\zeta\to_p\zeta^*$ satisfy
\[
N^{-1}\sum_\ell\Psi_\ell(\hat\zeta)=o_p(N^{-1/2}),
\qquad E[\Psi_\ell(\zeta^*)]=0.
\]
Assume that $\zeta$ is finite dimensional, $A=E[\partial_\zeta\Psi_\ell(\zeta^*)]$ is nonsingular, and $G_{DR}$ is continuously differentiable in a neighborhood of $\zeta^*$. Then
\[
\sqrt N\{G_{DR}(\hat\zeta)-G_{DR}(\zeta^*)\}\ \xrightarrow{d}\ N(0,V_{DR}),\qquad
V_{DR}=E[\psi^{DR}_\ell\psi^{DR}_\ell{}'],
\]
where $\psi^{DR}_\ell=-DG_{DR}(\zeta^*)A^{-1}\Psi_\ell(\zeta^*)$. Define
\[
\hat A=N^{-1}\sum_{\ell=1}^N
\partial_\zeta\Psi_\ell(\hat\zeta),
\qquad
\hat\psi^{DR}_\ell
=-DG_{DR}(\hat\zeta)\hat A^{-1}\Psi_\ell(\hat\zeta).
\]
Let $\widehat{\operatorname{Var}}\{G_{DR}(\hat\zeta)\}$ be
the covariance estimator \eqref{eq:analytical-block-cov}
computed using $\hat\psi^{DR}_\ell$. Then
\[
N\,\widehat{\operatorname{Var}}\{G_{DR}(\hat\zeta)\}
\xrightarrow{p}V_{DR}.
\]
\end{prop}

\noindent The proposition centers the DR penalty estimate $G_{DR}(\hat\zeta)$ at its pseudo-true value $G_{DR}(\zeta^*)$ without assuming that both models are correct. The variance estimator accounts for estimation of both nuisance models, since these estimation effects need not vanish when only one model is correctly specified.

\begin{remark}[Cross-fitting]\label{rem:crossfit}
The nuisance functions can be estimated by machine learning with Neyman orthogonalization following \citet{chernozhukov2018double}, without imposing parametric restrictions. This procedure requires cross-fitting. Clusters are divided into folds, and the estimating equations for each fold are evaluated using nuisance models estimated on the other folds.
\end{remark}

\subsection{Symmetric Matching and Synthetic Type FE}\label{sec:symmetric}

The baseline pseudo-panel construction is asymmetric. Pre-birth outcomes are imputed for the parents observed at birth, whereas post-birth outcomes come from whichever parents are observed at each event time, so the composition of parents can change with event time. I consider an alternative construction that matches the parents observed at birth symmetrically in both directions.\footnote{This construction is suggested in footnote~13 of \citet{kleven2025child}.} For pre-birth outcomes, parents observed at first birth are matched to childless donors, as in the baseline. For post-birth outcomes, they are instead matched to parents with the same first-birth year and observed characteristics. This symmetric construction produces a synthetic history for each parent type, defined by gender, first-birth year, age at first birth and matching characteristics. At each event time, each type is weighted by the number of its parents observed at birth, limiting changes in composition across event times. These synthetic histories also allow type fixed effects to absorb persistent differences in outcome levels across types.

The symmetric construction uses fewer post-birth parents than the baseline. A post-birth parent contributes to the pseudo-panel only if a parent of the same type is observed at birth. For example, if the parents observed at birth are drawn only from 2011 onward, symmetric matching cannot use post-birth outcomes of parents who gave birth before 2011, even when those outcomes are observed after 2011. The baseline construction still uses these cohorts at later event times. In the repeated-cross-section applications, I therefore include parents observed at birth from the first survey year, even if their full pre-birth history cannot be reconstructed. For these parents, pre-birth outcomes are imputed only for event times with available donors.

For gender $g$, define the synthetic type as $\theta=(F_\theta,\tilde a_\theta,X_\theta)$, containing first-birth year, exact age at first birth and the matching characteristics. Let $N^P_{\theta,g,0}$ be the number of its parents observed at birth. For each available event time, form the outcome
\[
\bar Y^{B,g}_{\theta,\tau}=
\begin{cases}
\bar Y^C_{(X_\theta,\,F_\theta+\tau)},&\tau<0,\\
\bar Y^P_\theta(\tau),&\tau\ge0.
\end{cases}
\]
The first line uses the same backward mapping as the baseline, and the second averages observed parents of that type. Each type’s mean outcome at each available event time is weighted by $W^g_{\theta,\tau}=N^P_{\theta,g,0}$, the number of parents observed at birth.\footnote{By the WLS equivalence of Section~\ref{sec:estimation}, this cell regression reproduces the regression on individual-level imputations in which each parent observed at birth carries the synthetic history of his or her type.} These weights replace changing post-birth parent counts in the baseline pseudo-panel. Let $\Theta^g_\tau$ denote the set of types included at event time $\tau$. The realized type share is then $N^P_{\theta,g,0}\mathbf1\{\theta\in\Theta^g_\tau\}/
\sum_{\kappa\in\Theta^g_\tau}N^P_{\kappa,g,0}.$
The share is constant only when the same types are included at every event time.

When the composition of observed types changes, persistent differences in their untreated outcomes can appear as changes caused by childbirth. Type fixed effects absorb these differences. Including type fixed effects, however, requires an additional restriction to identify the event-time coefficients. Because $t=F_\theta+\tau$, the type, year and event-time effects are linearly dependent. With only the base-period (e.g., $\tau^* = -2$) coefficient normalized to 0, adding $c(\tau+2)$ to every event-time coefficient, subtracting $ct$ from the year effects, and adding $c(F_\theta-2)$ to the type effects leaves all cell means estimated by the regression unchanged for any constant $c$. Footnote~12 of \citet{kleven2025child} raises this concern for individual fixed effects and argues that additivity makes separate identification feasible. It does not, as \citet[Section~2.3]{thakral2026child} establish for that case.

To identify the event-time coefficients, I set two pre-birth coefficients to zero. For instance, one can impose $\beta^g_{-3}=\beta^g_{-2}=0$ and estimate
\begin{equation}\label{eq:anchored-synth-fe}
\bar Y^{B,g}_{\theta,\tau}
=\alpha^g_\theta+
\sum_{k\notin\{-3,-2\}}\beta^{g}_k\mathbf1\{\tau=k\}
+\phi^g_{\tilde a_\theta+\tau}+\psi^g_{F_\theta+\tau}+u^g_{\theta,\tau},
\end{equation}
using the weights $W^g_{\theta,\tau}$. The two zero restrictions determine the event profile's level and linear trend. They require childbirth to have no effect at the restricted base periods.\footnote{Since each type fixes the individual's year of birth, $t-a=F_\theta-\tilde a_\theta$. For any constant $d$, adding $da$ to the age effects, subtracting $dt$ from the year effects and adding $d(F_\theta-\tilde a_\theta)$ to the type effects leaves both the event-time coefficients and the cell means estimated by the regression unchanged. Hence the two zero restrictions do not remove this dependence. The remaining collinearity is handled by choosing a linearly independent set of age, year and type indicators, without imposing further restrictions on the event-time coefficients.} The following conditions connect the constructed outcomes to population cell means and ensure a well-defined regression. Let $\mathcal R^{S,g} =\{(\theta,\tau):\tau\in\mathcal T,\ \theta\in\Theta^g_\tau\}$. For $r=(\theta,\tau)\in\mathcal R^{S,g}$, write $a_r=\tilde a_\theta+\tau$ and $t_r=F_\theta+\tau$.

\begin{assumption}[Representativeness of Observed Parents]\label{assu:sym-repr}
For every $r\in\mathcal R^{S,g}$ with $\tau\ge0$,
\[
\E[Y_{t_r} \mid r,\; \text{observed},\; g]
=
\E[Y_{t_r}(F_\theta) \mid r,\; g].
\]
\end{assumption}

\begin{assumption}[Symmetric Matching Regularity]\label{assu:sym-reg}
For every $r\in\mathcal R^{S,g}$, Assumption~\ref{assu:lln} holds and $N^P_{\theta,g,0}/N\to_p p^g_\theta>0$ for every included type, where $N$ is the number of clusters. For each regression considered, impose its zero restrictions on event-time coefficients and drop redundant age, year and type indicators. Each omitted indicator must be a linear combination of the retained fixed-effect indicators. Let $\mathbf d_r^{S,g}$ denote the resulting regressor vector. Assume that
\[
\frac{
\sum_{r\in\mathcal R^{S,g}}N^P_{\theta(r),g,0}\,
\mathbf d_r^{S,g}(\mathbf d_r^{S,g})'
}{
\sum_{r\in\mathcal R^{S,g}}N^P_{\theta(r),g,0}
}
\xrightarrow{p}Q^{S,g},
\qquad Q^{S,g}\succ0.
\]
\end{assumption}

\begin{prop}[Oracle Projection]\label{prop:sym-fe-oracle}
Let $\mu^{S,g,*}_r$ be the untreated pre-birth or treated post-birth population mean of $r$, and let Assumptions~\ref{assu:NA}--\ref{assu:CI} and~\ref{assu:sym-repr}--\ref{assu:sym-reg} hold. The event-time coefficient estimates and the cell means estimated by regression~\eqref{eq:anchored-synth-fe} converge to those of the weighted population projection
\begin{equation}\label{eq:anchored-synth-proj}
\min_{\alpha,\beta,\phi,\psi}\sum_{r\in\mathcal R^{S,g}}p^g_{\theta(r)}
\left(\mu^{S,g,*}_r-\alpha_{\theta(r)}
-\sum_{k\notin\{-3,-2\}}\beta_k\mathbf1\{\tau_r=k\}
-\phi_{a_r}-\psi_{t_r}\right)^2,
\end{equation}
which are unique even though the nuisance effects are not.
\end{prop}

As for the baseline regression, the coefficients recover causal effects if (i) each type's untreated mean is the sum of a type-specific constant and common age and year effects; and (ii) the treatment effect at each event time is the same across types. When treatment effects differ across types, part of that variation can be absorbed by the age and year effects. An event-time coefficient can then contain effects from different event times with weights that may be negative, so it need not equal an average effect at that event time. Appendix~\ref{app:bidirectional-comparisons} shows how the regression weights treatment effects and illustrates why its coefficients can differ from average effects, with or without type FE.

\subsection{A Fully Nonparametric Cohort-Time Estimator}\label{sec:np}

So far I have focused on alternative constructions of the pseudo-panel on which an additive regression is run. I now propose a nonparametric estimator that relies neither on backward or forward matching nor on the regression. Instead of imputing a parent's pre-birth outcome, it estimates parents' mean untreated current outcome from childless donors observed in the same year with the same characteristics, under a parallel-trends assumption as follows.
\begin{assumption}[Conditional Parallel Trends, CPT]\label{assu:CPT}
For all $s < f \le t$, and all $x$ in the common support of $X$ given $F=f$ and given $F>t$:
\begin{equation}\label{eq:CPT}
\E[Y_t(\infty) - Y_s(\infty) \mid F = f, \, X = x] = \E[Y_t(\infty) - Y_s(\infty) \mid F > t, \, X = x].
\end{equation}
\end{assumption}
Assumption~\ref{assu:CPT} is neither weaker nor stronger than Assumption~\ref{assu:PT} (PT). Equal cohort averages can hide within-cell differences that offset across cells. Conversely, equal trends within cells need not aggregate to equal cohort averages because each cohort differs in the distribution of $X$ and untreated trends vary with characteristics.\footnote{Write $\Delta_f(x)=\E[Y_t(\infty)-Y_s(\infty)\mid F=f,X=x]$ and $\Delta_{>t}(x)=\E[Y_t(\infty)-Y_s(\infty)\mid F>t,X=x]$. Assumption~\ref{assu:CPT} sets $\Delta_f(x)=\Delta_{>t}(x)$ for every $x$, but the cohort-level trends are $\sum_x\Pr(X=x\mid F=f)\Delta_f(x)$ and $\sum_x\Pr(X=x\mid F>t)\Delta_{>t}(x)$, which can differ when $\Pr(X=x\mid F=f)\neq\Pr(X=x\mid F>t)$ and $\Delta_f(x)$ varies with $x$.} Conditioning on covariates is the standard way to make parallel trends more credible when treated and control units differ in characteristics that predict their untreated trends \citep{heckman1997matching,santanna2020doubly,callaway2021difference}.

\begin{prop}[Conditional ATT]\label{prop:conditional ATT}
	Let Assumptions~\ref{assu:NA}, \ref{assu:CI}, and \ref{assu:CPT} hold. For every cohort $f\ge2$, post-birth date $t\ge f$ and $x$ in the common support, the conditional cohort-time ATT is identified by a within-period comparison:
	\begin{equation}\label{eq:np-identification}
		\begin{aligned}
\text{ATT}_x(f,t) &:= \E[Y_t(f) - Y_t(\infty) \mid F = f, X=x] \\
&= \E[Y_t \mid F=f, X=x] - \E[Y_t \mid F > t, X=x].
\end{aligned}
	\end{equation}
\end{prop}

The nonparametric estimator aggregates parent--donor differences across cohorts and types with the numbers of parents observed at birth as weights, as in symmetric matching and in the spirit of \citet{callaway2021difference}. Let $\theta = (X, \tilde{a})$ be a type defined by the discrete matching variables $X$ and exact age at first birth $\tilde{a}$. For cohort $f$, type $\theta$, and event time $\tau$, the current age is $a = \tilde{a} + \tau$ and the calendar year is $t = f + \tau$. The type-specific post-birth effect is
\[
\widehat{\text{ATT}}^g(f,\tau \mid \theta)
=
\bar{Y}^{P,g}_{f,\tau}(\theta) - \bar{Y}^{C,g}_{t}(X_\theta, \tilde{a}_\theta + \tau),
\qquad t = f + \tau,
\]
where $\bar{Y}^{P,g}_{f,\tau}(\theta)$ is the mean outcome of cohort-$f$ parents of gender $g$ and type $\theta$ observed at event time $\tau$, and $\bar{Y}^{C,g}_{t}(X_\theta, \tilde{a}_\theta + \tau)$ is the mean outcome of childless donors in the corresponding current-age/current-year cell. For a fixed set $\mathcal S_g$ of cohort--type pairs, common to all reported post-birth event times, define
\[
\hat q^g(f,\theta)=\frac{N^{P,g}_{f,\theta,0}}
{\sum_{(f',\theta')\in\mathcal S_g}N^{P,g}_{f',\theta',0}},
\qquad (f,\theta)\in\mathcal S_g,
\]
where $N^{P,g}_{f,\theta,0}$ is the number of cohort-$f$ parents of gender $g$ and type $\theta$ observed at birth. The aggregate is
\[
\widehat{\mathrm{ATT}}^g(\tau)=\sum_{(f,\theta)\in\mathcal S_g}
 \hat q^g(f,\theta)\widehat{\mathrm{ATT}}^g(f,\tau\mid\theta).
\]
The scaled effect divides $\widehat{\mathrm{ATT}}^g(\tau)$ by the average of the donor means $\bar{Y}^{C,g}_{t}(X_\theta, \tilde{a}_\theta + \tau)$ under the same weights, which estimates the counterfactual level, and the scaled penalty is the difference between the scaled effects of men and women. When included types change over event time, I average the effects of the included types within each cohort, weighting by the numbers of parents observed at birth. I then average the included cohort effects using each cohort’s total number of parents observed at birth, including those from types unavailable at that event time (Appendix~\ref{app:analytical-implementation}). Proposition~\ref{prop:cell-clt} applies to the nonparametric estimator and its scaled penalty for inference because both are smooth functions of the cluster sums and counts.

This estimator has advantages over the standard pre-birth imputation approach. First, as its donors are contemporaneous, the estimator needs no pre-birth years for imputation. Parents observed in the first few years of the survey remain usable, which enlarges the sample.\footnote{For parents observed with first birth in the first survey year, date $f-1$ precedes the data. Proposition~\ref{prop:conditional ATT} then requires Assumptions~\ref{assu:NA}, \ref{assu:CI} and~\ref{assu:CPT} to hold from that date.} Second, it takes differences within cells and therefore does not extrapolate age and year effects through an additive regression. Third, it allows treatment effects to differ across cohorts and types. It estimates the effect of each cohort and type separately and averages them using positive weights constructed from the counts of parents observed at birth and normalized to sum to one. Hence it avoids the negative-weights problem in the baseline regression approach. 

However, using contemporaneous donors also means that, at later event times, parents are compared with older childless individuals, who are more likely to be never-parents. These donors may have different untreated outcomes from parents. Assumptions~\ref{assu:CI} and~\ref{assu:CPT} may therefore be harder to justify at later event times. The pseudo-OLS regression instead uses donors at pre-birth ages, where the pool still contains many future parents, at the cost of imposing an additive specification.

\section{Monte Carlo Simulation}\label{sec:simulation}

Parents and childless individuals may differ in ways that matching variables cannot capture. The composition of parents observed can also change after birth, and treatment effects may differ across cohorts. I examine these problems by simulating individual fertility and outcome histories, including proxies for the unobserved characteristics that drive selection. I use these individual histories to compare panel and pseudo-panel estimates, examine when synthetic type FE help, assess the coverage of conventional, analytical and bootstrap confidence intervals, and measure computation time for individual-level pseudo-OLS and cell-level WLS. Appendix~\ref{app:sim-coverage-all} reports additional inference results for the remaining estimators.

\subsection{Benchmark Design and Selection}\label{sec:sim-design}

I generate Monte Carlo samples of $N=50{,}000$ people for $20$ years, $t=0,\ldots,19$. Gender and $X_i\sim\operatorname{Bernoulli}(0.5)$ are independently drawn, initial age is uniform on $\{22,\ldots,30\}$, and the latent individual component is $\alpha_i\sim N(0,1)$. Two observed continuous proxies are $V_{1i}=0.7\alpha_i+\eta_{1i}$ and $V_{2i}=0.5\alpha_i+\eta_{2i}$,
with independent $\eta_{1i}\sim N(0,0.5^2)$ and $\eta_{2i}\sim N(0,0.4^2)$. Untreated potential outcomes are computed by:
\[
Y^0_{it}=1.2+0.065a_{it}-0.0007a_{it}^2+0.16X_i
 +0.018t+\zeta_t+0.70\alpha_i+\varepsilon_{it},
\]
where $\zeta_t\sim N(0,0.04^2)$ is a common calendar shock and $\varepsilon_{it}\sim N(0,0.25^2)$ is individual noise. The continuous proxies predict outcomes through their correlation with $\alpha_i$. For $t=2,\ldots,18$, the first-birth hazard is
\[
\Pr(F_i=t\mid F_i\geq t,X_i,\alpha_i,g_i)
 =\operatorname{logit}^{-1}\{-2.1+0.18X_i-0.016(a_{it}-30)^2+\lambda_{g_i}\alpha_i\}.
\]
The female first-birth effect is $\pi_w(\tau)=-0.30(1-e^{-0.55(\tau+1)})$ and the male effect is $\pi_m(\tau)=0.02(1-e^{-0.45(\tau+1)})$ for $\tau\geq0$. Both are zero before birth. The positive male-minus-female penalty averages $0.28$ over event times 0 to 10. I call this the benchmark design, and use variations of the design for further simulations.

I use the event-time window $[-5,10]$. The panel and pseudo-panel regressions use first-birth cohorts $F_i=5,\ldots,9$, whose full event window falls inside the observation period. The nonparametric estimator (NP) retains the additional cohorts $F_i=2,3,4$ as it does not require pre-birth imputations. I match parents to donors based on gender, $X$, age and calendar year. I use the available cells at each event time, without requiring every type to be observed throughout the event window. Panel OLS uses parents’ observed outcome histories to estimate a regression with event-time, age and year indicators, omitting the indicator for $\tau=-2$. Panel FE additionally includes an individual fixed effect and normalizes the coefficients at $-3$ and $-2$ to zero. Pseudo-OLS, DR, and DR with cross-fitting (DR-CF) regress the constructed pseudo-panel outcomes on event-time, age and year indicators, omitting the indicator for $\tau=-2$. NP uses the contemporaneous donor comparison.

DR uses both continuous proxies. Within each gender, covariate cell, and pre-birth event time, its donor outcome model is
\begin{equation}\label{eq:sim-outcome}
m_c(V)=\alpha_c+b_1V_1+b_2V_2+b_3V_1^2+b_4V_1V_2,
\end{equation}
and its propensity model includes cell intercepts:
\begin{equation}\label{eq:sim-propensity}
\Pr(S=1\mid c,V)=\operatorname{logit}^{-1}(\gamma_c+\gamma_1V_1+\gamma_2V_2).
\end{equation}
DR-CF divides individuals into two folds, estimates both models on each fold and evaluates them using the other fold. All observations of an individual are assigned to the same fold, regardless of whether they appear as a parent or a donor.

I generate 500 replications in each of three scenarios: no selection with $(\lambda_m,\lambda_w)=(0,0)$, symmetric selection with $(0.25,0.25)$, and asymmetric selection with $(0.25,0.15)$. Under symmetric selection the donor bias is the same for both genders, so the biases cancel in the gender difference. Under asymmetric selection, positive selection into parenthood is stronger for men, so the bias does not cancel in the difference. 

\begin{table}[htbp]
\centering
\setlength{\belowcaptionskip}{6pt}
\caption{Selection scenarios and estimation error}\label{tab:mc-summary}
\footnotesize
\setlength{\tabcolsep}{3pt}\renewcommand{\arraystretch}{1.05}
\begin{tabular}{lrrrrrrrrrrr}
\toprule
 & \multicolumn{3}{c}{Selection gap} & \multicolumn{6}{c}{RMSE} & \multicolumn{2}{c}{Pseudo-OLS bias} \\
\cmidrule(lr){2-4}\cmidrule(lr){5-10}\cmidrule(lr){11-12}
Scenario & Men & Women & M$-$W & \shortstack{Panel\\OLS} & \shortstack{Panel\\FE} & \shortstack{Pseudo\\OLS} & NP & DR & DR-CF & $\tau=0$ & $\tau=5$ \\
\midrule
No selection & 0.000 & 0.000 & 0.000 & 0.053 & 0.042 & 0.045 & 0.011 & 0.051 & 0.051 & 0.001 & 0.002 \\
Symmetric & 0.370 & 0.371 & -0.001 & 0.051 & 0.044 & 0.045 & 0.011 & 0.050 & 0.050 & -0.001 & -0.002 \\
Asymmetric & 0.370 & 0.227 & 0.143 & 0.071 & 0.042 & 0.108 & 0.094 & 0.078 & 0.078 & 0.065 & 0.101 \\
\bottomrule
\end{tabular}
\begin{minipage}{\linewidth}\vspace{0.35em}\footnotesize\textit{Notes:} 500 replications with $N=50{,}000$. The selection gap is the difference in the mean of $\alpha_i$ between parents and childless individuals. RMSE is the root mean squared error of the unscaled penalty over event times 0--10, computed within each replication and averaged over replications. The last two columns report the bias of pseudo-OLS at event times 0 and 5.\end{minipage}
\end{table}

\begin{figure}[p!]
\centering
\includegraphics[width=\textwidth]{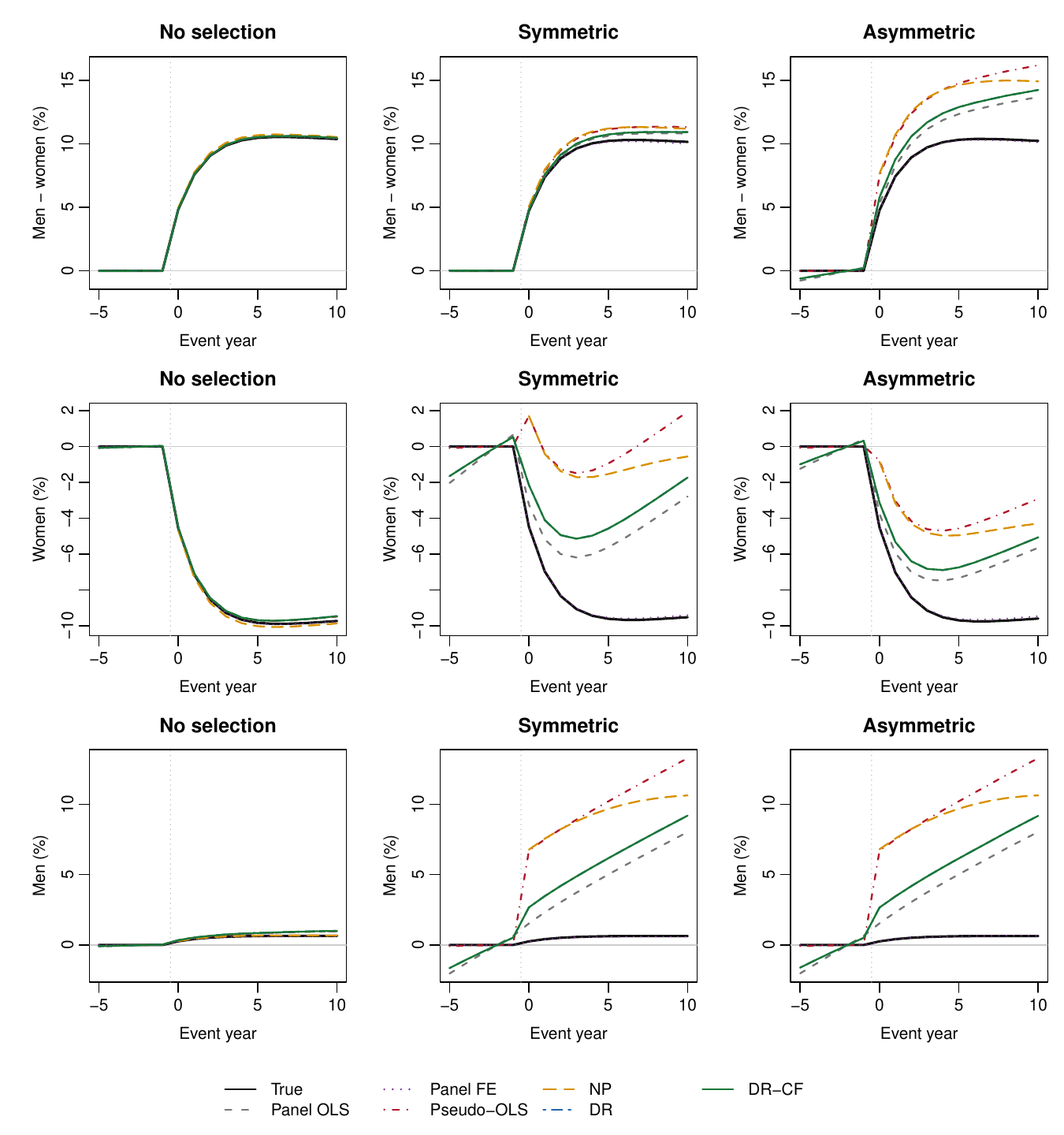}
\caption{Scaled event-time paths in the simulations. Rows show the men-minus-women penalty, women and men. NP has no pre-birth path and is shown only after birth.}
\label{fig:sim-profiles}
\end{figure}

For each scenario, Table~\ref{tab:mc-summary} reports the selection gap (the difference in mean $\alpha_i$ between parents and individuals who remain childless throughout the observation period) by gender and the average root mean squared error (RMSE) of each estimator across the eleven post-birth event times, and over replications. Without selection, or with symmetric selection, every estimator is unbiased and they differ only in precision. Under asymmetric selection, pseudo-OLS acquires a bias, and DR lowers the RMSE by more than a quarter.\footnote{In a version of the asymmetric design with 30,000 people, the RMSE of DR rises from 0.086 to 0.111, the value for pseudo-OLS in that design, as the proxy loadings fall from $(0.7,0.5)$ to zero, so the gain from DR depends on how informative the covariates are about $\alpha_i$, and uninformative covariates do no harm in this design. Cross-fitting gives the same results.} Panel OLS is also biased in the same direction, and its bias grows with event time. The mean of $\alpha_i$ differs across first-birth cohorts. At a given calendar year, parents who had their first child earlier are observed at later event times. Comparisons across event times therefore also involve comparisons across first-birth cohorts with different average $\alpha_i$. Panel FE removes this bias and has the lowest RMSE among the regression estimators in every scenario.

Figure~\ref{fig:sim-profiles} reports the scaled penalties and the separate effects for women and men. Under symmetric selection, pseudo-OLS shifts both women's and men's estimated effects upward, understating women's loss and overstating men's gain. The biases cancel in the unscaled gender difference, but the scaled penalty remains above the true value because the counterfactual means used as denominators are underestimated.

\subsection{Type Fixed Effects and Heterogeneous Treatment Effects}\label{sec:sim-type-fe}\label{app:sim-fixed-cell}

I first examine the role of type FE when untreated outcome levels differ across types. I use the benchmark design without selection. I add a type-specific constant, \(\lambda h(B_i)\), to women's untreated outcomes. Here \(B_i\) is the individual's year of birth. I use a quadratic function so that age and year fixed effects alone cannot absorb these differences across types. The function is standardized to have mean zero and variance one, so \(\lambda\) measures the standard deviation of the added constant across types.\footnote{With initial ages uniform on $\{22,\ldots,30\}$, the years of birth are $B_i\in\{-30,\ldots,-22\}$ and $h(B_i)=\{(B_i+26)^2-20/3\}/\sqrt{308/9}$. I use $\lambda\in\{0,0.2,0.4,0.6\}$.} The level is identical for a parent and her matched donors, who share her age and calendar year, so the donor comparison remains valid and the effect on women is unchanged. I use the same 500 simulated samples for every value of $\lambda$.

\begin{table}[htbp]
\centering
\caption{Persistent type levels and type fixed effects}\label{tab:fe-type-benefit}

\small\setlength{\tabcolsep}{6pt}\renewcommand{\arraystretch}{1.08}
\begin{tabular}{@{}lrrrr@{}}
\toprule
 & \multicolumn{2}{c}{Pseudo-OLS / Sym} & \multicolumn{2}{c}{Sym+FE} \\
\cmidrule(lr){2-3}\cmidrule(lr){4-5}
Type-level SD ($\lambda$) & Bias & RMSE & Bias & RMSE \\
\midrule
0.0 & 0.003 & 0.050 & 0.002 & 0.046 \\
0.2 & 0.012 & 0.052 & 0.002 & 0.046 \\
0.4 & 0.022 & 0.059 & 0.002 & 0.046 \\
0.6 & 0.032 & 0.068 & 0.002 & 0.046 \\
\bottomrule
\end{tabular}
\begin{minipage}{\linewidth}\vspace{0.35em}\footnotesize\textit{Notes:} 500 replications with $N=50{,}000$. Unscaled 0--10 average with true value 0.28. Sym+FE fixes $-3,-2$ and the other regressions omit $-2$.\end{minipage}
\end{table}

Table~\ref{tab:fe-type-benefit} reports the bias and RMSE of the average penalty for increasing values of $\lambda$. The bias of pseudo-OLS grows with $\lambda$ and reaches 0.032 at $\lambda=0.6$, because additive age and year effects cannot absorb the type-specific level.\footnote{Symmetric matching without FE (Sym) gives the same estimates as pseudo-OLS here, because the parents observed at birth are observed at every later event time and no compositional change occurs.} These biases arise even when every parent is observed at every event time. I also shorten the observation period by discarding the last four survey years, so that only earlier first-birth cohorts remain observable at later event times. At \(\lambda=0.6\), the bias of pseudo-OLS and Sym increases from 0.032 to 0.040. 

Sym+FE absorbs the added levels exactly through type fixed effects, so its penalty estimates do not change with $\lambda$. Adding type fixed effects requires estimating a separate intercept for each type and identifying event-time effects from changes within types. This can reduce precision when few observations are available for each type. However, in this design, the type FE cost no precision. With or without added levels, the RMSE of Sym+FE is lower than that of pseudo-OLS. The 90 types per gender have about 78 parents observed at birth on average.

I next allow treatment effects to differ across first-birth cohorts, using the benchmark design without selection and with no added type levels ($\lambda=0$). I multiply women's post-birth effects by 0.5, 0.75, 1, 1.25 and 1.5 for cohorts $F=5,\ldots,9$, respectively, and leave men's effects unchanged. All estimators, including NP, use these same cohorts. Every parent is observed at every event time.

\begin{table}[htbp]
\centering
\caption{Matching regressions and cohort heterogeneity}\label{tab:bidirectional-mc}
\small\setlength{\tabcolsep}{4pt}\renewcommand{\arraystretch}{1.08}
\begin{tabular}{@{}lrrrrrr@{}}
\toprule
 & \multicolumn{2}{c}{Pseudo-OLS / Sym} & \multicolumn{2}{c}{Sym+FE} & \multicolumn{2}{c}{NP} \\
\cmidrule(lr){2-3}\cmidrule(lr){4-5}\cmidrule(lr){6-7}
Design & Bias & RMSE & Bias & RMSE & Bias & RMSE \\
\midrule
Common effect & 0.003 & 0.050 & 0.002 & 0.046 & -0.001 & 0.015 \\
Heterogeneous effects & -0.313 & 0.317 & -0.299 & 0.303 & -0.001 & 0.015 \\
\bottomrule
\end{tabular}
\begin{minipage}{\linewidth}\vspace{0.35em}\footnotesize\textit{Notes:} 500 replications with $N=50{,}000$. All estimators use first-birth cohorts 5--9. Errors are measured against their common ATT, averaged over event times 0--10 in unscaled units: 0.280 with common effects and 0.257 with heterogeneous effects. Sym+FE fixes $-3,-2$; pseudo-OLS and Sym omit $-2$.\end{minipage}
\end{table}

As reported in Table~\ref{tab:bidirectional-mc}, all estimators have little bias under common effects. With heterogeneous effects, the biases of pseudo-OLS and Sym are $-0.313$ and that of Sym+FE is $-0.299$, against an average true penalty of 0.257. The estimated penalties are therefore negative even though the true penalty is positive in every cohort. Type FE cannot absorb differences in treatment effects that begin at birth. NP shows no bias in either case and a low RMSE of 0.015 with heterogeneous effects, which it is designed to accommodate. 

\subsection{Inference from Complete Histories}\label{sec:sim-coverage}

I examine whether confidence intervals using analytical inference attain their nominal coverage for pseudo-OLS and NP when complete histories are available. The data are generated from the benchmark design without the continuous proxies, with $N=10{,}000$ people and 500 replications under either no or symmetric selection. I compare three inference methods for pseudo-OLS: heteroskedasticity-robust standard errors from the individual-level pseudo-panel regression, the analytical inference of Proposition~\ref{prop:cell-clt}, and a cluster bootstrap that resamples complete individual histories with 199 bootstrap samples per replication. NP is assessed with the analytical and bootstrap standard errors. The 95\% confidence intervals are computed for the average unscaled penalty across event times from $0$ to $10$.

\begin{table}[htbp]
\centering
\caption{Coverage of 95\% confidence intervals for pseudo-OLS and NP}\label{tab:coverage}

\footnotesize\setlength{\tabcolsep}{3pt}\renewcommand{\arraystretch}{1.08}
\begin{tabular}{@{}llcccccc@{}}
\toprule
 &  & \multicolumn{3}{c}{No selection} & \multicolumn{3}{c}{Symmetric selection} \\
\cmidrule(lr){3-5}\cmidrule(lr){6-8}
Estimator & Variance estimate & SD & Mean SE & Cov. (\%) & SD & Mean SE & Cov. (\%) \\
\midrule
Pseudo-OLS & Robust SE & 0.112 & 0.039 & 51.0 & 0.115 & 0.039 & 47.4 \\
 & Analytical SE & 0.112 & 0.113 & 94.4 & 0.115 & 0.111 & 94.8 \\
 & Cluster bootstrap & 0.112 & 0.113 & 94.8 & 0.115 & 0.112 & 94.2 \\
\addlinespace
NP & Analytical SE & 0.025 & 0.027 & 96.6 & 0.027 & 0.026 & 94.4 \\
 & Cluster bootstrap & 0.025 & 0.027 & 96.4 & 0.027 & 0.026 & 95.0 \\
\bottomrule
\end{tabular}
\begin{minipage}{\linewidth}\vspace{0.35em}\footnotesize\textit{Notes:} 500 replications with $N=10{,}000$. The target is the unscaled penalty averaged over event times 0--10. SD is the standard deviation of the estimates across replications, Mean SE is the average estimated standard error, and Cov.\ is the coverage of the 95\% confidence interval. Robust SE is the heteroskedasticity-robust standard error of the individual-level regression. The cluster bootstrap resamples complete person histories with 199 bootstrap samples.\end{minipage}
\end{table}

Table~\ref{tab:coverage} reports the Monte Carlo dispersion of the estimate, the mean standard error from each method and the coverage of the resulting confidence intervals. Confidence intervals based on robust standard errors from the individual-level pseudo-panel regression cover the truth only about half of the time. Confidence intervals based on the analytical inference and the cluster bootstrap both attain nominal coverage, for NP as well as for pseudo-OLS, and they agree with each other. Coverage is unchanged under symmetric selection. Appendix~\ref{app:sim-coverage-all} applies the analytical inference to DR, Sym+FE and the scaled penalty. Coverage is between 94\% and 97\%, except for the scaled NP penalty under symmetric selection.

\subsection{Selective Response and Compositional Drift}\label{sec:sim-symmetric}

In the benchmark design without selection on $\alpha_i$, I allow parents in first-birth cohorts 2--9 not to respond to survey waves from time to time. At event times $\tau=1,\ldots,10$, their response probability is $\operatorname{logit}^{-1}\{3-0.5\tau+2X_i+\mathbf1(\tilde a_i>27)\}$. Conditional on $X_i$ and age at first birth, responses are independent across waves, so a parent can return after missing a wave. All outcomes at and before childbirth and all donor observations remain available.

I consider three scenarios. In S1, all parents respond in every wave, and treatment effects at each event time are the same across types within each gender. S2 introduces selective nonresponse and retains homogeneous effects. S3 additionally allows women’s treatment effects to depend on $X_i$ and age at first birth:
\[
\pi_w(\tau,X_i,\tilde a_i)=
 [-0.45+0.35X_i+0.18\cdot\mathbf1(\tilde a_i>27)](1-e^{-0.55(\tau+1)}).
\]
Parents with $X_i=1$ or an age at first birth above 27 are more likely to respond and have smaller penalties. Panel OLS, Panel FE, pseudo-OLS, Sym and Sym+FE all use first-birth cohorts 5–9. I estimate the two panel regressions using all available observations and, separately, only parents observed at every post-birth event time. For NP, I report estimates using cohorts 2–9 and the common cohorts 5–9.

\begin{table}[htbp]
\centering
\caption{Selective response and heterogeneous effects}\label{tab:symmetric-sim}

\small\setlength{\tabcolsep}{6pt}\renewcommand{\arraystretch}{1.08}
\begin{tabular}{@{}lrrrr@{}}
\toprule
 & \multicolumn{1}{c}{S1} & \multicolumn{1}{c}{S2} & \multicolumn{2}{c}{S3} \\
\cmidrule(lr){2-2}\cmidrule(lr){3-3}\cmidrule(lr){4-5}
Estimator & RMSE & RMSE & RMSE & Bias (in S3) \\
\midrule
Panel OLS & 0.061 & 0.061 & 0.063 & -0.010 \\
Panel OLS, balanced & 0.061 & 0.130 & 0.183 & -0.128 \\
Panel FE & 0.046 & 0.046 & 0.046 & 0.005 \\
Panel FE, balanced & 0.046 & 0.099 & 0.169 & -0.137 \\
\addlinespace
Pseudo-OLS & 0.051 & 0.051 & 0.054 & -0.012 \\
NP (cohorts 2--9) & 0.011 & 0.011 & 0.022 & 0.018 \\
NP (cohorts 5--9) & 0.015 & 0.015 & 0.015 & -0.001 \\
Sym & 0.051 & 0.052 & 0.056 & 0.019 \\
Sym+FE & 0.047 & 0.048 & 0.050 & 0.015 \\
\bottomrule
\end{tabular}
\begin{minipage}{\linewidth}\vspace{0.35em}\footnotesize\textit{Notes:} 500 replications with $N=50{,}000$. RMSE and bias of the unscaled penalty averaged over event times 0--10, relative to the full population of first-birth cohorts 5--9. S1 has full response and common effects, S2 selective response and common effects, and S3 selective response and heterogeneous effects. In S2 and S3, cohorts 2--9 follow the same response rule. The balanced versions keep only parents who respond at every post-birth event time. Panel FE and Sym+FE fix the coefficients at $-3$ and $-2$.\end{minipage}
\end{table}

Table~\ref{tab:symmetric-sim} reports RMSE and bias relative to the average penalty of the full parent population in cohorts 5--9. In S3, selective nonresponse lowers the average penalty among respondents by 0.027, but all estimators using the unbalanced data have absolute biases below 0.02. Pseudo-OLS has a lower RMSE than Panel OLS in all three scenarios. On the common cohorts, NP has an RMSE of 0.015 in all three scenarios and a bias of $-0.001$ in S3. Balancing the panel excludes more than three quarters of parents and disproportionately retains those with smaller penalties. Both panel regressions then underestimate the full-population penalty by about 0.13.

\subsection{Computational Comparison}\label{sec:sim-computation}

Lastly, I compare the regression time of individual-level pseudo-OLS and cell-level WLS as the sample size and the number of matching categories increase. I modify the benchmark design so that first-birth probabilities depend only on age. The matching covariate has $K$ equally sized categories. At a given sample size, I vary $K$ while keeping each individual's fertility and outcome histories fixed, so that only the matching partition changes.

\begin{table}[htbp]
\centering
\caption{Regression time of the individual-level and cell-level regressions}\label{tab:timing}

\footnotesize\setlength{\tabcolsep}{5pt}\renewcommand{\arraystretch}{1.08}
\begin{tabular}{@{}rrrrrrr@{}}
\toprule
 &  & \multicolumn{2}{c}{Regression rows} & \multicolumn{2}{c}{Seconds} & Time ratio \\
\cmidrule(lr){3-4}\cmidrule(lr){5-6}
People $N$ & $K$ & Individual & Cell & Individual & Cell &  \\
\midrule
\multicolumn{7}{l}{\emph{A. Increasing people, $K=4$}} \\
10,000 & 4 & 43,840 & 5,744 & 0.044 & 0.024 & 1.8 \\
25,000 & 4 & 108,240 & 5,760 & 0.086 & 0.024 & 3.6 \\
50,000 & 4 & 214,320 & 5,760 & 0.158 & 0.024 & 6.6 \\
100,000 & 4 & 430,224 & 5,760 & 0.297 & 0.024 & 12.4 \\
200,000 & 4 & 865,664 & 5,760 & 0.612 & 0.024 & 25.5 \\
\midrule
\multicolumn{7}{l}{\emph{B. Increasing matching categories, $N=50{,}000$}} \\
50,000 & 2 & 214,320 & 2,880 & 0.157 & 0.019 & 8.5 \\
50,000 & 4 & 214,320 & 5,760 & 0.158 & 0.024 & 6.6 \\
50,000 & 8 & 214,320 & 11,520 & 0.155 & 0.035 & 4.4 \\
50,000 & 16 & 214,320 & 22,976 & 0.153 & 0.056 & 2.7 \\
50,000 & 32 & 214,320 & 44,896 & 0.154 & 0.096 & 1.6 \\
50,000 & 64 & 214,320 & 80,400 & 0.157 & 0.166 & 0.9 \\
\bottomrule
\end{tabular}
\begin{minipage}{\linewidth}\vspace{0.35em}\footnotesize\textit{Notes:} Medians of six runs on one thread. $N$ is the number of people, and the data contain $20N$ person-year records. $K$ is the number of categories of the matching covariate. Times cover the regressions for both genders and the scaling, and exclude data preparation and inference. The time ratio is individual over cell.\end{minipage}
\end{table}

Table~\ref{tab:timing} reports the number of observations used in each regression and the computation time. Panel A increases the number of people from 10,000 to 200,000 at $K=4$. Panel B varies $K$ from 2 to 64 at $N=50{,}000$. Both regressions use the window $[-5,10]$, the base period $-2$, and the same parents and donors. The computational advantage of WLS grows with sample size. At a fixed $K$, the number of cells remains nearly constant, while the number of individual observations increases with the number of parents. Finer matching categories create more cells and reduce this advantage. At $K=64$, the two regressions take about the same time.

\section{Empirical Application: Child Penalties in Korea}\label{sec:empirical}

I apply the proposed estimators to three Korean survey datasets and estimate the child penalties for employment and earnings. Korea provides an interesting setting. It has the largest gender wage gap among OECD countries \citep{oecd2026korea}, and offers both a long-running household panel and several repeated cross-sectional surveys. Korea is also the country with the largest child penalty and the largest gap between pseudo-panel and panel estimates in the validation of \citet[Figure~3]{kleven2025atlas}. 

For the panel validation, I use the 2006--2024 waves of the Korean Labor and Income Panel Study (KLIPS), an annual household panel run by the Korea Labor Institute. It follows household members aged 15 and over, recording their employment and earnings over time. I also use two repeated cross-sectional surveys. The Household Income and Expenditure Survey (HIES), conducted by Statistics Korea, collects information on household income and spending. The Household Finance and Living Conditions Survey (HFLC) interviews around 20,000 households annually about their assets, debts and income, as well as the employment of household members. The samples cover 2006--2025 for HIES and 2012--2025 for HFLC.

\subsection{Implementation}\label{sec:empirical-data}

In KLIPS, the first-birth year comes from the respondent's direct report and, when there is none, from the earliest birth year of an own child in the household roster. I exclude respondents who report different first-birth years across waves and those whose roster child was born at least two years before the reported first birth.\footnote{Dating the excluded respondents from the roster changes either penalty by at most 0.09 percentage points. Excluding respondents who reported no previous birth although an own child was listed in the household roster in the same wave changes the employment and earnings penalties by at most 0.51 and 0.80 percentage points, respectively.} Employment equals one for respondents who worked in the week before the interview or were temporarily absent from a job. Event time is the number of years since childbirth at the interview. For matching, I subtract event time from the survey year and the parent's age to obtain cohort year and age at birth. Earnings are annual pretax labor income for the preceding calendar year, deflated to 2020 prices using the consumer price index. Event time for earnings is therefore computed based on income year. The earnings sample includes respondents with zero labor income. The donor pool includes observations before the assigned first-birth year. It also includes respondents without an assigned birth year whose most recent fertility response, as of the interview, indicates no previous birth. 

HIES and HFLC record household heads and spouses. I infer the first-birth year from the age of the oldest unmarried child in the household roster, and use respondents with no such child as donors.\footnote{The HIES roster includes only co-resident children, whereas HFLC also records children living apart for work or study. In HIES, I also exclude donors with reported nonresident unmarried children. Birth years inferred from children's ages in whole years can be off by one year.} HFLC provides annual records with household identifiers. For HIES employment, I use the annual file, whereas for earnings I use monthly average labor income from the quarterly files. This is because the annual files contain no individual labor income from 2017 onward.\footnote{Household serial numbers do not link the two sets of files. Re-estimating the pseudo-OLS and DR employment penalties on the quarterly files changes them by less than two percentage points.} 

As in the simulations, the event window is $[-5,10]$ with base period $-2$. I use two years before birth as the base period because employment may already change during pregnancy in the year before birth. I study first births at ages 20--45, use pre-birth donors aged 15--45, and follow parents up to age 55. Following \citet{kleven2025child}, I use parent observations from five years after each survey begins for pseudo-OLS and DR: 2011 in KLIPS and HIES, and 2017 in HFLC. In KLIPS, I also restrict first births to 2011 onward. The KLIPS panel regression applies the same first-birth cohort restriction, keeps all pre-birth observations, and requires at least one outcome before and one after birth.

For matching, I use two covariate resolutions. \textit{Coarse} cells are defined by gender, age, survey year and marital status or spouse presence, and \textit{fine} cells further condition on education in four broad categories.\footnote{The categories are below elementary school, elementary or middle school, high school or two-year college, and university or above in KLIPS, and below high school, high school, college or university, and graduate school in HIES and HFLC. Observations without a matching variable are dropped.} In KLIPS, parents observed at birth for both outcomes are interviewed less than one year after the first birth. Table~\ref{tab:data-summary} reports the survey years, the numbers of household interviews, parents observed at birth and donor records, and the share of those parents with donors one year before birth. Nearly all parents are matched with coarse cells, and adding education to the cell definition lowers the match rate by 1 to 5 percentage points.

\begin{table}[t]
\centering
\caption{Data and estimation samples}\label{tab:data-summary}
\footnotesize
\begin{tabular}{llrrrr}
\toprule
Survey & Years & \shortstack{Household\\interviews/year} & \shortstack{Parents at\\birth (W/M)} & Donors & \shortstack{Match rate\\(coarse / fine)} \\
\midrule
KLIPS & 2006--2024 & 4,999--12,132 & 970 / 972 & 45,539 & 96.1\% / 92.0\% \\
HIES & 2006--2025 & 8,947--13,529 & 1,413 / 1,352 & 33,888 & 99.2\% / 95.6\% \\
HFLC & 2012--2025 & 17,863--19,744 & 820 / 793 & 33,099 & 99.6\% / 98.5\% \\
\bottomrule
\end{tabular}
\begin{minipage}{\linewidth}\vspace{0.3em}\footnotesize\textit{Notes:} Employment samples for pseudo-OLS and DR. Parent counts (women/men) and donor counts are observations pooled across years using coarse cells. For Sym and Sym+FE, the parent counts are 2,154/2,068 in HIES and 1,896/1,850 in HFLC; KLIPS is unchanged. The match rate is the share of parents observed at birth with at least one donor one year before birth.\end{minipage}
\end{table}

I apply the DR correction in KLIPS using additional covariates including detailed education (the survey's original attainment code with nine levels), urban residence, province and the spouse age gap. I use detailed education in HIES, and detailed education, metropolitan residence and the spouse age gap in HFLC. The correction is applied within coarse cells, because the median number of donors per fine cell is only 4 to 11 (Appendix Table~\ref{tab:dr-overlap}). Province is excluded from the propensity model because some province categories contain only parents or only donors within matching cells. 

Symmetric matching with type fixed effects (Sym+FE) weights each type by its number of parents observed at birth. Parents observed after birth are included only if their type is also represented at birth. If I excluded parents observed at birth in the first five years, all later observations of those first-birth cohorts would also be excluded. Together with other types not represented at birth, this would remove 45\% of the post-birth employment observations used by pseudo-OLS in HIES and 69\% in HFLC. For symmetric matching in HIES and HFLC I therefore include parents observed at birth or after birth from the first survey year onward. I impute pre-birth outcomes whenever matching donors are available. With coarse cells, symmetric matching then has 93--97\% as many post-birth records as pseudo-OLS and includes 77--82\% of the post-birth records used by pseudo-OLS. NP compares parents with childless donors of the same age and matching characteristics in the same survey year. I therefore extend the donor age limit to 55. NP excludes first-birth cohorts with fewer than two parents observed at birth. In KLIPS, both Sym+FE and NP use first births from 2011 onward, matching the panel's cohort restriction.

I report both unscaled and scaled penalties, averaged over event times $0$--$10$. Standard errors are computed using the analytical variance formulas in Propositions~\ref{prop:cell-clt} and~\ref{prop:dr-inference}. I define a cluster for inference as a set of respondents connected by person or household identifiers across waves in KLIPS, a household in HFLC, and a household--year for employment or a household--quarter for earnings in HIES.

\subsection{Estimation Results}\label{sec:empirical-results}\label{sec:empirical-klips}\label{sec:empirical-rcs}

\begin{table}[t]
\centering
\caption{KLIPS child penalties: estimation window $[-5,10]$}\label{tab:klips}\label{tab:balanced-sensitivity}

\footnotesize\setlength{\tabcolsep}{4pt}\renewcommand{\arraystretch}{1.08}
\begin{tabular}{@{}lrrrr@{}}
\toprule
 & \multicolumn{2}{c}{Employment} & \multicolumn{2}{c}{Annual earnings} \\
\cmidrule(lr){2-3}\cmidrule(lr){4-5}
Estimator & Unscaled (SE) & Scaled (SE) & Unscaled (SE) & Scaled (SE) \\
\midrule
\multicolumn{5}{l}{\emph{Actual-panel estimators, first births from 2011}} \\
Panel OLS & 0.389 (0.036) & 46.9\% (2.94) & 2,014 (190) & 65.6\% (4.96) \\
Balanced [-5,10] & 0.797 (0.229) & 66.3\% (7.94) & 3,282 (1,880) & 126.2\% (133.61) \\
Balanced [-2,10] & 0.747 (0.196) & 70.1\% (10.67) & 3,149 (1,458) & 114.3\% (79.75) \\
Balanced [-2,5] & 0.407 (0.060) & 48.5\% (4.34) & 2,201 (354) & 72.3\% (9.90) \\
\addlinespace
\multicolumn{5}{l}{\emph{Pseudo-panel estimators}} \\
Pseudo-OLS (fine) & 0.315 (0.027) & 42.1\% (2.53) & 1,861 (158) & 64.7\% (4.26) \\
Pseudo-OLS (coarse) & 0.291 (0.027) & 40.1\% (2.62) & 1,756 (164) & 62.2\% (4.46) \\
DR (coarse) & 0.296 (0.027) & 40.5\% (2.60) & 1,834 (155) & 64.1\% (4.22) \\
Sym+FE & 0.298 (0.077) & 41.1\% (6.90) & 1,862 (405) & 65.4\% (11.91) \\
NP (fine) & 0.266 (0.023) & 37.9\% (2.62) & 1,122 (197) & 46.9\% (5.33) \\
NP (coarse) & 0.240 (0.026) & 35.4\% (2.97) & 970 (187) & 42.5\% (5.32) \\
\bottomrule
\end{tabular}
\begin{minipage}{\linewidth}\vspace{0.3em}\footnotesize\textit{Notes:} First births from 2011. Men-minus-women penalties averaged over event times 0--10, with base period $-2$. Sym+FE also fixes the coefficient at $-3$. Unscaled employment penalties are differences in employment rates, and unscaled earnings penalties are in 10,000 won at 2020 prices. Standard errors of the scaled penalties are in percentage points. Balanced $[a,b]$ keeps parents observed at or before $a$ and at or after $b$. The four panel rows use 870/827, 116/115, 195/208 and 455/440 women/men for employment and 944/903, 111/107, 201/200 and 490/471 for earnings.\end{minipage}
\end{table}

Table~\ref{tab:klips} reports the estimated penalties in KLIPS. Pseudo-OLS with fine cells gives penalties close to the panel estimates for both outcomes.\footnote{I also estimate the panel regression with individual FE, which performs well in the simulations. Here it is imprecise and gives penalties of 37.4\% (SE 9.84) for employment and 27.9\% (SE 11.60) for earnings. I normalize the coefficients at $-3$ and $-2$ to zero. The estimates are sensitive to these restrictions, and the panel shows strong pre-birth trends for both men and women.} Relative to fine-cell pseudo-OLS, pseudo-OLS with coarse cells gives slightly lower penalties, while DR and Sym+FE give similar penalties for both outcomes. Sym+FE has a standard error about three times as large as those of the other estimators.\footnote{With type FE, the restrictions at $-3$ and $-2$ identify the linear trend in the event-time coefficients. Sampling error in the donor means at these dates and in the numbers of parents observed at birth used as weights affects the post-birth coefficients.} These estimates lie within the panel's 95\% confidence intervals, except for the employment estimates from coarse-cell pseudo-OLS and DR. NP employment penalties are close to the other estimates, although the fine- and coarse-cell estimates fall 3.2 and 5.7 percentage points below the lower bound of the panel's confidence interval. For earnings, NP estimates substantially lower penalties.\footnote{With fine cells, pseudo-OLS estimates that fathers' earnings rise by 16.8\% after the first birth and NP estimates 4.5\%. Mothers' losses are estimated at 47.9\% and 42.3\%, respectively. The gap for fathers arises because childless men of the same age in the same year earn more than the counterfactual that the regression extrapolates from pre-birth donors.} As in the simulations, balancing the panel can sharply reduce the sample and substantially change the estimate. Requiring $[-2,10]$ leaves only around 200 parents of each gender from the earliest cohorts, raises the employment penalty to 70.1\% and makes the earnings penalty insignificant. The shorter window $[-2,5]$ keeps about half of the parents and stays close to the unbalanced estimate.

\begin{table}[t]
\centering
\caption{Repeated-cross-section child penalties: $[-5,10]$}\label{tab:rcs-results}
\footnotesize
\begin{tabular}{lcccccc}
\toprule
 & \multicolumn{2}{c}{HFLC employment} & \multicolumn{2}{c}{HIES employment} & \multicolumn{2}{c}{HIES earnings} \\
\cmidrule(lr){2-3}\cmidrule(lr){4-5}\cmidrule(lr){6-7}
Estimator & Scaled & SE & Scaled & SE & Scaled & SE \\
\midrule
Pseudo-OLS (fine) & 32.0\% & 2.1 & 33.5\% & 1.5 & 49.7\% & 1.3 \\
Pseudo-OLS (coarse) & 30.9\% & 2.0 & 33.3\% & 1.4 & 51.1\% & 1.3 \\
DR (coarse) & 32.3\% & 2.0 & 34.3\% & 1.4 & 51.1\% & 1.3 \\
Sym+FE & 37.5\% & 5.2 & 35.1\% & 6.5 & 59.8\% & 4.4 \\
NP (fine) & 26.7\% & 2.0 & 32.6\% & 1.3 & 52.6\% & 1.6 \\
NP (coarse) & 27.1\% & 1.8 & 32.4\% & 1.3 & 56.0\% & 1.6 \\
\bottomrule
\end{tabular}
\begin{minipage}{\linewidth}\vspace{0.3em}\footnotesize\textit{Notes:} Scaled penalties in percent, averaged over event times 0--10, with analytical standard errors in percentage points. Sym+FE fixes the coefficients at $-3$ and $-2$. Sym+FE and NP use parents from the first survey year.\end{minipage}
\end{table}

For the two repeated cross-sections, Table~\ref{tab:rcs-results} reports the penalties. Pseudo-OLS and DR estimates for employment penalties are 31--34\% in both surveys, below the KLIPS estimates, and move by less than two percentage points with the cell specification or the DR adjustment. HIES earnings penalties are about 50\%. NP is close to the other estimates for HIES employment, 4 to 5 percentage points lower in HFLC and up to 5 percentage points higher than pseudo-OLS for HIES earnings. Sym+FE is again imprecise, with standard errors about two and a half to four and a half times those of pseudo-OLS.\footnote{Relative to Sym on the same coarse cells, the employment SE increases by factors of 2.6 in KLIPS, 5.0 in HIES and 3.0 in HFLC. HIES and HFLC have larger samples of parents observed at birth and smaller absolute SEs.}

Figure~\ref{fig:klips} plots the scaled event-time coefficients for the three surveys. In the KLIPS panel, women's employment falls by almost 50\% in the birth year and recovers only slightly over the next decade, while men's employment rises slightly after birth. The pseudo-panel estimators trace similar post-birth paths, but after the first few years they generally lie above the panel for both genders.\footnote{Averaged over years 0--10, pseudo-OLS estimates women's loss at 41.1\% against 46.2\% in the panel and men's gain at 1.0\% against 0.7\%.} Both shifts are upward, so they partly offset in the gender difference. The repeated cross-sections show similar patterns. Women's outcomes drop by 37--67\% in the birth year and then recover partly, faster than in KLIPS. The estimators agree for women and differ for men. Sym+FE estimates a larger post-birth rise for men, up to 8\% for HIES employment and 18\% for HIES earnings, which is why its penalty estimates are larger than the others.

Pseudo-panel methods estimate different pre-birth behaviors from the actual panel. In the KLIPS panel, men's employment rises over the five years before birth, while women's falls, by 24\% in the year before birth alone. Among parents whose year of marriage is observed in KLIPS, 56\% have their first child in the year of marriage or the year after, so for most parents the years before birth are the years of marriage and pregnancy. Men's employment rises from 92\% to 95\% around marriage, regardless of when the first birth occurs. Women's employment falls near marriage mainly when they have their first birth soon, from 81\% to 40\% for women who give birth in the year of marriage or before and from 87\% to 81\% for women who give birth two or more years later (Appendix~\ref{app:marriage}). The pseudo-panel shows no such patterns. Its donors are matched on the parents’ marital status at birth and are therefore mostly married. With fine cells, women’s imputed employment rate two years before birth is 67\%, already lower than the 76\% observed in the panel. 

\begin{figure}[p]
\centering
\includegraphics[width=\linewidth,height=.83\textheight,keepaspectratio]{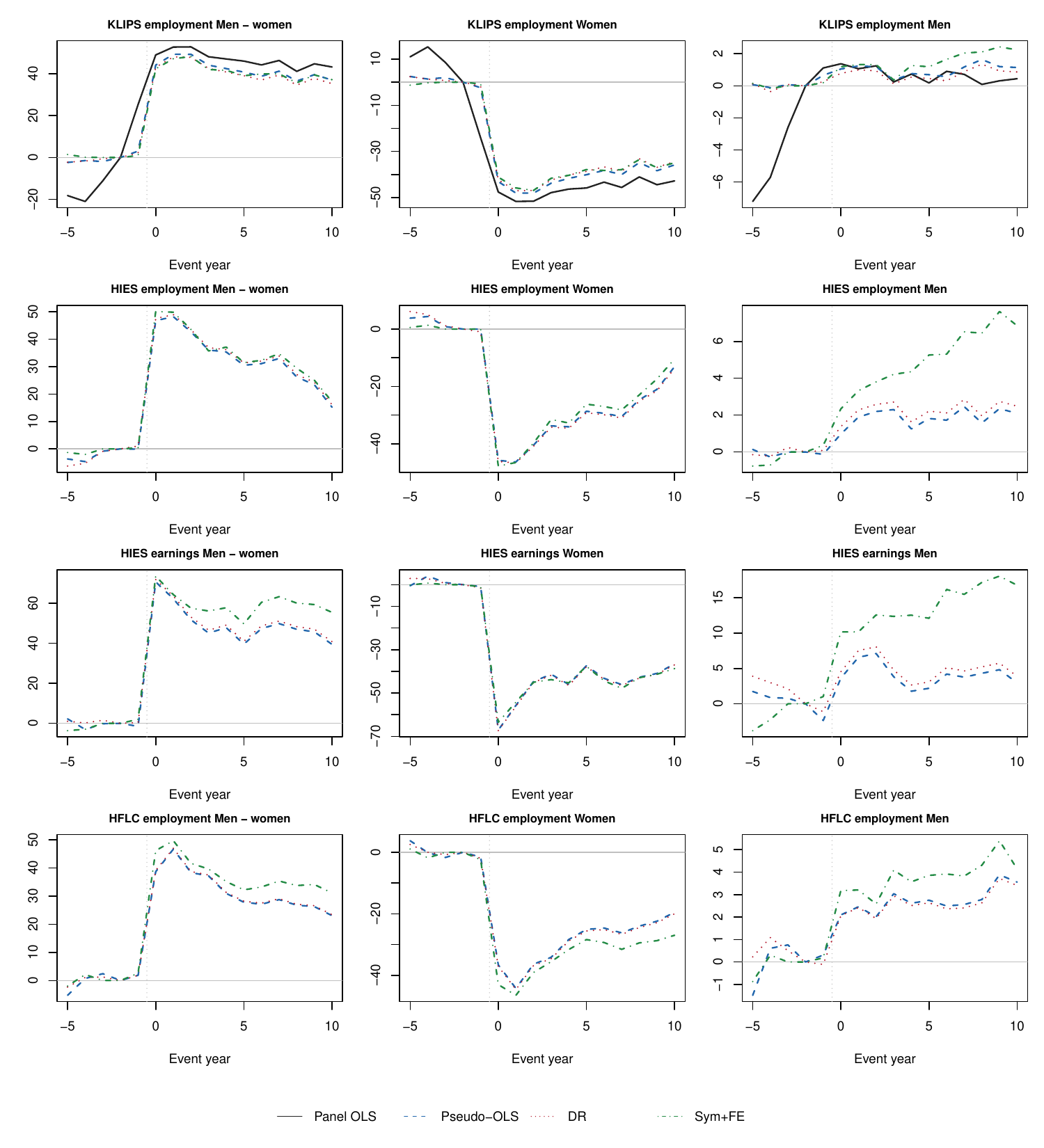}
\caption{Scaled event-time paths in the panel and the pseudo-panel. The window is $[-5,10]$ with base period $-2$, and Sym+FE also fixes $-3$. Pseudo-OLS uses fine cells, and DR and Sym+FE coarse cells.}\label{fig:klips}
\end{figure}

In sum, I revisit the country with the largest discrepancy between panel and pseudo-panel estimates in the validation of \citet{kleven2025atlas}. In my KLIPS sample, the pseudo-panel estimates are close to the average panel penalties.\footnote{My sample differs from theirs in the survey years, the first-birth cohorts and the employment measure, so my results are not directly comparable to theirs.} Women's employment in the panel begins to fall before birth, while the pseudo-panel places most of the decline at birth. Measured relative to two years before birth, the average post-birth penalties are nevertheless close to each other. The Japanese panel in \citet[Figure~3D]{kleven2025atlas} also shows similar average penalties despite differences in the pre-birth estimates.

\subsection{Robustness Analysis}\label{sec:empirical-robustness}

To assess the robustness of the headline results, I conduct seven additional analyses, each changing one element of the pseudo-panel implementation: the donor pool, the DR specification, the composition of parents, the event window, the cohorts used by NP, the survey weights and the method of inference. I also report the penalties by subgroup.

\emph{Donor pool.} The donor pool has two distinct groups, future parents and never-parents. One might suspect that future parents are a more comparable control group than never-parents. In Appendix~\ref{app:future-never}, using the KLIPS panel, I compare donors with an observed later birth and donors observed to remain childless through age 45. Within cells defined by age, year, marital status, education and urban residence, future fathers are 8.7 percentage points (SE 4.2) more likely to work and earn 5.0 million won (SE 2.2 million) more per year than men who remain childless, while the differences for women are statistically insignificant. Using fine cells and keeping the same parents and pre-birth event times, restricting donors to respondents with an observed later birth lowers the pseudo-OLS employment penalty from 42.1\% to 41.8\% and lowers the earnings penalty from 65.8\% to 65.7\%. The corresponding DR estimates fall from 41.5\% to 40.9\% and from 66.1\% to 65.4\%. All these changes are smaller than two percentage points.

\emph{DR specification.} The DR estimates are not sensitive to the model specification. Across five outcome and propensity specifications the KLIPS employment penalty stays at 40.5--40.7\%.\footnote{The specifications include a saturated or a cubic-age, linear-year outcome model with propensity scores on detailed education and urban residence, on broad education and urban residence, or on detailed education alone, and one uses the outcome model only.} Dropping the spouse age gap from both models changes the KLIPS penalties from 40.5\% to 40.8\% for employment and from 64.1\% to 63.0\% for earnings. Two-fold cross-fitting changes the penalties by 0.3 to 4.6 percentage points in KLIPS and by at most 1.5 points in HIES and HFLC on the same samples (Appendix~\ref{app:coarse-cell}).

\emph{Composition.} In pseudo-OLS, the composition of parents changes over event time. Table~\ref{tab:composition-restored} shows that, among the mothers it uses, the share who gave birth after age 30 falls from 66\% at birth to 57\% ten years later in KLIPS, from 62\% to 38\% in HIES and from 67\% to 46\% in HFLC. Weighting each post-birth cell by the number of parents of that type observed at birth reduces the fall to 7 percentage points in KLIPS, 9 in HIES and 6 in HFLC. Some types observed at birth have no observations ten years later, so symmetric matching does not fully preserve the composition at birth. Nevertheless, the penalty estimates differ little from pseudo-OLS (Appendix~\ref{app:bidirectional-comparisons}).

\begin{table}[htbp]
\centering
\caption{Composition of mothers by event time}\label{tab:composition-restored}
\footnotesize
\begin{tabular}{llrrrrrr}
\toprule
 & & \multicolumn{3}{c}{Mean age at birth} & \multicolumn{3}{c}{Share above age 30 at birth} \\
\cmidrule(lr){3-5}\cmidrule(lr){6-8}
Survey & Weights & 0 & 5 & 10 & 0 & 5 & 10 \\
\midrule
KLIPS & Pseudo-OLS & 32.1 & 31.6 & 31.1 & 66.2\% & 63.6\% & 56.9\% \\
 & Sym & 32.1 & 31.8 & 31.2 & 66.2\% & 64.7\% & 59.0\% \\
\addlinespace
HIES & Pseudo-OLS & 31.6 & 30.5 & 29.7 & 62.0\% & 47.9\% & 38.2\% \\
 & Sym & 30.9 & 30.4 & 30.0 & 52.6\% & 48.0\% & 43.2\% \\
\addlinespace
HFLC & Pseudo-OLS & 32.3 & 31.1 & 30.3 & 67.4\% & 56.4\% & 46.2\% \\
 & Sym & 31.5 & 31.1 & 30.8 & 60.5\% & 57.5\% & 54.1\% \\
\bottomrule
\end{tabular}
\begin{minipage}{\linewidth}\vspace{0.3em}\footnotesize\textit{Notes:} Women in the employment samples, coarse cells. Columns are event times. Pseudo-OLS weights the post-birth records that are observed, and Sym weights each type by its number of parents observed at birth.\end{minipage}
\end{table}

\emph{Event window.} Appendix~\ref{app:tau05} re-estimates every model on the window $[-2,5]$, except Sym+FE, which also requires one more pre-birth date. Employment penalties averaged over years 0--5 are larger than the ten-year averages, consistent with larger employment losses in the first years after birth. The KLIPS employment penalty is 51.3\% in the panel and 45.3\% in the standard pseudo-panel regression, and its HIES and HFLC employment penalties rise to 40.8\% and 37.6\%. The gap between NP and the regressions for KLIPS earnings shrinks from about 18 to about 5 percentage points, so the two counterfactuals for fathers diverge mainly at later event times.

\emph{NP cohorts.} For comparison with the panel, the main KLIPS NP estimates use first births from 2011 onward. Since NP requires no pre-birth donors, I also estimate it including parents observed at birth from 2006. Without the cohort restriction the coarse NP penalty is 34.5\% instead of 35.4\% for employment and 52.5\% instead of 42.5\% for earnings. With fine cells the penalties are 36.9\% instead of 37.9\% and 54.9\% instead of 46.9\%. 

\emph{Survey weights.} Using the survey weights in the donor means, the parent counts and the regressions changes the HIES and HFLC pseudo-OLS penalties by at most 1.3 percentage points (Appendix Table~\ref{tab:weighted-rcs}).

\emph{Inference.} I compare the analytical standard errors with the standard deviations of a cluster bootstrap that repeats the entire estimation procedure on each bootstrap sample (Appendix~\ref{app:resampling-restored}). For coarse cells, the bootstrap standard deviations are within 19\% of the analytical standard errors in most cases. For Sym+FE, they are 46\% larger for KLIPS employment, 26\% for KLIPS earnings and 35\% for HIES employment. For NP, they are 55\% larger for KLIPS earnings. The analytical standard errors use a first-order approximation that may miss some finite-sample variation, while bootstrap standard deviations can be sensitive to extreme estimates in relatively small samples.

\emph{Subgroups.} By education and age at first birth, the pseudo-panel employment penalties in KLIPS vary less than the panel estimates, and the panel penalty is 61.9\% for parents aged 30 or younger at birth against 44.6\% under pseudo-OLS, so the agreement in the aggregate penalty does not extend to every subgroup (Appendix~\ref{app:heterogeneity}).

\section{Conclusion}\label{sec:conclusion}

Pseudo-panel event studies are widely applicable to settings where only repeated cross-sections are available. Even where panel data are available, small samples and limited coverage over time or across population groups can restrict the analysis. Selective attrition can also introduce survivor bias. Larger repeated cross-sections, therefore, can support more detailed analyses. In this paper, I show that the pseudo-panel design identifies cohort-time treatment effects when untreated pre-birth outcomes can be imputed from matched childless donors to future parents. I compute the pseudo-panel estimates using weighted least squares with covariate-cell means and derive analytical standard errors that account for the repeated use of observations. In simulations, the resulting confidence intervals attain near-nominal coverage, while conventional intervals cover about half as often. I therefore recommend the cell-level WLS regression with analytical standard errors as the default inference option. 

The proposed extensions each address distinct weaknesses of the baseline approach. Doubly robust estimation accommodates continuous covariates and variables too detailed for exact matching. Symmetric matching limits changes in composition among the types that remain observed, and synthetic type FE absorb persistent differences in their untreated outcomes. When treatment effects are heterogeneous, the nonparametric estimator computes cohort- and type-specific effects and aggregates them with positive weights. In the Korean panel, the standard pseudo-panel regression gives an employment penalty close to the panel estimate. For applied work, I recommend reporting both unscaled and scaled penalties together with the event-study estimates for women and men. With pseudo-panels, pre-birth coefficients close to zero do not rule out anticipation among parents. Comparing analytical standard errors with standard deviations from a cluster bootstrap is particularly useful for the extensions such as Sym+FE and NP when estimates depend heavily on a few observations or small counterfactual denominators. 

The validity of the pseudo-panel design hinges on the conditional independence assumption, which rules out selection on unobservables. A fruitful avenue for future research would be a sensitivity analysis that bounds the penalty when this assumption is relaxed. The framework can also be applied to other settings such as divorce, retirement, and the onset of a disability, where surveys record past event dates but cannot reveal when untreated respondents will experience the event. The marriage application in Appendix~\ref{app:marriage} illustrates the importance of informative matching variables. Marital status helps account for differences between parents and childless donors, but cannot be used for matching when marriage itself is the event. Where panel data are available, examining pre-event paths and subgroup estimates can reveal differences that are not apparent in the average effects.

\bibliographystyle{chicago}
\bibliography{refs}

\clearpage
\pagenumbering{arabic}\renewcommand{\thepage}{A-\arabic{page}}
\appendix
\counterwithin{table}{section}\counterwithin{figure}{section}\numberwithin{equation}{section}
\counterwithin{theorem}{section}\counterwithin{prop}{section}\counterwithin{corollary}{section}\counterwithin{lemma}{section}
\counterwithin{assumption}{section}\counterwithin{definition}{section}\counterwithin{remark}{section}
\makeatletter
\@for\cnt:=table,figure,equation,theorem,prop,corollary,lemma,assumption,definition,remark\do{%
  \expandafter\edef\csname theH\cnt\endcsname{\noexpand\thesection.\noexpand\arabic{\cnt}}}
\makeatother

\begin{center}
{\Large\bf Supplementary Appendix}

\vspace{0.8em}
{\large for ``Pseudo-Panel Event Studies: Econometric Foundations for Credible Applications''}

\vspace{0.8em}
Dongwoo Kim (Simon Fraser University and Korea University)
\end{center}

\vspace{1.5em}

\FloatBarrier
\section{Technical Proofs}\label{app:proofs}

\subsection{Proof of Theorem~\ref{thm:identification} (Identification of ATT)}\label{app:proof-identification}

\begin{proof}
Take $(f,t)$ with $2 \le f \le t \le T$. Then
\begin{align*}
     \text{ATT}(f,t) &= \E[Y_t(f) - Y_t(\infty) \mid F = f]\\
    &= \E[Y_t(f) - Y_{f-1}(f) \mid F = f] + \E[Y_{f-1}(f) - Y_t(\infty) \mid F = f]\\
    &= \E[Y_t(f) - Y_{f-1}(f) \mid F = f] + \E[Y_{f-1}(\infty) - Y_t(\infty) \mid F = f]\\
    &= \E[Y_t(f) - Y_{f-1}(f) \mid F = f] + \E[Y_{f-1}(\infty) - Y_t(\infty) \mid F > t]\\
    &= \E[Y_t - Y_{f-1} \mid F = f] + \E[Y_{f-1} - Y_t \mid F > t]\\
    &= \E[Y_t \mid F = f] - \E[Y_t \mid F > t]
    - \E[Y_{f-1} \mid F = f] + \E[Y_{f-1} \mid F > t],
\end{align*}
where the third equality uses Assumption~\ref{assu:NA} averaged over $X$, the fourth uses Assumption~\ref{assu:PT} averaged over cohorts with $F>t$, and the fifth uses $Y_t=Y_t(F)$ together with Assumption~\ref{assu:NA} averaged over $X$ and those cohorts. From the fifth line on, all outcomes are realized outcomes. Thus parallel trends is used only at the cohort level. In the last line, the first two expectations are identified directly from the data in year $t$. The last two terms are not. To identify them, take any cohort $h\ge f$. Assumptions~\ref{assu:NA} and~\ref{assu:CI} give
\begin{align*}
\E[Y_{f-1} \mid F=h, X=x]
&= \E[Y_{f-1}(\infty) \mid F=h, X=x] \\
&= \E[Y_{f-1}(\infty) \mid F>f-1, X=x]
= \E[Y_{f-1} \mid F>f-1, X=x].
\end{align*}
The first and last equalities use conditional no anticipation, with the last equality also averaging over future cohorts. The middle equality uses conditional independence. Setting $h=f$ and integrating over the covariate distribution of cohort $f$ yields
\begin{align*}
    \E[Y_{f-1} \mid F = f] &= \E_x\!\left[\E[Y_{f-1} \mid F > f-1, \, X = x] \,\middle|\, F = f\right],
\end{align*}
where the right-hand side is identified (the inner expectation is identified from childless individuals in year $f-1$, and the outer expectation integrates over the distribution of $X$ conditional on $F = f$). Similarly, applying the display above to every cohort $h>t$ and averaging over these cohorts conditional on $X=x$ gives $\E[Y_{f-1}\mid F>t,X=x]=\E[Y_{f-1}\mid F>f-1,X=x]$, so
\begin{align*}
    \E[Y_{f-1} \mid F > t] &= \E_x\!\left[\E[Y_{f-1} \mid F > f-1, \, X = x] \,\middle|\, F > t\right],
\end{align*}
where the outer expectation is over the distribution of $X$ conditional on $F > t$, which is identified from the data at time $t$. Substituting the two expressions into the last line of the first display gives \eqref{eq:identification}, so $\text{ATT}(f,t)$ is identified.
\end{proof}

\subsection{Proof of Proposition~\ref{prop:oracle} (Convergence to Oracle OLS)}\label{app:proof-oracle}

\begin{proof}
By the WLS equivalence of Section~\ref{sec:estimation}, the pseudo-panel OLS estimator is
\[
\widehat{\mathbf b}^{\textup{OLS}} = \Big(\sum_{r\in\mathcal R^g}w_r\mathbf d_r\mathbf d_r'\Big)^{-1}\sum_{r\in\mathcal R^g}w_r\mathbf d_r\bar Y_r,
\]
with the cells $r = (X, t, \tau)$, cell means $\bar{Y}_r$, regressors $\mathbf d_r$ and weights $w_r$ of that section. I first show that each cell mean $\bar Y_r$ converges to the corresponding oracle mean $\mu^*_r$, the mean of $Y^*_{i\tau}$ in cell $r$ under the limiting distribution of regression observations.

\emph{Step 1: Cell-mean convergence.} Assumption~\ref{assu:lln}(i) applies to each cell numerator and denominator. For post-birth cells ($\tau \ge 0$), their ratio satisfies
\[
\bar{Y}_r \to_p \E[Y_{it} \mid X_i = X, \, F_i = t-\tau] = \E[Y_{it}(F_i) \mid X_i = X, \, F_i = t-\tau],
\]
which is the oracle post-birth population mean. For pre-birth cells ($\tau < 0$), $\bar{Y}_r = \bar{Y}^C_{(X, t)}$ is the donor-cell mean of the $N^C_{(X,t)}$ childless individuals observed in year $t$, so
\[
\bar{Y}_r \to_p \E[Y_{it} \mid X_i = X, \, F_i > t].
\]
By Assumptions~\ref{assu:NA} and~\ref{assu:CI}, this equals $\E[Y_{it}(\infty) \mid X_i = X, \, F_i = f]$ for any $f > t$ (equation~\eqref{eq:CI-cell}). In particular, it equals the untreated mean for parents with $F_i = t - \tau$, whose pre-birth outcomes are imputed from this donor cell. By Assumption~\ref{assu:lln}(i), the oracle pre-birth mean among these parents observed at birth converges to the same limit.

\emph{Step 2: WLS convergence.} Since $X$ includes year of birth, the regressors $\mathbf d_r$ are fixed on the limiting set of included cells. With $M=\sum_rw_r$, Assumption~\ref{assu:lln}(i) gives $w_r/M\to_p p_r$, so $M^{-1}\sum_rw_r\mathbf d_r\mathbf d_r'$ converges to $Q^*=\sum_rp_r\mathbf d_r\mathbf d_r'$. By Step~1, $\bar Y_r\to_p\mu^*_r$ for every cell. The continuous mapping theorem gives
\[
\widehat{\mathbf b}^{\textup{OLS}} \to_p (Q^*)^{-1}\sum_rp_r\mathbf d_r\mu_r^*
=\left(\E[\mathbf d_{i\tau}\mathbf d_{i\tau}']\right)^{-1}\E[\mathbf d_{i\tau}Y^*_{i\tau}],
\]
where both expectations use the limiting distribution of regression observations with cell probabilities $p_r$. Replacing $\bar Y_r$ by the oracle sample cell means gives the same limit, since these means have the same limits and the regressors and weights are unchanged.
\end{proof}

\subsection{Proof of Proposition~\ref{prop:cell-clt} (Asymptotic Normality)}\label{app:proof-cell-clt}

\begin{proof}
By the multivariate i.i.d.\ CLT, $\sqrt N(\hat\eta-\eta)\Rightarrow N(0,\Sigma)$. On the event that the set of included cells equals its limit, cell means, type shares and row weights are ratios with positive limiting denominators. The limiting design matrix $Q^*$ is nonsingular. With nonzero population counterfactual denominators, the composition $G$ is continuously differentiable near $\eta$. The delta method gives
\[
\sqrt N\{G(\hat\eta)-G(\eta)\}
 =J\sqrt N(\hat\eta-\eta)+o_p(1)\Rightarrow N(0,J\Sigma J').
\]
This event has probability tending to one, so the result also holds for the rule that selects the cells in the sample. Partitioning $J$ and $\Sigma$ yields \eqref{eq:two-sample-var}. Since $H_\ell$ contains the cells of both genders, $\Sigma$ includes the covariance between the men's and the women's cell means, and $V$ is the variance of the men-minus-women penalty.

For the covariance estimator, $DG$ is continuous at $\eta$ and $\hat\eta\to_p\eta$, so $\hat J\to_p J$, and the sample covariance $\hat\Sigma$ of $H_\ell$ converges to $\Sigma$ under finite second moments. Since \eqref{eq:analytical-block-cov} equals $\hat J\hat\Sigma\hat J'/N$, $N\widehat{\operatorname{Var}}(\hat P)\to_p V$.
\end{proof}

\begin{remark}[Bootstrap]\label{rem:bootstrap}
The i.i.d.\ cluster bootstrap consistently estimates the distribution of $\sqrt N(\hat\eta-\eta)$, and the delta method for the bootstrap gives the same limiting distribution for the recomputed statistic. This requires that the bootstrap samples select the limiting set of cells with conditional probability tending to one. For the recomputed bootstrap estimate $\hat P^{\mathrm{boot}}$, consistency of its standard deviation additionally requires conditional uniform integrability in probability of $N\|\hat P^{\mathrm{boot}}-\hat P\|^2$. This can fail for a scaled penalty whose denominator is close to zero in some bootstrap samples.
\end{remark}

\begin{remark}[Fixed annual sample sizes]\label{rem:fixed-year-sampling}
Suppose each cluster belongs to a single survey year. Replace Assumption~\ref{assu:reg}(i) by i.i.d.\ sampling within each year, with finite second moments, independence across years and the number of clusters $N_s$ in each year fixed by design \citep[see also][]{santanna2020doubly}. Write $N=\sum_sN_s$, $N_s/N\to\lambda_s>0$, $\mu_s=\E[H_\ell\mid s]$, $\Sigma_s=\operatorname{Var}(H_\ell\mid s)$, $\eta=\sum_s\lambda_s\mu_s$ and $\eta_N=\sum_s(N_s/N)\mu_s$. Under the remaining conditions of Proposition~\ref{prop:cell-clt}, the within-year central limit theorems and the delta method give
\[
\sqrt N\{\hat P-G(\eta_N)\}\ \xrightarrow{d}\ N(0,V_{\mathrm{fixed}}),\qquad
V_{\mathrm{fixed}}=J\!\left(\sum_s\lambda_s\Sigma_s\right)J',
\]
where $J=DG(\eta)$.\footnote{Centering at $P^*=G(\eta)$ is valid if $N_s/N-\lambda_s=o(N^{-1/2})$ for every $s$. Using $G(\eta_N)$ as the center does not require this additional rate condition.} Let $\hat\Sigma_s$ be the sample covariance of $H_\ell$ within year $s$, using divisor $N_s-1$. A consistent variance estimator is
\[
\widehat{\operatorname{Var}}_{\mathrm{fixed}}(\hat P)
=\frac1N\hat J\!\left(\sum_s\frac{N_s}{N}\hat\Sigma_s\right)\hat J',
\]
with $N\widehat{\operatorname{Var}}_{\mathrm{fixed}}(\hat P)\to_p V_{\mathrm{fixed}}$.

The pooled covariance additionally includes $\sum_s\lambda_s(\mu_s-\eta)(\mu_s-\eta)'$, a positive semidefinite term arising from random annual sample sizes. Under the bootstrap conditions in Remark~\ref{rem:bootstrap}, the corresponding bootstrap draws $N_s$ clusters with replacement within each year and recomputes the statistic.
\end{remark}

\subsection{Derivatives for the Analytical Standard Errors}\label{app:analytical-implementation}

Here I present the derivatives and cluster influences used in Propositions~\ref{prop:cell-clt} and~\ref{prop:dr-inference} to compute the covariance estimator \eqref{eq:analytical-block-cov}. The influence of a cluster includes the contributions of all its observations. A person who is observed in several years, first as a donor and later as a parent, or whose observations also enter the estimation of the DR outcome and propensity models, contributes one term for each of these uses. The derivatives treat the set of cells used in the regression as fixed. Parts (ii) and (iv) of Assumption~\ref{assu:reg} justify this, because the sample selects the limiting set of cells with probability tending to one.

\subsubsection{Regression coefficients and scaled penalties}

Suppress gender and let $\mathbf d_r$ collect the regressors for the estimator under consideration, $\mathbf b$ its coefficient vector, $w_r$ the weight of row $r$, $\mu_r$ its outcome mean, and $Q=\sum_rw_r\mathbf d_r\mathbf d_r'$. Let $e_\tau$ select the event coefficient $\beta_\tau=e_\tau'\mathbf b$, with $e_\tau=0$ for an event time whose coefficient is fixed at zero. Differentiating the first-order condition of the weighted regression, $\sum_rw_r\mathbf d_r(\mu_r-\mathbf d_r'\mathbf b)=0$, gives
\[
d\mathbf b=Q^{-1}\sum_r\mathbf d_r\{w_r\,d\mu_r+(\mu_r-\mathbf d_r'\mathbf b)\,dw_r\}.
\]
This expression permits a nonzero population regression residual. With $W_\tau=\sum_{r:\tau_r=\tau}w_r$ and $\bar{\mathbf d}_\tau=W_\tau^{-1}\sum_{r:\tau_r=\tau}w_r\mathbf d_r$, the counterfactual denominator is $c_\tau=(\bar{\mathbf d}_\tau-e_\tau)'\mathbf b$, and
\begin{align*}
dc_\tau&=(\bar{\mathbf d}_\tau-e_\tau)'d\mathbf b
+\frac{1}{W_\tau}\sum_{r:\tau_r=\tau}
 (\mathbf d_r'\mathbf b-e_\tau'\mathbf b-c_\tau)\,dw_r,\\
d(\beta_\tau/c_\tau)&=\frac{e_\tau'd\mathbf b}{c_\tau}
-\frac{\beta_\tau}{c_\tau^2}\,dc_\tau.
\end{align*}
The second term of $dc_\tau$ reflects the change in the weights with which the estimated counterfactual is averaged. Holding the denominator fixed omits both terms.

The expressions above give the change in a coefficient or a scaled penalty for given changes $d\mu_r$ and $dw_r$ in the cell means and row weights. It remains to express these changes in terms of the observations. Let $f_i$ denote the number of times each observation $i$ is counted. In the observed sample $f_i=1$ for all $i$, and in a bootstrap sample $f_i$ is the number of times the observation is drawn. The derivative of the estimate with respect to $f_i$ at $f_i=1$ measures the influence of observation $i$. Let $a_i$ be the survey weight of the observation, which is held fixed and equals one in the unweighted estimates. The cell mean is $\mu_r=\sum_{i\in r}f_ia_iY_i/\sum_{i\in r}f_ia_i$, where the sums run over the observations contributing to this mean. For such an observation $i$,
\[
\frac{\partial\mu_r}{\partial f_i}
=\frac{a_i(Y_i-\mu_r)}{\sum_{j\in r}f_ja_j}.
\]
The mean derivative is zero for observations not contributing to the mean. The derivative of a regression weight $w_r$ with respect to $f_i$ is $a_i$ when observation $i$ is one of the parents counted in $w_r$, and zero otherwise. Setting $d\mu_r=\partial\mu_r/\partial f_i$ and $dw_r=\partial w_r/\partial f_i$ in the expressions above gives $\partial\hat P/\partial f_i$, the influence of observation $i$ on the penalty. The cluster-level influence in \eqref{eq:analytical-block-cov} is $\hat\psi_\ell=N\sum_{i\in\ell}\partial\hat P/\partial f_i$. Since $\hat P$ does not change when all $f_i$ are multiplied by a common constant, these influences sum to zero. For Sym+FE, $\mathbf d_r$ contains the event-time dummies other than those at $-3$ and $-2$ and the retained type, age and year indicators specified in Assumption~\ref{assu:sym-reg}, and the same formulas apply.

\subsubsection{Nonparametric estimator}

Suppress gender and fix an event time $\tau$. In the application the types included can differ across event times. Let $q_{f\theta}$ be the count of parents of cohort $f$ and type $\theta$ observed at birth and $K_f=\sum_\theta q_{f\theta}$ the count of the cohort. Let $\mathcal S_{f\tau}$ be the types of cohort $f$ with both a parent cell and a donor cell at $\tau$, $U_{f\tau}=\sum_{\theta\in\mathcal S_{f\tau}}q_{f\theta}$ the number of parents observed at birth in those types, $\mathcal F_\tau$ the cohorts with nonempty $\mathcal S_{f\tau}$, and $K_\tau=\sum_{f\in\mathcal F_\tau}K_f$. The estimator averages the included types within each cohort and then the cohorts in $\mathcal F_\tau$. With $\Delta_{f\theta,\tau}=\mu^P_{f\theta,\tau}-\mu^D_{f\theta,\tau}$, the difference between the parent and donor means,
\[
\widehat{\mathrm{ATT}}(\tau)=\sum_{f\in\mathcal F_\tau}\sum_{\theta\in\mathcal S_{f\tau}}\omega_{f\theta,\tau}\Delta_{f\theta,\tau},
\qquad
\omega_{f\theta,\tau}=\frac{K_f}{K_\tau}\frac{q_{f\theta}}{U_{f\tau}}.
\]
Its derivative with respect to the cell means is therefore $\sum_{f\in\mathcal F_\tau}\sum_{\theta\in\mathcal S_{f\tau}}\omega_{f\theta,\tau}(d\mu^P_{f\theta,\tau}-d\mu^D_{f\theta,\tau})$. 

The weights are estimated as well, because they are ratios of the counts of parents observed at birth, so the delta method also needs the derivative with respect to these counts. Write $\Delta_{f\tau}=\sum_{\theta\in\mathcal S_{f\tau}}q_{f\theta}\Delta_{f\theta,\tau}/U_{f\tau}$ for the cohort average, so that $\widehat{\mathrm{ATT}}(\tau)=\sum_{f\in\mathcal F_\tau}K_f\Delta_{f\tau}/K_\tau$. For a type of a cohort in $\mathcal F_\tau$, the derivative with respect to its count is
\[
\frac{\partial\widehat{\mathrm{ATT}}(\tau)}{\partial q_{f\theta}}
=\frac{\Delta_{f\tau}-\widehat{\mathrm{ATT}}(\tau)}{K_\tau}
+\mathbf1\{\theta\in\mathcal S_{f\tau}\}
 \frac{K_f}{K_\tau U_{f\tau}}(\Delta_{f\theta,\tau}-\Delta_{f\tau}).
\]
For the unscaled ATT, this derivative is zero when the differences $\Delta_{f\theta,\tau}$ are the same across all cohorts and types. A type absent at $\tau$ still affects the weight of its cohort, and its absence removes only the second term. The denominator of the scaled effect uses the same formulas with $\Delta_{f\theta,\tau}$ replaced by $\mu^D_{f\theta,\tau}$. The derivative of the scaled penalty then follows from the ratio derivative, the difference between genders and the average over event times. The NP estimator combines cell means and weights through finite sums and ratios with nonzero denominators. It is therefore a continuously differentiable function $G$ of the cluster moments $\hat\eta$, and Proposition~\ref{prop:cell-clt} applies. When every type is included at every event time, $U_{f\tau}=K_f$ and the weights reduce to $\hat q^g(f,\theta)$ of Section~\ref{sec:np}.

\subsubsection{DR estimator}\label{app:dr-variance}

Fix a gender and a pre-birth event time. For the implemented linear outcome model and logit propensity model, both models are estimated on all matching cells jointly, with a separate intercept for each cell and slopes common across cells. Let $v$ and $z$ denote the covariates of the outcome model (OM) and of the propensity score model (PS) that vary within a cell. The OM for the donors of cell $c$ is $m_c(v)=\iota_c+v'\alpha$. The PS is a logit for being a parent, with cell intercepts and slopes $\gamma$, so the weight of donor $j$ within its cell is proportional to $\exp(z_j'\gamma)$. For a pre-birth row $r$ with donor cell $c$, let $\bar v^P_r$ be the mean of $v$ over the parents observed at birth whose pre-birth rows form $r$. Since the weights are normalized within the cell, the cell intercepts of both models cancel from the imputation:
\[
\hat\mu^{DR}_r
=\bar v^P_r{}'\hat\alpha+
 \sum_{j\in\mathcal D_c}\tilde w_j(Y_j-v_j'\hat\alpha),
\qquad
\tilde w_j=\frac{a_j\exp(z_j'\hat\gamma)}
 {\sum_{k\in\mathcal D_c}a_k\exp(z_k'\hat\gamma)}.
\]
The direct parent contribution comes from differentiating $\bar v^P_r$. The direct donor contribution is the normalized weighted residual centered at its weighted mean. With $\bar z^{\tilde w}=\sum_{j\in\mathcal D_c}\tilde w_j z_j$, the nuisance derivatives are
\[
\partial_\alpha\hat\mu^{DR}_r
=\bar v^P_r-\sum_j\tilde w_jv_j,
\qquad
\partial_\gamma\hat\mu^{DR}_r
=\sum_j\tilde w_j(z_j-\bar z^{\tilde w})
 (Y_j-v_j'\hat\alpha).
\]
These terms enter the regression and scaling derivatives given above.

The two derivatives multiply the influence functions of $\hat\alpha$ and $\hat\gamma$, which are those of a regression and of a logit with cell intercepts. The intercepts need not be handled separately. For the OM, subtract from $v_j$ its $a_j$-weighted mean over the donors of the cell and call the result $\tilde v_j$. The Hessian of the slopes is $\sum_ja_j\tilde v_j\tilde v_j'$ and the score of donor $j$ is $a_j\tilde v_j(Y_j-\hat m_j)$, where $\hat m_j$ is the estimated donor mean. For the logit, let $S_i=1$ for parents and $S_i=0$ for donors, let $\hat e_i$ be the estimated probability, and subtract from $z_i$ its cell mean weighted by $a_i\hat e_i(1-\hat e_i)$ to obtain $\tilde z_i$. The Hessian of the slopes is $\sum_i a_i\hat e_i(1-\hat e_i)\tilde z_i\tilde z_i'$ and the score of observation $i$ is $a_i\tilde z_i(S_i-\hat e_i)$, with sums over parents and donors. This within-cell centering is analogous to partialling out fixed effects in a linear regression. It accounts for estimation of the cell intercepts while requiring only the smaller matrices for the slope coefficients to be inverted.

\subsection{Proof of Proposition~\ref{prop:dr} (Double Robustness and Efficiency)}\label{app:proof-dr}

\begin{proof}
Fix the gender $g$, first-birth cohort $f$ and pre-birth event time $\tau$ as in Section~\ref{sec:dr}, and suppress these indices, so that $\mu^0=\mu^0_{g,\tau}$, $m=m_g$ and $Y$ is the donor outcome in the donor year. Write $m^\dagger,w^\dagger$ for the probability limits of the estimated functions. Under the stated laws of large numbers and a positive denominator, the average of the imputations converges to
\[
E_1[m^\dagger(Z)]+\frac{E_0[w^\dagger(Z)\{m(Z)-m^\dagger(Z)\}]}{E_0[w^\dagger(Z)]}.
\]
The donor disturbance $Y-m(Z)$ has conditional mean zero, so its weighted average vanishes in the limit. Conditional NA and CI at $Z$ give $m(Z)=E[Y(\infty)\mid S=1,Z]$ in the donor year, hence $\mu^0=E_1[m(Z)]$.

If $m^\dagger=m$, the correction is zero for any admissible limiting weights. Consistency of the OM at each parent's covariates also gives consistency of the parent-specific imputation. If the odds are correct, Bayes' rule gives
\[
\frac{E_0[w(Z)h(Z)]}{E_0[w(Z)]}=E_1[h(Z)]
\]
for each integrable $h$. Taking $h=m-m^\dagger$ cancels the OM error in the average. The second argument concerns the average only and does not require consistency of each parent's imputation, which need not hold here. Together the two arguments prove parts~(a) and~(b).

For part~(c), recall the sampling scheme of the proposition. The parents' covariates $\{Z_i:S_i=1\}$ are an i.i.d.\ sample of size $n_1$ from a distribution $P_1$, and the donor observations $\{(Y_j,Z_j):S_j=0\}$ are an i.i.d.\ sample of size $n_0$ from a distribution $P_0$. The two samples are independent, and $n_s/N\to\lambda_s$. Write $m(z)=E_0[Y\mid Z=z]$, $\rho(z)=dP_1^Z/dP_0^Z(z)$ and $\mu=E_1[m(Z)]=\mu^0$. Assume $\rho$ is bounded and $E_1[m^2]+E_0[\rho^2(Y-m)^2]<\infty$. The two distributions are otherwise unrestricted. The maintained NA and CI restrictions give $\mu$ its untreated-outcome interpretation.

The estimator attains the semiparametric efficiency bound $V^{\mathrm{eff}}_{g,\tau}$ of \eqref{eq:dr-eff-bound} under the following conditions. Write the estimated functions as $m_{\hat\alpha}$ and $\rho_{\hat\gamma}$ in fixed-dimensional families, with $m_{\alpha_0}=m$ and $\rho_{\gamma_0}=\rho$. Estimated odds may be rescaled by the known sampling-odds constant to estimate $\rho$, and this common factor leaves the normalized weighted average unchanged. Assume the nuisance estimators are regular under local perturbations of the two distributions and $\hat\xi-\xi_0=O_p(N^{-1/2})$, where $\xi=(\alpha',\gamma')'$. In a neighborhood of $\xi_0$, require $\rho_\gamma>0$ and twice continuous differentiability of $m_\alpha$, $\rho_\gamma$ and $\rho_\gamma(Y-m_\alpha)$. Their values and first derivatives have square-integrable envelopes under the relevant distributions, and second derivatives have integrable envelopes. These conditions provide a local uniform law of large numbers and a bounded second derivative for the sample ratio. They concern regular fixed-dimensional estimators at distributions where both models are correct, including full-sample estimation.

Let $\mathbb P_s$ denote the empirical average in sample $s$. The statistic is
\[
T_N(\hat\xi)=\mathbb P_1m_{\hat\alpha}
+\frac{\mathbb P_0[\rho_{\hat\gamma}(Y-m_{\hat\alpha})]}
{\mathbb P_0\rho_{\hat\gamma}}.
\]
At the true nuisance functions, $E_0\rho=1$ and $E_0[\rho(Y-m)]=0$. The ordinary sample CLT therefore gives $\mathbb P_0\rho-1=O_p(N^{-1/2})$ and $\mathbb P_0[\rho(Y-m)]=O_p(N^{-1/2})$. Expanding only this oracle ratio yields
\[
T_N(\xi_0)-\mu=(\mathbb P_1-P_1)m
+\mathbb P_0[\rho(Y-m)]+O_p(N^{-1}).
\]
These root-$N$ statements apply to fixed true functions. To handle estimated functions, differentiate $T_N$ with respect to the nuisance parameters. Dots denote derivatives at $\xi_0$, and $\varepsilon=Y-m(Z)$:
\begin{align*}
\partial_\alpha T_N(\xi_0)
&=\mathbb P_1\dot m-\frac{\mathbb P_0(\rho\dot m)}{\mathbb P_0\rho},\\
\partial_\gamma T_N(\xi_0)
&=\frac{\mathbb P_0(\dot \rho\varepsilon)}{\mathbb P_0\rho}
-\frac{\mathbb P_0(\rho\varepsilon)\mathbb P_0\dot \rho}{(\mathbb P_0\rho)^2}.
\end{align*}
Both derivatives have zero population limits: $E_1\dot m=E_0(\rho\dot m)$ and $E_0(\dot \rho\varepsilon)=0$. The moment conditions give $\partial_\xi T_N(\xi_0)=O_p(N^{-1/2})$. A Taylor expansion with the regular root-$N$ nuisance estimates and bounded Hessian gives $T_N(\hat\xi)-T_N(\xi_0)=O_p(N^{-1})$. This holds when the same observations are used to estimate the nuisance parameters and to compute $T_N$, so no sample splitting is needed. Consequently,
\begin{equation}\label{eq:dr-eff-expansion}
\hat\mu^{DR}-\mu
=\frac1{n_1}\sum_{i=1}^{n_1}\{m(Z_i)-\mu\}
+\frac1{n_0}\sum_{j=1}^{n_0}\rho(Z_j)\{Y_j-m(Z_j)\}
+o_p(N^{-1/2}).
\end{equation}
Independence of the two samples gives the normal limit and variance in \eqref{eq:dr-eff-bound}.

To establish efficiency, consider regular submodels with mean-zero scores $h_1(Z)$ and $h_0(Y,Z)$. Differentiating $\mu=E_1[m(Z)]$ along these directions gives
\[
\dot\mu(h_1,h_0)
=E_1[(m-\mu)h_1]+E_0[\rho(Y-m)h_0].
\]
The donor term follows from $\dot m(z)=E_0[(Y-m)h_0\mid Z=z]$. Since the two distributions are unrestricted, their tangent spaces are the full mean-zero square-integrable spaces. The influence pair $(m-\mu,\,\rho(Y-m))$ is therefore the canonical gradient in the two-sample expansion. Under total-sample normalization its variance is $V^{\mathrm{eff}}_{g,\tau}$. Expansion~\eqref{eq:dr-eff-expansion} shows that, when the outcome model and the propensity model are both correctly specified, the influence function of $\hat\mu^{DR}$ is this canonical gradient. Together with the stated local regularity, this implies that the asymptotic variance of $\hat\mu^{DR}$ equals $V^{\mathrm{eff}}_{g,\tau}$.
\end{proof}

\subsection{Proof of Proposition~\ref{prop:dr-inference} (Parametric DR Inference)}\label{app:proof-dr-inference}

The additional regularity conditions for Proposition~\ref{prop:dr-inference} are as follows. In an open neighborhood of $\zeta^*$, the estimating functions $\Psi_\ell(\zeta)$ are continuously differentiable, and their values and first derivatives have square-integrable envelopes. The sample derivative satisfies a uniform law of large numbers on that neighborhood. For scaled penalties, differentiability of $G_{DR}$ requires nonzero limiting counterfactual denominators.

\begin{proof}
Expand the sample estimating equation around $\zeta^*$. Consistency, the derivative law of large numbers and nonsingularity give
\[
\sqrt N(\hat\zeta-\zeta^*)
=-A^{-1}N^{-1/2}\sum_\ell\Psi_\ell(\zeta^*)+o_p(1).
\]
The CLT for independent clusters and the delta method for $G_{DR}$ yield the stated influence representation and normal limit, with variance $E[\psi^{DR}_\ell\psi^{DR}_\ell{}']$. The local moment envelopes imply convergence of the plug-in score covariance and derivative matrices. Therefore $N$ times the covariance estimator \eqref{eq:analytical-block-cov} converges in probability to this variance.
\end{proof}

\subsection{Proof of Proposition~\ref{prop:sym-fe-oracle} (Oracle Projection)}\label{app:proof-sym-fe-oracle}

\begin{proof}
Let $\bar{\mathbf Y}^{B,g}$ stack the means of the included cells and $\boldsymbol\mu^{S,g,*}$ the population means $\mu^{S,g,*}_r$. For a pre-birth row, the donor-cell mean converges to the mean outcome of the childless in that cell. Assumption~\ref{assu:NA} makes this an untreated mean, and Assumption~\ref{assu:CI} equates it to the untreated mean of the parents of that type. For a post-birth row, the mean of the observed parents converges to the treated mean of the type by Assumption~\ref{assu:sym-repr}. With the law of large numbers and the fixed set of cells of Assumption~\ref{assu:sym-reg}, $\bar{\mathbf Y}^{B,g}\to_p\boldsymbol\mu^{S,g,*}$. Let $D_g$ have rows $(\mathbf d_r^{S,g})'$, retaining the event-time indicators and omitting only redundant fixed-effect indicators as specified in Assumption~\ref{assu:sym-reg}. Define $\Pi_{N,g}=\operatorname{diag}(N^P_{\theta(r),g,0}/N)$ and $\Pi_g=\operatorname{diag}(p^g_{\theta(r)})$. The rank condition of Assumption~\ref{assu:sym-reg} gives
\[
\widehat{\mathbf b}=(D_g'\Pi_{N,g}D_g)^{-1}D_g'\Pi_{N,g}\bar{\mathbf Y}^{B,g}
\to_p(D_g'\Pi_gD_g)^{-1}D_g'\Pi_g\boldsymbol\mu^{S,g,*},
\]
which solves \eqref{eq:anchored-synth-proj}. The cell means estimated by the regression are $D_g\widehat{\mathbf b}$. Omitting different redundant fixed-effect indicators leaves these estimates and the event-time coefficients unchanged: the retained fixed-effect indicators reproduce the same age, year and type effects, while the event-time indicators are unchanged.
\end{proof}

\subsection{Proof of Proposition~\ref{prop:conditional ATT} (Conditional ATT)}\label{app:proof-np}

\begin{proof}
Fix $(f,t,x)$ with $2 \le f \le t \le T$ and $x$ in the common support. Since $Y_t=Y_t(f)$ for cohort $f$,
\[
\textup{ATT}_x(f,t) = \E[Y_t \mid F = f, X = x] - \E[Y_t(\infty) \mid F = f, X = x].
\]
Assumption~\ref{assu:CPT} gives
\begin{align*}
\E[Y_t(\infty) \mid F = f, X = x]
&= \E[Y_{f-1}(\infty) \mid F = f, X = x]\\
&\quad+ \E[Y_t(\infty) - Y_{f-1}(\infty) \mid F > t, X = x].
\end{align*}
By Assumption~\ref{assu:CI}, $\E[Y_{f-1}(\infty) \mid F = f, X = x]=\E[Y_{f-1}(\infty) \mid F > t, X = x]$, and therefore
\[
\E[Y_t(\infty) \mid F = f, X = x] = \E[Y_t(\infty) \mid F > t, X = x].
\]
Finally, Assumption~\ref{assu:NA} for the cohorts after $t$ gives $\E[Y_t(\infty) \mid F > t, X = x]=\E[Y_t \mid F > t, X = x]$, which proves \eqref{eq:np-identification}.
\end{proof}

\FloatBarrier
\section{Symmetric Matching without Type FE}\label{app:exact-sym}\label{app:bidirectional-comparisons}

Section~\ref{sec:symmetric} proposes the symmetric matching pseudo-panel regression with type FE. Here I consider the same regression without type FE (denoted by Sym). I provide the probability limit of Sym, derive how the regression weights heterogeneous treatment effects, and compare the two estimators in the three Korean surveys.

Sym uses the same cells and weights as in Section~\ref{sec:symmetric}:
\begin{equation}\label{eq:anchored-sym-reg}
\bar Y^{B,g}_r=\sum_{k\ne\tau^\star}\beta^{S,g}_k\mathbf1\{\tau_r=k\}
+\phi^g_{a_r}+\psi^g_{t_r}+u^g_r,\qquad \tau^\star=-2.
\end{equation}
Without type FE, the linear dependence among type, year and event-time effects does not arise, so a single base period suffices.

\begin{prop}[Oracle Projection of Sym]\label{prop:sym-anchored-oracle}
Let Assumptions~\ref{assu:NA}--\ref{assu:CI} and \ref{assu:sym-repr}--\ref{assu:sym-reg} hold, and let $\boldsymbol\mu^{S,g,*}$ collect the untreated pre-birth and treated post-birth population means of the included rows. Then the coefficients of \eqref{eq:anchored-sym-reg} converge to the unique WLS projection of $\boldsymbol\mu^{S,g,*}$ on its event-time, age and year regressors, with limiting row weights $p^g_{\theta(r)}$.
\end{prop}

\begin{proof}
For each included pre-birth cell, Assumption~\ref{assu:NA} equates the observed and untreated donor means, and Assumption~\ref{assu:CI} equates this mean to that of the type's parents. For each post-birth row, Assumption~\ref{assu:sym-repr} gives the treated population mean. The joint law of large numbers and the fixed finite set of cells in Assumption~\ref{assu:sym-reg} therefore give $\bar{\mathbf Y}^{B,g}\to_p\boldsymbol\mu^{S,g,*}$. Let $D_g$ be the matrix of linearly independent regressors of \eqref{eq:anchored-sym-reg}, and define $\Pi_{N,g}=\operatorname{diag}(N^P_{\theta(r),g,0}/N)$ and $\Pi_g=\operatorname{diag}(p^g_{\theta(r)})$. Then
$\widehat{\mathbf b}^{S,g}=(D_g'\Pi_{N,g}D_g)^{-1}D_g'\Pi_{N,g}\bar{\mathbf Y}^{B,g}
\to_p(D_g'\Pi_gD_g)^{-1}D_g'\Pi_g\boldsymbol\mu^{S,g,*}$.
The inverse is continuous by the rank condition of Assumption~\ref{assu:sym-reg}.
\end{proof}

The proposition gives a probability limit, not a causal interpretation. The limit is the best additive approximation to the cell means, whether or not these means are additive in event time, age and year. It equals the causal effects when two conditions hold. The untreated means are the sum of age and year effects common to all types, with no type-specific levels, and the treatment effect at each event time is the same across types. The first condition is stronger than the one for Sym+FE, which allows type-specific levels.

In addition, the weights hold the mix of types fixed only when every type has a cell at every event time. A type without a cell at some event time drops out of the regression there, so the mix of types still changes over event time.

Proposition~\ref{prop:cell-clt} applies to Sym and Sym+FE with their respective designs. The row weights $N^P_{\theta,g,0}$ are counts contained in $\hat\eta$, so the analytical variance of Appendix~\ref{app:analytical-implementation} accounts for their sampling variation, as well as for donor reuse, repeated observations of the same person and the dependence between women and men in the same household.

The projection is not an average of the type-specific effects with the weights of the parents observed at birth when these effects differ. Write $\delta^g_\theta(\tau)=\E[Y_{F_\theta+\tau}(F_\theta)-Y_{F_\theta+\tau}(\infty)\mid\theta,g]$ for the effect of type $\theta$ and suppose that the untreated means are additive,
\begin{equation}\label{eq:additive-proj}
\E[Y_{F_\theta+\tau}(\infty)\mid\theta,g]=\alpha^{0,g}_\theta+
\phi^{0,g}_{\tilde a_\theta+\tau}+\psi^{0,g}_{F_\theta+\tau}.
\end{equation}
In the Sym+FE regression, let $C$ span the type, age and year effects, let $E$ contain the free event-time indicators, let $W$ be the diagonal matrix of row weights and let $M_C$ be the weighted residual-maker for $C$. The population event-time coefficients are
\[
\beta^{g,*}=(E'WM_CE)^{-1}E'WM_Cd^g,
\qquad d^g_{\theta,\tau}=\mathbf{1}\{\tau\ge0\}\delta^g_\theta(\tau).
\]
They equal the treatment effect at each event time when it is the same across types. Otherwise the constant row weights do not make them a convex average of the effects at each event time. For example, take equally weighted cohorts $F=2010,\ldots,2014$, birth ages $24,\ldots,30$, and event times $-5,\ldots,5$. Set all untreated means to zero and all post-birth effects to one for $F=2010$ and zero for the other cohorts. With both $\beta_{-3}$ and $\beta_{-2}$ fixed to zero, the regression with type FE gives $\beta_{-5}=0.169$, $\beta_0=0.305$, and $\beta_5=1.028$, whereas the effects weighted by the parents observed at birth are zero before birth and $0.2$ afterward. Two valid pre-birth restrictions therefore do not remove the contamination from heterogeneous effects. For Sym, the same projection formula applies when untreated means contain only age and year effects, using its own event-time indicators and $C$ containing only the linearly independent age and year indicators.

The simulation in Section~\ref{sec:sim-type-fe} illustrates this distortion using individual fertility and outcome histories.

\begin{table}[htbp]
\centering
\caption{Symmetric matching with and without type fixed effects}\label{tab:bidirectional-comparison}

\footnotesize\setlength{\tabcolsep}{3pt}\renewcommand{\arraystretch}{1.08}
\begin{tabular}{@{}lrrrrrr@{}}
\toprule
 & \multicolumn{3}{c}{Coarse cells} & \multicolumn{3}{c}{Fine cells} \\
\cmidrule(lr){2-4}\cmidrule(lr){5-7}
Sample & Sym & Sym+FE & FE SD & Sym & Sym+FE & FE SD \\
\midrule
KLIPS employment & 41.1\% & 41.1\% & 10.1 & 43.3\% & 37.6\% & 14.2 \\
KLIPS earnings & 63.1\% & 65.4\% & 15.0 & 66.1\% & 67.7\% & 24.1 \\
HIES employment & 36.0\% & 35.1\% & 8.7 & 37.0\% & 40.2\% & 9.3 \\
HIES earnings & 50.5\% & 59.8\% & 5.2 & 50.4\% & 55.5\% & 6.7 \\
HFLC employment & 33.8\% & 37.5\% & 6.1 & 34.7\% & 43.2\% & 8.7 \\
\bottomrule
\end{tabular}
\begin{minipage}{\linewidth}\vspace{0.3em}\footnotesize\textit{Notes:} Scaled penalties, averages over years 0--10. Both estimators use the same cells and weights. FE SD is the bootstrap standard deviation of Sym+FE in percentage points. SDs use 200 draws.\end{minipage}
\end{table}

Table~\ref{tab:bidirectional-comparison} reports the Sym and Sym+FE penalties in the Korean surveys on the same included cells. With coarse cells, Sym differs from pseudo-OLS by about 3 percentage points at most. Sym+FE differs from Sym by up to 4 points for employment and is 2 to 9 points higher for earnings. Moving from coarse to fine cells changes Sym by up to 3.0 percentage points and Sym+FE by up to 5.7 points. The bootstrap standard deviation of Sym+FE in KLIPS also increases with fine cells. With coarse cells, fixing the coefficients at $-5$ and $-2$ instead of $-3$ and $-2$ reduces the KLIPS employment and earnings analytical SEs from 6.90 and 11.91 to 3.99 and 6.74 percentage points, respectively, with the cells and weights unchanged. 

\begin{table}[htbp]
\centering
\caption{Share of parents observed at birth represented at each event time}\label{tab:anchor-support-restored}
\footnotesize
\begin{tabular}{llrrrrr}
\toprule
Survey & Gender & $-5$ & $-2$ & $0$ & $5$ & $10$ \\
\midrule
KLIPS & Women & 85.3\% & 95.3\% & 100.0\% & 68.6\% & 30.9\% \\
 & Men & 82.9\% & 94.8\% & 100.0\% & 68.6\% & 32.3\% \\
\addlinespace
HIES & Women & 62.5\% & 84.8\% & 100.0\% & 79.5\% & 60.3\% \\
 & Men & 60.7\% & 84.1\% & 100.0\% & 81.2\% & 62.0\% \\
\addlinespace
HFLC & Women & 41.8\% & 75.7\% & 100.0\% & 81.3\% & 45.3\% \\
 & Men & 41.1\% & 75.4\% & 100.0\% & 81.6\% & 44.3\% \\
\bottomrule
\end{tabular}
\begin{minipage}{\linewidth}\vspace{0.3em}\footnotesize\textit{Notes:} Employment samples with coarse cells. Each entry is the share of the parents observed at birth whose type has a cell at that event time.\end{minipage}
\end{table}

Table~\ref{tab:anchor-support-restored} shows why constant weights do not hold the mix of types fixed. The types with a cell ten years after birth cover only 31--32\% of the parents observed at birth in KLIPS, 60--62\% in HIES and 44--45\% in HFLC.

\section{Additional Simulations}\label{app:additional-sims}

\subsection{Analytical Coverage for All Estimators}\label{app:sim-coverage-all}

Section~\ref{sec:sim-coverage} reports the analytical standard errors of the unscaled penalty and their coverage for pseudo-OLS and NP. This experiment extends the analysis to DR and Sym+FE and to the scaled penalty, whose standard error includes the uncertainty in the counterfactual denominators. The design is the same. I generate $N=10{,}000$ complete individual histories from the equations of Section~\ref{sec:sim-design}, without selection or with symmetric selection. I also include the two proxies of that section so that DR has covariates to adjust for.

Table~\ref{tab:coverage-all-analytical} reports the Monte Carlo standard deviation, the mean analytical standard error (SE) and the coverage of 95\% confidence intervals for the unscaled and the scaled penalty. The analytical SEs match the Monte Carlo standard deviations, and coverage is between 94\% and 97\% with one exception. The exception is the scaled NP penalty under symmetric selection, whose coverage is 86\%. This shortfall reflects a bias in the estimate, not an error in the SE. Symmetric selection cancels out in the unscaled penalty, but it biases each gender's counterfactual level and hence the scaled penalty, by 0.4 to 0.8 percentage points for every estimator. The SE of NP is only 0.98 percentage points, against about 4 points for the other estimators, so the same bias lowers its coverage.

\begin{table}[htbp]
\centering
\caption{Coverage of confidence intervals based on analytical standard errors}\label{tab:coverage-all-analytical}

\footnotesize\setlength{\tabcolsep}{4pt}\renewcommand{\arraystretch}{1.08}
\begin{tabular}{@{}llrrrrrr@{}}
\toprule
 &  & \multicolumn{3}{c}{Unscaled average} & \multicolumn{3}{c}{Scaled average (\%)} \\
\cmidrule(lr){3-5}\cmidrule(lr){6-8}
Scenario & Estimator & SD & Mean SE & Cov. (\%) & SD & Mean SE & Cov. (\%) \\
\midrule
No selection & Pseudo-OLS & 0.112 & 0.113 & 95.2 & 3.62 & 3.65 & 94.8 \\
 & DR & 0.130 & 0.132 & 95.4 & 4.22 & 4.28 & 96.2 \\
 & Sym+FE & 0.106 & 0.105 & 95.0 & 3.42 & 3.39 & 94.2 \\
 & NP & 0.026 & 0.027 & 96.4 & 0.89 & 0.90 & 96.8 \\
\addlinespace
Symmetric selection & Pseudo-OLS & 0.106 & 0.111 & 95.4 & 4.06 & 4.29 & 94.6 \\
 & DR & 0.123 & 0.130 & 95.8 & 4.40 & 4.66 & 96.0 \\
 & Sym+FE & 0.100 & 0.104 & 96.6 & 3.81 & 3.97 & 95.8 \\
 & NP & 0.026 & 0.026 & 96.0 & 0.98 & 0.98 & 86.4 \\
\bottomrule
\end{tabular}
\begin{minipage}{\linewidth}\vspace{0.35em}\footnotesize\textit{Notes:} 500 replications with $N=10{,}000$. SD is the standard deviation of the estimates across replications, Mean SE is the average analytical standard error, and Cov.\ is the share of replications in which the 95\% confidence interval contains the true value. Both penalties are averages over years 0--10. The true unscaled penalty is 0.28. The true scaled penalty is 9.5\% without selection and 9.3\% under symmetric selection, and 9.7\% and 9.4\% for NP, which uses more cohorts. Scaled target values are approximated by simulation.\end{minipage}
\end{table}

\FloatBarrier

\FloatBarrier
\section{Marriage and First Birth}\label{app:marriage}

Section~\ref{sec:empirical-results} shows that employment in the panel moves sharply before the first birth. To understand why, I examine how marriage and the first birth are linked. The year of marriage is the year of the first interview in which a respondent who was previously never married reports being married. It is therefore unknown for respondents who are already married when they enter the survey. For the marriage-event profiles in Figure~\ref{fig:marriage-profiles}A and Table~\ref{tab:marriage-gender}, I keep marriages from 2011 onward at ages 20--45 and use the event window $[-5,10]$ with base period $-2$.

Marriage and the first birth are close in time. Of the 1,697 parents in the panel sample of Table~\ref{tab:klips}, the year of marriage is observed for 1,094. Most of the others were already married when they were first interviewed. Among these 1,094 parents, 28.9\% have their first child in the year of marriage and another 27.3\% in the following year, while 14.7\% have it before the recorded year of marriage.\footnote{The year of marriage is the year of the first interview that reports it, so a marriage that takes place between two interviews is dated late.} Among the parents observed at birth in the employment sample, 98.8\% of women and 97.1\% of men are married, so matching on marital status compares married parents with married childless donors.

Figure~\ref{fig:marriage-profiles}A plots employment by year relative to marriage, for men and for women grouped by the timing of their first birth. Men's employment rises from 92.0\% in the year before marriage to 95.4\% in the year of marriage. Women's employment falls from 86.4\% to 60.3\% on average, but the fall is concentrated among women whose first birth is close. It falls from 81\% to 40\% for women who give birth in the year of marriage or before. It falls from 90\% to 64\% for women who give birth a year later, many of whom are pregnant in the year of marriage. For women who give birth two or more years after marriage, employment falls from 87\% to 81\% in the year of marriage, and the larger fall comes around the birth. For women with no observed birth it falls from 89\% to 76\%. These profiles therefore do not separate the effects of marriage, pregnancy and the first birth.

\begin{figure}[!htbp]
\centering
\includegraphics[width=\linewidth]{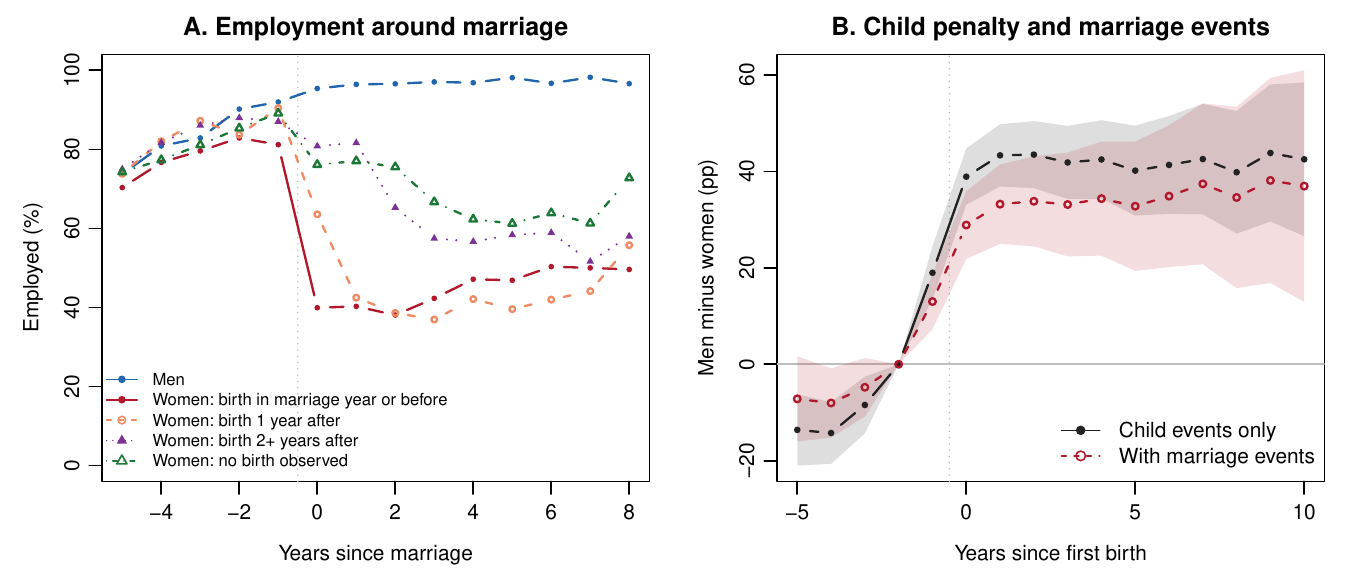}
\caption{Employment around marriage and the child penalty with marriage events. Panel A groups women by the year of their first birth relative to the year of marriage. Panel B plots the men-minus-women child event coefficients from the panel regression with and without indicators for the years relative to marriage. The shaded areas show pointwise 95\% confidence intervals.}\label{fig:marriage-profiles}
\end{figure}

Figure~\ref{fig:marriage-profiles}B shows how the child event study in the KLIPS panel changes when indicators for the years relative to marriage are added to the regression. The regression uses the 1,094 parents whose year of marriage is observed, because the marriage indicators cannot be defined for the others. For this reason the penalty without marriage indicators, 41.8 percentage points (SE 4.3), differs from the 38.9 points in Table~\ref{tab:klips}. Adding the marriage indicators lowers the average post-birth penalty to 34.4 points (SE 6.5) and the coefficient in the year before birth from 19.0 to 13.0 points. For many parents, marriage, pregnancy and the first birth occur close together. Marriage indicators can therefore capture changes associated with pregnancy and childbirth as well as marriage.

I also use marriage as the event to examine whether the agreement between the panel and pseudo-panel extends beyond the childbirth context. In the child penalty analysis, parents are matched to childless donors with the same marital status. For marriage, donors are respondents who are still unmarried, some of whom marry later, and marital status can no longer be used for matching. Pseudo-OLS matches on gender, age, year and education, and DR adjusts for the covariates of Section~\ref{sec:empirical-data}, except the spouse age gap, within cells of gender, age and year. Table~\ref{tab:marriage-gender} reports the scaled effects averaged over years 0--10. Here the pseudo-panel estimates a larger gender difference than the panel. The difference in employment is 54.5\% (SE 3.20) in the pseudo-panel and 43.1\% (SE 3.18) in the panel. Relative to the panel, the pseudo-panel shifts the estimated effects upward for both genders, by 18.5 percentage points for men and 7.1 for women, so the shifts do not cancel in the gender difference. The two designs differ in the counterfactual. In the year of marriage, the estimated employment rate of men had they not married is 79.8\% in the pseudo-panel and 94.1\% in the panel. Unmarried donors work less than men who are about to marry. This pattern is consistent with the selection bias described in Section~\ref{sec:framework}. For earnings, the gap is 86.7\% in the pseudo-panel and 54.9\% in the panel. DR makes no meaningful adjustment in this case. 

\begin{table}[htbp]
\centering
\caption{Employment and earnings after marriage in the panel and the pseudo-panel}\label{tab:marriage-gender}
\footnotesize
\begin{tabular}{@{}lrrrrrr@{}}
\toprule
 & \multicolumn{3}{c}{Employment} & \multicolumn{3}{c}{Annual earnings} \\
\cmidrule(lr){2-4}\cmidrule(lr){5-7}
 & Panel & Pseudo-OLS & DR & Panel & Pseudo-OLS & DR \\
\midrule
Men & 1.6\% & 20.1\% & 21.1\% & 17.1\% & 60.2\% & 60.1\% \\
Women & -41.5\% & -34.4\% & -33.3\% & -37.8\% & -26.5\% & -25.7\% \\
Men $-$ women & 43.1\% & 54.5\% & 54.3\% & 54.9\% & 86.7\% & 85.8\% \\
 & (3.18) & (3.20) & (3.26) & (6.50) & (8.27) & (8.13) \\
\bottomrule
\end{tabular}
\begin{minipage}{\linewidth}\vspace{0.3em}\footnotesize\textit{Notes:} KLIPS, marriages from 2011, window $[-5,10]$ with base period $-2$. Entries are scaled changes in the outcome after marriage, averaged over event times 0--10, and may also reflect changes associated with subsequent childbirth. Analytical SEs of the difference, in percentage points, are in parentheses. Employment sample: 655 women and 693 men in the panel, 744 and 767 observed in the marriage year in the pseudo-panel.\end{minipage}
\end{table}

\FloatBarrier
\section{Additional Robustness Analyses}\label{app:empirical-diagnostics}

\subsection{Donor Pool}\label{app:future-never}\label{app:diagnostics}

\begin{figure}[!b]
\centering
\includegraphics[width=\linewidth]{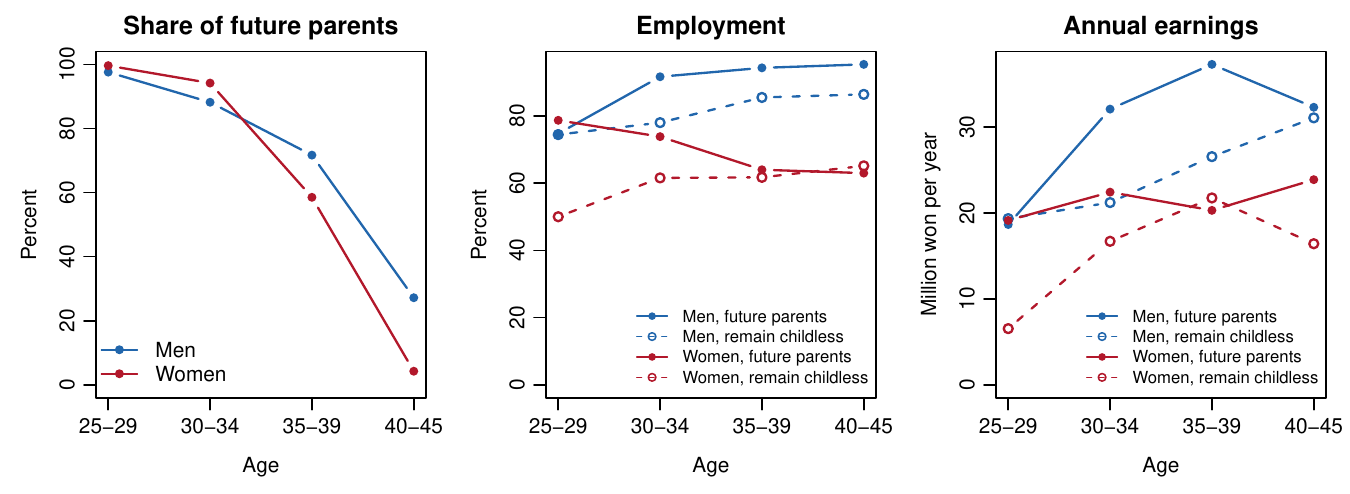}
\caption{Future parents and never-parents among childless donors in KLIPS. Future parents have a birth observed in a later wave, and never-parents remain childless through at least age 45.}\label{fig:future-profiles}
\end{figure}

In KLIPS, I identify future parents among donors who have a birth observed in a later wave and never-parents who remain childless in every wave through at least age 45. The remaining donors, 59\% of the donor records at ages 25--45, are observed neither with a later birth nor up to age 45 and are left out of this comparison. Figure~\ref{fig:future-profiles} compares the two classified groups by age. The left panel plots the share of future parents among them. Below age 35 the share is at least 88\%, and at ages 40--45 it falls to 27\% for men and 4\% for women. The other two panels plot the employment rates and annual earnings of the two groups. At ages 30--45 future fathers are 9 to 14 percentage points more likely to work than never-parents, and at ages 30--39 they earn about 11 million won more per year. Below age 35, future mothers are 12 to 29 percentage points more likely to work than women who remain childless. At later ages the two groups of women work at similar rates.

To see whether these gaps remain within matching cells, I regress each outcome, separately by gender, on an indicator for future parents and a fixed effect for every combination of age, year, marital status, education and urban residence.\footnote{Standard errors are clustered as in Section~\ref{sec:empirical-data}.} For men the gaps remain. Future fathers are 8.7 percentage points (SE 4.2) more likely to work and earn 5.0 million won (SE 2.2 million) more per year than never-parents with the same characteristics. For women the differences are insignificant, $-0.8$ points (SE 11.9) for employment and 2.5 million won (SE 4.9 million) for earnings.

\FloatBarrier
Table~\ref{tab:donor-composition}\phantomsection\label{app:donor-composition} re-estimates the KLIPS penalties using fine cells for both estimators and future parents as the only donors, which excludes both the never-parents and the unclassified donors. The two columns being compared use the same parents and the same pre-birth event times, so only the donor pool changes. The pseudo-OLS employment penalty falls from 42.1\% to 41.8\%, while the earnings penalty falls from 65.8\% to 65.7\%. The restriction slightly reduces men's estimated employment gain, while women's estimated loss is nearly unchanged. For earnings, women's larger estimated loss is offset by men's smaller estimated gain.

\begin{table}[htbp]
\centering
\caption{Restricting donors to future parents}\label{tab:donor-composition}

\footnotesize\setlength{\tabcolsep}{5pt}\renewcommand{\arraystretch}{1.08}
\begin{tabular}{@{}lrrrrrr@{}}
\toprule
 & \multicolumn{3}{c}{Employment} & \multicolumn{3}{c}{Annual earnings} \\
\cmidrule(lr){2-4}\cmidrule(lr){5-7}
Estimator & \shortstack{Full\\sample} & \shortstack{All\\donors} & \shortstack{Future\\parents only} & \shortstack{Full\\sample} & \shortstack{All\\donors} & \shortstack{Future\\parents only} \\
\midrule
Pseudo-OLS & 42.1\% & 42.1\% & 41.8\% & 64.7\% & 65.8\% & 65.7\% \\
DR & 41.6\% & 41.5\% & 40.9\% & 65.2\% & 66.1\% & 65.4\% \\
\bottomrule
\end{tabular}
\begin{minipage}{\linewidth}\vspace{0.3em}\footnotesize\textit{Notes:} Scaled KLIPS penalties, averages over years 0--10, using fine cells for both estimators. The last two columns of each outcome use the same pre-birth cells and all original post-birth parents.\end{minipage}
\end{table}

\FloatBarrier
\subsection{Matching Variables and the DR Adjustment}\label{app:coarse-cell}

Table~\ref{tab:xs-decomposition} reports the KLIPS employment penalty under different sets of matching variables. Matching on age and year alone gives an unscaled penalty of 0.422, and adding marital status lowers it to 0.291. Adding urban residence, education or both changes it by at most 0.03. With both added, the number of pre-birth outcomes that can be imputed falls by 15\%. DR instead uses these variables as covariates in its two models within the coarser cells, so no pre-birth outcome is lost.

\begin{table}[htbp]
\centering
\caption{Matching variables and the employment penalty}\label{tab:xs-decomposition}
\small
\begin{tabular}{lrr}
\toprule
Matching cells & Unscaled penalty & Imputed pre-birth outcomes \\
\midrule
Age and year & 0.422 & 9,693 \\
+ marital status & 0.291 & 8,851 \\
+ marital status, urban & 0.299 & 8,373 \\
+ marital status, education & 0.315 & 8,278 \\
+ marital status, urban, education & 0.316 & 7,557 \\
DR within marital cells & 0.296 & --- \\
DR within marital and urban cells & 0.300 & --- \\
\bottomrule
\end{tabular}
\begin{minipage}{\linewidth}\vspace{0.3em}\footnotesize\textit{Notes:} Unscaled KLIPS employment penalty, average over years 0--10. Every row also matches on gender. The last column is the number of pre-birth outcomes, one per parent and event time, whose cell contains a donor.\end{minipage}
\end{table}

With coarse cells, the DR adjustment raises imputed pre-birth employment by an average of 1.6 percentage points for women and 0.4 for men. For annual earnings it is 650,000 won for women and 200,000 won for men. The adjustment comes almost entirely from the outcome model. The residual correction offsets 4 to 12\% of the outcome-model term, except for men's earnings, where it has the same sign. Table~\ref{tab:dr-overlap} reports how many donors a cell contains. The median is 8 to 30 donors in coarse cells and 4 to 11 in fine cells. In fine cells a single donor often receives the entire weight, which is why I fit DR within coarse cells.

\begin{table}[htbp]
\centering
\caption{Donors per matching cell}\label{tab:dr-overlap}

\footnotesize\setlength{\tabcolsep}{4pt}\renewcommand{\arraystretch}{1.08}
\begin{tabular}{@{}lrrrrrr@{}}
\toprule
 & \multicolumn{3}{c}{Coarse cells} & \multicolumn{3}{c}{Fine cells} \\
\cmidrule(lr){2-4}\cmidrule(lr){5-7}
Sample & Median $n$ & Median ESS & \shortstack{90th max.\\weight} & Median $n$ & Median ESS & \shortstack{90th max.\\weight} \\
\midrule
KLIPS employment & 8.0 & 7.2 & 73.5\% & 4.0 & 3.9 & 100.0\% \\
KLIPS earnings & 8.0 & 7.3 & 67.9\% & 4.0 & 3.9 & 100.0\% \\
HIES employment & 13.0 & 12.7 & 50.0\% & 5.0 & 5.0 & 100.0\% \\
HIES earnings & 30.0 & 28.5 & 23.9\% & 11.0 & 11.0 & 50.0\% \\
HFLC employment & 24.0 & 17.8 & 54.0\% & 8.0 & 5.9 & 98.0\% \\
\bottomrule
\end{tabular}
\begin{minipage}{\linewidth}\vspace{0.3em}\footnotesize\textit{Notes:} Donor cells used to impute pre-birth outcomes. $n$ is the number of donors in a cell and ESS is their effective number under the DR weights, $(\sum_j w_j)^2/\sum_j w_j^2$. The last column in each block is the 90th percentile across cells of the largest weight share of a single donor.\end{minipage}
\end{table}

I further investigate the finite-sample stability of the DR adjustment using cross-fitting. Table~\ref{tab:crossfit} compares the DR estimate in KLIPS, in which both models are estimated on the full sample, with two-fold and five-fold cross-fitting, in which each model is estimated on the other folds. Folds are assigned to clusters of connected persons and households, so all records of a cluster, in every year and as parent or donor, fall in the same fold. A pre-birth cell is retained only if it has donors in every fold; all post-birth parent observations remain. This requirement keeps 87--88\% of the imputed pre-birth outcomes, weighted by parents, with two folds but only 41--42\% with five folds. The cross-fitted estimate therefore differs from the full-sample estimate for two reasons, the smaller set of cells and the cross-fitting itself. To separate them, the third row of the table re-estimates DR without cross-fitting on the cells that cross-fitting keeps. On these cells, two-fold cross-fitting lowers the employment penalty from 41.7\% to 41.3\% and the earnings penalty from 65.7\% to 61.1\%.

\begin{table}[htbp]
\centering
\caption{KLIPS DR with and without cross-fitting}\label{tab:crossfit}

\footnotesize\setlength{\tabcolsep}{4pt}\renewcommand{\arraystretch}{1.08}
\begin{tabular}{@{}lrrrrrr@{}}
\toprule
 & \multicolumn{3}{c}{Employment} & \multicolumn{3}{c}{Annual earnings} \\
\cmidrule(lr){2-4}\cmidrule(lr){5-7}
 & No CF & $K=2$ & $K=5$ & No CF & $K=2$ & $K=5$ \\
\midrule
Scaled penalty & 40.5\% & 41.3\% & 43.1\% & 64.1\% & 61.1\% & 65.1\% \\
Unscaled penalty & 0.296 & 0.303 & 0.320 & 1,834 & 1,703 & 1,886 \\
\addlinespace
No CF on the same cells, scaled & --- & 41.7\% & 42.2\% & --- & 65.7\% & 65.0\% \\
\addlinespace
Imputed pre-birth outcomes kept & 100\% & 86.6\% & 41.1\% & 100\% & 88.2\% & 42.4\% \\
\bottomrule
\end{tabular}
\begin{minipage}{\linewidth}\vspace{0.3em}\footnotesize\textit{Notes:} DR with coarse cells on the window $[-5,10]$, 0--10 averages. No CF estimates the outcome and propensity models on the full sample. $K$ folds are assigned to connected person--household clusters. Cross-fitting retains pre-birth cells with donors in every fold; all post-birth parent observations remain. Retention shares weight pre-birth outcomes by parents. Unscaled employment penalties are differences in employment rates, and unscaled earnings penalties are in 10,000 won at 2020 prices.\end{minipage}
\end{table}

Table~\ref{tab:rcs-crossfit} reports the DR penalties in HIES and HFLC with and without cross-fitting, for coarse and fine cells. Cross-fitting uses two folds to minimize the loss of included cells. In HFLC the folds are assigned to households. In HIES they are assigned to household--years for employment and to household--quarters for earnings. As in KLIPS, the table also reports the estimate without cross-fitting on the cells that cross-fitting keeps. On these cells, cross-fitting changes the penalties by at most 1.5 percentage points.

\begin{table}[htbp]
\centering
\caption{Repeated-cross-section DR and two-fold cross-fitting}\label{tab:rcs-crossfit}

\footnotesize\setlength{\tabcolsep}{4pt}\renewcommand{\arraystretch}{1.08}
\begin{tabular}{@{}lrrrrrr@{}}
\toprule
 & \multicolumn{3}{c}{Coarse cells} & \multicolumn{3}{c}{Fine cells} \\
\cmidrule(lr){2-4}\cmidrule(lr){5-7}
Sample & \shortstack{Full\\sample} & \shortstack{Same\\cells} & \shortstack{Cross-\\fitting} & \shortstack{Full\\sample} & \shortstack{Same\\cells} & \shortstack{Cross-\\fitting} \\
\midrule
HIES employment & 34.3\% & 34.3\% & 35.0\% & 33.4\% & 33.4\% & 33.3\% \\
HIES earnings & 51.1\% & 50.9\% & 51.1\% & 49.6\% & 49.6\% & 49.3\% \\
HFLC employment & 32.3\% & 32.6\% & 32.2\% & 32.0\% & 31.9\% & 30.4\% \\
\bottomrule
\end{tabular}
\begin{minipage}{\linewidth}\vspace{0.3em}\footnotesize\textit{Notes:} DR on the window $[-5,10]$, scaled 0--10 averages. Full sample estimates the outcome and propensity models on all records. Two-fold cross-fitting retains pre-birth cells with donors in both folds; all post-birth parent observations remain. Same cells estimates DR without cross-fitting on that support. Folds are assigned to households in HFLC and to household--years (employment) or household--quarters (earnings) in HIES.\end{minipage}
\end{table}

\FloatBarrier
\subsection{A Shorter Estimation Window: \texorpdfstring{$[-2,5]$}{[-2,5]}}\label{app:tau05}
\suppressfloats[t]

I re-estimate every estimator except Sym+FE on the window $[-2,5]$, with base period $-2$, and average the penalties over years 0--5. The survey years, donor ages, parent eligibility rules and KLIPS first-birth cohorts are the same as in the main analysis, and the panel regression keeps the parents with an outcome before and after birth within this window. Sym+FE is omitted because the window does not contain year $-3$, which Sym+FE fixes at zero. Tables~\ref{tab:klips-tau05} and \ref{tab:rcs-tau05} report the penalties, and Figure~\ref{fig:tau05} plots the event-time paths. The employment penalties are larger than the ten-year averages, consistent with larger employment losses in the first years after birth. In KLIPS the employment penalty is 51.3\% in the panel and 45.3\% in the standard pseudo-panel regression.

\begin{table}[htbp]
\centering
\caption{KLIPS child penalties: estimation window $[-2,5]$}\label{tab:klips-tau05}

\footnotesize\setlength{\tabcolsep}{4pt}\renewcommand{\arraystretch}{1.08}
\begin{tabular}{@{}lrrrr@{}}
\toprule
 & \multicolumn{2}{c}{Employment} & \multicolumn{2}{c}{Annual earnings} \\
\cmidrule(lr){2-3}\cmidrule(lr){4-5}
Estimator & Unscaled (SE) & Scaled (SE) & Unscaled (SE) & Scaled (SE) \\
\midrule
\multicolumn{5}{l}{\emph{Actual-panel estimators, first births from 2011}} \\
Panel OLS & 0.429 (0.030) & 51.3\% (2.49) & 1,851 (134) & 65.3\% (3.65) \\
Balanced [-5,5] & 0.438 (0.060) & 52.3\% (4.52) & 1,810 (323) & 67.7\% (9.94) \\
Balanced [-2,5] & 0.426 (0.047) & 52.0\% (3.65) & 1,952 (236) & 70.5\% (6.77) \\
\addlinespace
\multicolumn{5}{l}{\emph{Pseudo-panel estimators}} \\
Pseudo-OLS (fine) & 0.331 (0.023) & 45.3\% (2.28) & 1,636 (126) & 62.2\% (3.46) \\
Pseudo-OLS (coarse) & 0.313 (0.023) & 43.8\% (2.34) & 1,542 (133) & 60.0\% (3.68) \\
DR (coarse) & 0.320 (0.023) & 44.4\% (2.27) & 1,614 (124) & 61.6\% (3.41) \\
Sym (coarse) & 0.317 (0.024) & 44.4\% (2.36) & 1,583 (140) & 61.0\% (3.95) \\
NP (fine) & 0.310 (0.021) & 43.6\% (2.27) & 1,443 (138) & 57.2\% (3.77) \\
NP (coarse) & 0.283 (0.022) & 41.2\% (2.46) & 1,296 (150) & 53.6\% (4.11) \\
\bottomrule
\end{tabular}
\begin{minipage}{\linewidth}\vspace{0.3em}\footnotesize\textit{Notes:} Men-minus-women averages over years 0--5. Unscaled employment penalties are differences in employment rates, and unscaled earnings penalties are in 10,000 won at 2020 prices. Analytical standard errors of the scaled penalties are in percentage points. Balanced $[a,b]$ keeps parents observed at or before $a$ and at or after $b$. The three panel rows use 820/791, 252/244 and 411/410 women/men for employment and 905/875, 282/264 and 459/449 for earnings.\end{minipage}
\end{table}

\begin{table}[htbp]
\centering
\caption{Repeated-cross-section child penalties: $[-2,5]$}\label{tab:rcs-tau05}
\footnotesize
\begin{tabular}{lcccccc}
\toprule
 & \multicolumn{2}{c}{HFLC employment} & \multicolumn{2}{c}{HIES employment} & \multicolumn{2}{c}{HIES earnings} \\
\cmidrule(lr){2-3}\cmidrule(lr){4-5}\cmidrule(lr){6-7}
Estimator & Scaled & SE & Scaled & SE & Scaled & SE \\
\midrule
Pseudo-OLS (fine) & 37.6\% & 2.1 & 40.8\% & 1.4 & 53.2\% & 1.4 \\
Pseudo-OLS (coarse) & 36.4\% & 2.0 & 40.8\% & 1.4 & 55.0\% & 1.3 \\
DR (coarse) & 37.7\% & 2.0 & 41.8\% & 1.4 & 54.8\% & 1.3 \\
Sym (coarse) & 38.9\% & 1.6 & 42.7\% & 1.2 & 53.0\% & 1.2 \\
NP (fine) & 35.4\% & 1.7 & 40.2\% & 1.2 & 54.8\% & 1.4 \\
NP (coarse) & 34.8\% & 1.7 & 39.1\% & 1.2 & 55.3\% & 1.3 \\
\bottomrule
\end{tabular}
\begin{minipage}{\linewidth}\vspace{0.3em}\footnotesize\textit{Notes:} Scaled men-minus-women averages over years 0--5, with analytical standard errors in percentage points.\end{minipage}
\end{table}

\begin{figure}[p!]
\centering
\includegraphics[width=\linewidth,height=.83\textheight,keepaspectratio]{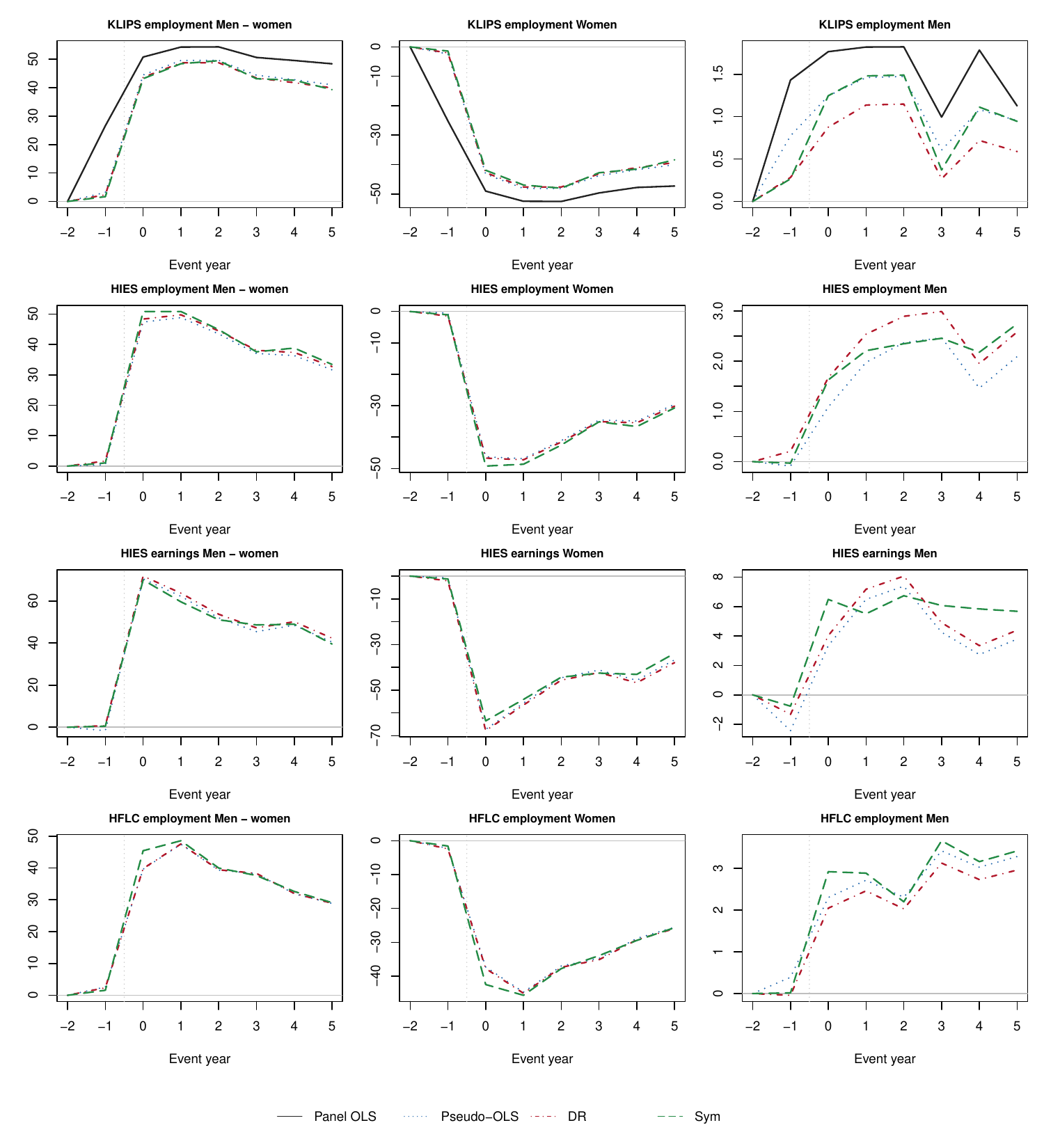}
\caption{Scaled event-time paths estimated on the window $[-2,5]$, with base period $-2$. Pseudo-OLS uses fine cells, and DR and Sym coarse cells.}\label{fig:tau05}
\end{figure}

\subsection{Survey Weights}\label{app:weighted-rcs}

Using the survey weights changes the HIES and HFLC pseudo-OLS penalties by at most 1.3 percentage points (Table~\ref{tab:weighted-rcs}).

\begin{table}[htbp]
\centering
\caption{Survey-weighted versus unweighted pseudo-OLS penalties}\label{tab:weighted-rcs}

\footnotesize\setlength{\tabcolsep}{4pt}\renewcommand{\arraystretch}{1.08}
\begin{tabular}{@{}lrrrrrr@{}}
\toprule
 & \multicolumn{3}{c}{Coarse cells} & \multicolumn{3}{c}{Fine cells} \\
\cmidrule(lr){2-4}\cmidrule(lr){5-7}
Sample & Unweighted & Weighted & Diff. (pp) & Unweighted & Weighted & Diff. (pp) \\
\midrule
HIES employment & 33.3\% & 34.3\% & 1.0 & 33.5\% & 33.8\% & 0.4 \\
HIES earnings & 51.1\% & 50.7\% & -0.5 & 49.7\% & 49.5\% & -0.2 \\
HFLC employment & 30.9\% & 30.3\% & -0.5 & 32.0\% & 30.6\% & -1.3 \\
\bottomrule
\end{tabular}
\begin{minipage}{\linewidth}\vspace{0.3em}\footnotesize\textit{Notes:} Scaled pseudo-OLS penalties, men-minus-women averages over years 0--10. Survey weights enter the donor means, parent counts and regressions. Diff.\ is the weighted minus the unweighted penalty in percentage points.\end{minipage}
\end{table}

\FloatBarrier
\subsection{Analytical and Bootstrap Standard Errors}\label{app:resampling-restored}

Table~\ref{tab:analytical-se} compares the analytical standard errors with the standard deviation of a cluster bootstrap on the same samples, using coarse cells. I draw 200 bootstrap samples of clusters with replacement and repeat the entire estimation procedure on each sample. The clusters are those of Section~\ref{sec:empirical-data}.

\begin{table}[!htbp]
\centering
\caption{Analytical standard errors and bootstrap standard deviations}\label{tab:analytical-se}

\footnotesize\setlength{\tabcolsep}{4pt}\renewcommand{\arraystretch}{1.0}
\begin{tabular}{@{}lrrrrrr@{}}
\toprule
 & \multicolumn{3}{c}{Employment} & \multicolumn{3}{c}{Earnings} \\
\cmidrule(lr){2-4}\cmidrule(lr){5-7}
Estimator & Penalty (\%) & Analytical SE & Bootstrap SD & Penalty (\%) & Analytical SE & Bootstrap SD \\
\midrule
\multicolumn{7}{l}{\emph{KLIPS}} \\
Pseudo-OLS & 40.1 & 2.62 & 2.91 & 62.2 & 4.46 & 4.94 \\
DR & 40.5 & 2.60 & 2.99 & 64.1 & 4.22 & 4.66 \\
Sym+FE & 41.1 & 6.90 & 10.08 & 65.4 & 11.91 & 14.98 \\
NP & 35.4 & 2.97 & 3.47 & 42.5 & 5.32 & 8.27 \\
\addlinespace
\multicolumn{7}{l}{\emph{HIES}} \\
Pseudo-OLS & 33.3 & 1.42 & 1.61 & 51.1 & 1.27 & 1.38 \\
DR & 34.3 & 1.38 & 1.54 & 51.1 & 1.27 & 1.41 \\
Sym+FE & 35.1 & 6.45 & 8.73 & 59.8 & 4.38 & 5.18 \\
NP & 32.4 & 1.27 & 1.48 & 56.0 & 1.56 & 1.56 \\
\addlinespace
\multicolumn{7}{l}{\emph{HFLC}} \\
Pseudo-OLS & 30.9 & 2.00 & 1.98 & --- & --- & --- \\
DR & 32.3 & 1.99 & 2.11 & --- & --- & --- \\
Sym+FE & 37.5 & 5.17 & 6.09 & --- & --- & --- \\
NP & 27.1 & 1.83 & 2.03 & --- & --- & --- \\
\bottomrule
\end{tabular}
\begin{minipage}{\linewidth}\vspace{0.3em}\footnotesize\textit{Notes:} Coarse cells, scaled men-minus-women averages over years 0--10. Standard errors and standard deviations are in percentage points. The bootstrap uses 200 samples of clusters.\end{minipage}
\end{table}

For KLIPS employment Sym+FE, the five most extreme bootstrap estimates account for 19\% of the bootstrap variance. KLIPS earnings NP is sensitive to the inclusion of a high-earning man who supplies two donor cells. Excluding comparisons that use those cells raises the scaled penalty by 7.7 percentage points, with the remaining cell means and counts held fixed and the weights renormalized. A cell with one donor has the same mean whenever he is sampled, but disappears when he is omitted, a change that the first-order approximation does not capture.

\FloatBarrier
\subsection{Heterogeneity}\label{app:heterogeneity}

In the KLIPS panel, employment penalties differ considerably by education and age at first birth as reported in Table~\ref{tab:child-hetero}. The penalty is 52.8\% for parents with a high school or two-year college education, much larger than 39.5\% for university graduates. Parents who had their first child at age 30 or younger also have a larger penalty than older parents, 61.9\% versus 41.2\%. These differences are smaller in the pseudo-panel. With coarse cells, pseudo-OLS estimates fall between 36.2\% and 44.6\% across all four groups. DR gives estimates of 42.9\% for the less educated and 46.9\% for younger parents, but still misses much of the difference between groups. The close aggregate estimates therefore conceal substantial discrepancies for some subpopulations.

\begin{table}[htbp]
\centering
\caption{Child penalty heterogeneity}\label{tab:child-hetero}
\footnotesize
\begin{tabular}{lrrrrr}
\toprule
Group & $n_w$ & $n_m$ & Panel & Pseudo-OLS & DR \\
\midrule
Overall & 898 & 895 & 46.9\% & 40.1\% & 40.5\% \\
High school or two-year college & 413 & 354 & 52.8\% & 43.3\% & 42.9\% \\
University or above & 475 & 537 & 39.5\% & 36.2\% & 38.1\% \\
Age at birth $\leq30$ & 340 & 174 & 61.9\% & 44.6\% & 46.9\% \\
Age at birth $>30$ & 558 & 721 & 41.2\% & 38.5\% & 39.3\% \\
\bottomrule
\end{tabular}
\begin{minipage}{\linewidth}\vspace{0.3em}\footnotesize\textit{Notes:} Scaled employment penalties, averages over years 0--10 with coarse cells. $n_w$ and $n_m$ are the women and men observed at birth in the pseudo-panel. Education groups follow the education recorded in each record, and the donor pool is not split by group.\end{minipage}
\end{table}

\end{document}